\documentclass[sigconf, nonacm]{acmart}
\usepackage[utf8]{inputenc}
\usepackage{graphicx}
\usepackage{textcomp}
\usepackage{xcolor}
\usepackage{booktabs}
\usepackage{amsmath}
\usepackage{amsfonts}
\usepackage{algorithm}
\usepackage{algpseudocode}
\usepackage{csquotes}
\usepackage{enumerate}
\usepackage{multirow}
\usepackage{multicol}
\usepackage{tabularx}
\usepackage{array}
\newtheorem{definition}{Definition}
\usepackage{xcolor}
\usepackage{subcaption}
\usepackage{flushend}
\usepackage{comment}
\usepackage{mathrsfs}
\usepackage{gensymb}
\usepackage{tikz}
\usetikzlibrary{trees}
\newtheorem{proposition}{Proposition}
\newtheorem{theorem}{Theorem}
\newtheorem{corollary}{Corollary}
\usepackage{hyperref} 
\hypersetup{
    colorlinks=true,
    linkcolor=blue,
    filecolor=magenta,      
    urlcolor=blue,
    citecolor=blue,
    }

\def\eg{\emph{e.g.}}
\def\ie{\emph{i.e.}}
\def\cf{\emph{cf.}}
\def\etal{\textit{et al.}}

\newcommand\vldbdoi{10.14778/3579075.3579086}
\newcommand\vldbpages{1126 - 1139}
\newcommand\vldbvolume{16}
\newcommand\vldbissue{5}
\newcommand\vldbyear{2023}
\newcommand\vldbauthors{\authors}
\newcommand\vldbtitle{\shorttitle} 
\newcommand\vldbavailabilityurl{URL_TO_YOUR_ARTIFACTS}
\newcommand\vldbpagestyle{empty}

\begin{document}
\sloppy

\title{On the Risks of Collecting Multidimensional Data\\Under Local Differential Privacy}

\author{Héber H. Arcolezi}
\affiliation{
  \institution{Inria and École Polytechnique (IPP)}
}
\email{heber.hwang-arcolezi@inria.fr}

\author{Sébastien Gambs}
\affiliation{
  \institution{Université du Québec à Montréal, UQAM}
}
\email{gambs.sebastien@uqam.ca}

\author{Jean-François Couchot}
\affiliation{
  \institution{Femto-ST Institute, Univ. Bourg. Franche-Comt\'e, CNRS}
}
\email{jean-francois.couchot@univ-fcomte.fr}

\author{Catuscia Palamidessi}
\affiliation{
  \institution{Inria and École Polytechnique (IPP)}
}
\email{catuscia@lix.polytechnique.fr}

\begin{abstract}
The private collection of multiple statistics from a  population is a fundamental statistical problem. One possible approach to realize this is to rely on the local model of differential privacy (LDP). 
Numerous LDP protocols have been developed for the task of frequency estimation of single and multiple attributes. These studies mainly focused on improving the utility of the algorithms to ensure the server performs the estimations accurately. 
In this paper, we investigate privacy threats (re-identification and attribute inference attacks) against LDP protocols for multidimensional data following two state-of-the-art solutions for frequency estimation of multiple attributes. 
To broaden the scope of our study, we have also experimentally assessed five widely used LDP protocols, namely, generalized randomized response, optimal local hashing, subset selection, RAPPOR and optimal unary encoding. 
Finally, we also proposed a countermeasure that improves both utility and robustness against the identified threats. 
Our contributions can help practitioners aiming to collect users' statistics privately to decide which LDP mechanism best fits their needs.
\end{abstract}

\maketitle

\pagestyle{\vldbpagestyle}
\begingroup\small\noindent\raggedright\textbf{PVLDB Reference Format:}\\
\vldbauthors. \vldbtitle. PVLDB, \vldbvolume(\vldbissue): \vldbpages, \vldbyear.\\
\href{https://doi.org/\vldbdoi}{doi:\vldbdoi}
\endgroup
\begingroup
\renewcommand\thefootnote{}\footnote{\noindent
This work is licensed under the Creative Commons BY-NC-ND 4.0 International License. Visit \url{https://creativecommons.org/licenses/by-nc-nd/4.0/} to view a copy of this license. For any use beyond those covered by this license, obtain permission by emailing \href{mailto:info@vldb.org}{info@vldb.org}. Copyright is held by the owner/author(s). Publication rights licensed to the VLDB Endowment. \\
\raggedright Proceedings of the VLDB Endowment, Vol. \vldbvolume, No. \vldbissue\ %
ISSN 2150-8097. \\
\href{https://doi.org/\vldbdoi}{doi:\vldbdoi} \\
}\addtocounter{footnote}{-1}\endgroup

\ifdefempty{\vldbavailabilityurl}{}{
\vspace{.3cm}
\begingroup\small\noindent\raggedright\textbf{PVLDB Artifact Availability:}\\
The source code, data, and/or other artifacts have been made available at \url{https://github.com/hharcolezi/risks-ldp}.
\endgroup
}

\pagenumbering{arabic}

\section{Introduction} \label{sec:introduction}

Private and public organizations regularly collect and analyze digital data about their collaborators, volunteers, clients, etc. 
However, due to the sensitive nature of this personal data, the collection of users' raw data on a centralized server should be avoided. 
The distributed version of Differential Privacy (DP)~\cite{Dwork2006,Dwork2006DP,dwork2014algorithmic}, known as Local DP (LDP)~\cite{Duchi2013,first_ldp}, aims to address such a challenge. 
Indeed, using an LDP mechanism, a user can sanitize her profile locally before transmitting it to the server, which leads to strong privacy protection even if the server used for the aggregation is malicious. 
The LDP model has a close connection with the concept of randomized response~\cite{Warner1965}, which provides ``plausible deniability'' to users' reports. 
For this reason, LDP has been already implemented in large-scale systems by Google~\cite{rappor}, Microsoft~\cite{microsoft} and Apple~\cite{apple}.

A fundamental task under LDP guarantees is frequency estimation~\cite{tianhao2017,kairouz2016discrete,rappor,microsoft,kairouz2016extremal,apple,Wang2017,Cormode2021}, in which the data collector estimates the number of users for each possible value of one attribute based on the sanitized data of the users. 
More recently, a new line of research started investigating \textit{security}~\cite{Cheu2021,wu2021poisoning,cao2021data,li2022fine} and \textit{privacy~\cite{chatzikokolakis2020bayes,Murakami2021,Gursoy2022,Gadotti2022} threats} to LDP protocols (mainly for frequency estimation), which are discussed in detail in Section~\ref{sec:rel_work}.

In this paper, we further investigate the \textbf{privacy threats for the users} when the server aims to perform \textit{frequency estimation of multiple attributes} under LDP guarantees. 
In this setting~\cite{xiao2,tianhao2017,wang2019,Arcolezi2021_rs_fd,Varma2022}, the profile of each user is characterized by $d$ attributes $\mathcal{A}=\{A_1,A_2,\ldots,A_d\}$, in which each attribute $A_j$ has a discrete domain of size $k_j=|A_j|$, for $j \in [d]$. 
There are $n$ users $\mathcal{U}=\{u_1,\ldots,u_n\}$, and each user $u_i$, for $i \in [n]$, holds a private tuple $\textbf{v}^{(i)}=\left[v^{(i)}_{1},v^{(i)}_{2},\ldots,v^{(i)}_{d}\right]$, in which $v^{(i)}_{j}$ represents the value of attribute $A_j$ in record $\textbf{v}^{(i)}$. 
Thus, for each attribute $A_j \in \mathcal{A}$, for $j \in [d]$, the aggregator's goal is to estimate a $k_j$-bins histogram.

To the best of our knowledge, for the task considered\footnote{This is a different task of joint distribution estimation under LDP guarantees~\cite{Zhang2018,Ren2018}.}, there are mainly three solutions for satisfying LDP by randomizing the user's tuple $\textbf{v}=[v_{1},v_{2}, \ldots, v_{d}]$\footnote{For simplicity, we omit the index notation $\textbf{v}^{(i)}$ and focus on one arbitrary user $u_i$.}, which are described in the following:

\begin{itemize}
\item  \textbf{Splitting (SPL).} This naïve solution directly splits the privacy budget $\epsilon$ by $d$ attributes and reports all attributes with $\frac{\epsilon}{d}$-LDP, thus incurring a high estimation error~\cite{xiao2,tianhao2017,wang2019,Arcolezi2021_allomfree}.

\item \textbf{Sampling (SMP).} Instead of splitting the privacy budget, one state-of-the-art solution allows users to randomly sample a single attribute and report it with $\epsilon$-LDP~\cite{xiao2,wang2019,Arcolezi2021_allomfree,tianhao2017}. 

\item \textbf{Random Sampling Plus Fake Data (RS+FD)~\cite{Arcolezi2021_rs_fd}.} One of the weakness of the SMP solution is that it discloses the sampled attribute, which might not be fair to all users (\eg, some users will sample age but others will sample sensitive attribute such as disease). 
The objective of the state-of-the-art RS+FD solution is precisely to enable users to ``hide'' the sampled attribute (\ie, $\epsilon$-LDP value) by also generating one uniformly random fake data for each non-sampled attribute. 
Thus, RS+FD creates \textit{uncertainty} on the server-side.
\end{itemize}

Focusing on the state-of-the-art solutions SMP and RS+FD, first, we empirically demonstrate through extensive experiments that \textbf{the SMP solution is vulnerable to re-identification attacks} when collecting users' multidimensional data several times with $\epsilon$ values commonly used by industry nowadays~\cite{linkedin,desfontaines_dp_real_world}. 
For instance, assume a user has multiple mobile applications each surveying the user with the SMP solution on different attributes. 
Another possible scenario is the situation in which the same mobile application is used on a regular basis but surveys users with different attributes.
This enables the user to sample a (possibly different) attribute each time, thus resulting in sending their sampled attribute along with their $\epsilon$-LDP report.
Nevertheless, we show that an adversary who can see every tuple containing $\langle$sampled attribute, $\epsilon$-LDP report$\rangle$ can construct a partial or complete profile of the user, which can possibly be unique (or in a small anonymity set of $k$ individuals) in the population considered. 
Therefore, once the set of $k$ individuals (referred to as top-$k$ in this paper) is characterized, one can leverage well-known attacks (\eg, homogeneity)~\cite{Samarati2001,Cohen2022,sweeney2015only,SWEENEY2002,samarati1998protecting,Machanavajjhala2006,Li2007}.

More specifically, to attack the SMP solution, our adversarial analysis focuses on the reduced \textbf{``plausible deniability''}~\cite{Warner1965,domingo2018connecting} of using the whole privacy budget $\epsilon$ to report a single attribute out of $d$ ones. 
In this setting, the adversary has a higher chance to infer the users' true value for each data collection performed.
Consequently, in multiple data collections, the adversary can build partial or even sometimes complete profiles of each user, then using it to perform a re-identification attack.
However, this depends on the LDP protocol being used as the encoding and randomization vary across them~\cite{tianhao2017,Cormode2021}. 
In our experiments, we have assessed five widely used LDP protocols for frequency estimation (\emph{a.k.a.} frequency oracle protocols~\cite{Wang2018,Wang2020_post_process}), namely Generalized Randomized Response (GRR)~\cite{kairouz2016discrete,kairouz2016extremal}, Optimal Local Hashing (OLH)~\cite{tianhao2017}, Subset Selection (SS)~\cite{wang2016mutual,Min2018} and two Unary Encoding (UE) protocols (Basic One-time RAPPOR~\cite{rappor} and Optimal UE~\cite{tianhao2017}). 
To assess the risks of re-identification we have also considered two privacy models, usual LDP and the relaxed version of LDP developed in~\cite{Murakami2021} (the latter in Appendix~\ref{appC:add_re_ident_smp}) for measuring re-identification risks.

Secondly, we observe that \textbf{since the RS+FD solution generates fake data uniformly at random in~\cite{Arcolezi2021_rs_fd,Varma2022}, it is possible to uncover the sampled attribute of users} in certain conditions. 
In this context, we evaluated the effectiveness of the RS+FD solution in hiding the sampled attribute to the aggregator by varying the privacy budget $\epsilon$, the LDP protocol and the fake data generation procedure. 
In particular, if the aggregator is able to break RS+FD into the SMP solution, the RS+FD solution might also be subject to the same vulnerability to re-identification attacks on multiple collections. 
Thus, we have proposed three attack models to uncover the sampled attribute of users using the RS+FD solution and evaluated its risks to re-identification attacks.
Lastly, as shown in our results, RS+FD is, to some extent, a natural countermeasure to re-identification attacks due to chaining errors from incorrectly predicting the sampled attribute and user's value in multiple collections.
Building on this, we have designed a stronger countermeasure that adapts RS+FD to generate fake data following non-uniform distributions, \textbf{almost fully preventing the inference of the sampled attribute while preserving utility}.

To summarize, this paper makes the following contributions:

\begin{itemize}
    \item   We investigate privacy threats against LDP protocols for multidimensional data following two state-of-the-art solutions for frequency estimation of multiple attributes, SMP~\cite{xiao2,tianhao2017,wang2019,Arcolezi2021_allomfree} and RS+FD~\cite{Arcolezi2021_rs_fd}, providing insightful adversarial analysis to help in LDP protocol selection.

    \item   We demonstrate through extensive experiments that the SMP solution is vulnerable to re-identification attacks due to the disclosure of the sampled attribute and lower ``plausible deniability''~\cite{Warner1965,domingo2018connecting} when using the whole privacy budget to report a single attribute.
    
    \item   We propose three attack models to predict the sampled attribute of users when collecting multidimensional data with the RS+FD solution with about a 2-20 fold increment over a random guess baseline model.
    
    \item   We show through empirical results that the RS+FD solution can prevent (to some extent) re-identification attacks
    
    \item   Finally, we present an adaptation of the RS+FD solution that can serve as a countermeasure to the identified privacy threats while improving both privacy and utility. 
\end{itemize}

\textbf{Outline.} In Section~\ref{sec:preliminaries}, we review the LDP privacy model, the LDP protocols and solutions for collecting multidimensional data investigated in this paper. 
Afterwards, in Section~\ref{sec:attacks}, we present the system overview and adversarial setting for both SMP and RS+FD solutions. 
In Section~\ref{sec:results}, we present our experimental evaluation and analyze our results before in Section~\ref{sec:countermeasure} presenting an improvement of the RS+FD as a countermeasure. 
Next, we provide a general discussion in Section~\ref{sec:disc}.
Finally in Section~\ref{sec:rel_work}, we review related work before concluding with future perspectives of this work in Section~\ref{sec:conclusion}.

\section{Preliminaries} 
\label{sec:preliminaries}

This section briefly reviews the LDP model, state-of-the-art LDP frequency estimation protocols and three solutions for multiple attribute frequency estimation under LDP.

\subsection{Local Differential Privacy} 
\label{sub:privacy_metrics}

In this paper, we use LDP (Local Differential Privacy)~\cite{first_ldp,Duchi2013} as the privacy model considered, which is formalized as:

\begin{definition}[$\epsilon$-Local Differential Privacy]\label{def:ldp} A randomized algorithm ${\mathcal{M}}$ satisfies $\epsilon$-local-differential-privacy ($\epsilon$-LDP), where $\epsilon>0$, if for any pair of input values $v_1, v_2 \in Domain(\mathcal{M})$ and any possible output $y$ of ${\mathcal{M}}$:

\begin{equation} \label{eq:ldp}
    \Pr[{\mathcal{M}}(v_1) = y] \leq e^\epsilon \cdot \Pr[{\mathcal{M}}(v_2) = y] \textrm{.}
\end{equation}
\end{definition}

In essence, LDP guarantees that it is unlikely for the data aggregator to reconstruct the data source regardless of the prior knowledge.
The privacy budget $\epsilon$ controls the privacy-utility trade-off for which lower values of $\epsilon$ result in tighter privacy protection. 
Similar to central DP, LDP also has several important properties, such as immunity to post-processing and composability~\cite{dwork2014algorithmic}. 

\subsection{LDP Frequency Estimation Protocols} 
\label{sub:ldp_protocols}

In this subsection, we review five state-of-the-art LDP protocols, which enables the aggregator to estimate the frequency of any value $v_i \in A_j$, for $i \in [k_j]$, under LDP guarantees. 
 
\subsubsection{Generalized Randomized Response} 
\label{sub:kRR}

Randomized response (RR)~\cite{Warner1965} is the classical technique for achieving LDP, which provides ``plausible deniability'' for individuals responding to embarrassing (binary) questions in a survey.
The Generalized RR (GRR)~\cite{kairouz2016discrete,kairouz2016extremal} protocol extends RR to the case of $k_j \geq 2$ while satisfying $\epsilon$-LDP. 
Given a value $v_i \in A_j$, for $i \in [k_j]$, \textit{GRR($v_i$)} outputs the true value with probability $p$, and any other value $v \in A_j \setminus \{v_i\}$ with probability $1-p$. 
More formally, the perturbation function is:

\begin{equation*}
    \forall{y \in A_j}  : \quad \Pr[y=a] = \begin{cases} p=\frac{e^{\epsilon}}{e^{\epsilon}+k_j-1} , \textrm{ if } a = v\\ q=\frac{1}{e^{\epsilon}+k_j-1}, \textrm{ otherwise} \textrm{,} \end{cases}
\end{equation*}

\noindent in which $y$ is the perturbed value sent to the aggregator. 
The GRR protocol satisfy $\epsilon$-LDP since $\frac{p}{q}=e^{\epsilon}$. 
To estimate the normalized frequency of $v_i \in A_j$, for $i \in [k_j]$, one counts how many times $v_i$ is reported, expressed as $C(v_i)$, and then computes~\cite{tianhao2017}:

\begin{equation}\label{eq:est_pure}
    \hat{f}(v_i) = \frac{C(v_i) - nq}{n(p - q)} \textrm{,}
\end{equation}

\noindent in which $n$ is the total number of users. 
In~\cite{tianhao2017}, it was proven that Eq.~\eqref{eq:est_pure} is an unbiased estimator (\ie, $\mathbb{E}(\hat{f}(v_i))=f(v_i)$).

\subsubsection{Optimal Local Hashing} \label{sub:olh}

Local hashing (LH) protocols can handle a large domain size $k_j$ by first using hash functions to map an input value to a smaller domain of size $g_j$ (typically $g_j \ll k_j)$, and then applying GRR to the hashed value in the smaller domain. 

The authors in~\cite{tianhao2017} have proposed Optimal LH (OLH), which selects $g_j=e^{\epsilon} + 1$. 
Given a value $v_i \in A_j$, for $i \in [k_j]$, in OLH, one reports $\langle H, GRR(H(v_i)) \rangle$ in which $H$ 
is randomly chosen from a family of universal hash functions that hash each value in $A_j$ to $[g_j]=\{1,\ldots,g_j\}$, which is the domain that GRR($\cdot$) will operate on. 
The hash values will remain unchanged with probability $p'$ and switch to a different value in $[g_j]$ with probability $q'$, as:

\begin{equation*}
    \forall{y \in [g_j]}  : \quad \Pr[y=\left( H,a \right)] = \begin{cases} p'=\frac{e^{\epsilon}}{e^{\epsilon}+g_j-1} , \textrm{ if } a = H(v)\\ q'=\frac{1}{e^{\epsilon}+g_j-1}, \textrm{ otherwise} \textrm{,} \end{cases}
\end{equation*}

\noindent in which $y$ is the hash function and perturbed value sent to the aggregator. 
From this, the aggregator can obtain the unbiased estimation of $v_i \in A_j$, for $i \in [k_j]$, with Eq.~\eqref{eq:est_pure} by setting $p=p'$ and $q=\frac{1}{g_j} \cdot p' + \left( 1 -  \frac{1}{g_j} \right) \cdot q' = \frac{1}{g_j}$~\cite{tianhao2017}. 

\subsubsection{Subset Selection} \label{sub:kSS}

The main idea of $\omega$-Subset Selection ($\omega$-SS)~\cite{wang2016mutual,Min2018} is to randomly select $\omega$ items within the input domain to report a subset of values (\ie, $\Omega \subseteq A_j$). 
The user's true value $v_i \in A_j$, for $i \in [k_j]$, has higher probability of being included in the subset $\Omega$, compared to other values in $A_j \setminus \{v_i\}$ that are sampled uniformly at random (without replacement). 
The optimal subset size $\omega=|\Omega|$ that minimizes the variance is $\omega=\frac{k_j}{e^{\epsilon}+1}$~\cite{wang2016mutual,Min2018}. 

Given a value $v_i \in A_j$, for $i \in [k_j]$, the $\omega$-SS protocol starts by initializing an empty subset $\Omega$. 
Afterwards, the true value $v_i$ is added to $\Omega$ with probability $p=\frac{\omega e^{\epsilon}}{\omega e^{\epsilon} + k_j - \omega}$. 
Finally, it adds values to $\Omega$ as follows~\cite{wang2016mutual,Min2018}:

\begin{itemize}
    \item If $v_i$ has been added to $\Omega$ in the previous step, then $\omega - 1$ values are sampled from $A_j \setminus \{v_i\}$ uniformly at random (without replacement) and are added to $\Omega$;
    
    \item If $v_i$ has not been added to $\Omega$ in the previous step, then $\omega$ values are sampled from $A_j \setminus \{v_i\}$ uniformly at random (without replacement) and are added to $\Omega$.
\end{itemize}

From this, the aggregator can obtain the unbiased estimation of $v_i \in A_j$, for $i \in [k_j]$, with Eq.~\eqref{eq:est_pure} by setting $p=\frac{\omega e^{\epsilon}}{\omega e^{\epsilon} + k_j - \omega}$ and $q=\frac{\omega e^{\epsilon}(\omega-1) + (k_j - \omega)\omega}{(k_j - 1)(\omega e^{\epsilon} + k_j - \omega)}$~\cite{wang2016mutual,Min2018}.

\subsubsection{Unary Encoding Protocols} \label{sub:ue_protocols}

Unary encoding (UE) protocols interpret the user's input $v_i \in A_j$, for $i \in [k_j]$ as a one-hot $k_j$-dimensional vector.
More specifically, $B=UE(v_i)$ is a binary vector with only the bit at the position $v_i$ sets to 1 and the other bits set to 0. 
One well-known UE-based protocol is the Basic One-time RAPPOR~\cite{rappor}, hereafter referred to as symmetric UE (SUE)~\cite{tianhao2017}, which randomizes the bits from $B$ independently with probabilities:

\begin{equation}  \label{eq:rappor_parameters}
    \forall{i \in [k_j]} : \quad \Pr[B_i'=1] =\begin{cases} p=\frac{e^{\epsilon/2}}{e^{\epsilon/2}+1}, \textrm{ if } B_i=1 \\ q=\frac{1}{e^{\epsilon/2}+1}, \textrm{ if } B_i=0 \textrm{.}\end{cases}
\end{equation}
Afterwards, the client sends $B'$ to the aggregator. 
More recently, to minimize the variance of the SUE protocol, the authors in~\cite{tianhao2017} proposed Optimal UE (OUE), which selects probabilities $p=\frac{1}{2}$ and $q=\frac{1}{e^{\epsilon}+1}$ in Eq.~\eqref{eq:rappor_parameters} asymmetrically (\ie, $p+q \neq 1$). 
The estimation method used in Eq.~\eqref{eq:est_pure} applies equally to both SUE and OUE protocols, in which both satisfy $\epsilon$-LDP for $\epsilon = ln\left( \frac{p(1-q)}{(1-p)q} \right )$~\cite{rappor,tianhao2017}.

\subsection{Multidimensional Frequency Estimation} \label{sub:multi_freq_est}

Let $n$ be the total number of users, $d\geq 2$ be the total number of attributes, $\textbf{k}=[k_1,k_2,\ldots,k_d]$ be the domain size of each attribute, $\mathcal{M}$ be a local randomizer and $\epsilon$ be the privacy budget. 
Each user holds a tuple $\textbf{v}=[v_1,v_2,\ldots,v_d]$, (\ie, a private discrete value per attribute). 
The two next subsections describes the SPL, SMP and RS+FD solutions for frequency estimation of multiple attributes.

\subsubsection{Standard Solutions} \label{sub:spl_smp_sol}

Previous works in the local DP setting considered the following approaches~\cite{tianhao2017,xiao2,wang2019,Arcolezi2021_allomfree}:

\begin{itemize}
    
    \item  \textbf{SPL.} On the one hand, due to the sequential composition theorem~\cite{dwork2014algorithmic}, users can split the privacy budget $\epsilon$ over the number of attributes $d$ and send all randomized values $y_j$, for $j \in [d]$, with $\frac{\epsilon}{d}$-LDP to the aggregator (\ie, a tuple $\textbf{y}=[y_1,y_2,\ldots,y_d]$). However, this naïve SPL solution leads to high estimation error~\cite{xiao2,wang2019,Arcolezi2021_allomfree,tianhao2017}.
    
    \item  \textbf{SMP.} Instead of splitting the privacy budget $\epsilon$, this state-of-the-art solution allows each user to sample a single attribute $j \in [d]$ at random and uses all the privacy budget to send it with $\epsilon$-LDP~\cite{xiao2,wang2019,Arcolezi2021_allomfree,tianhao2017}. 
    In this case, each user tells the aggregator which attribute is sampled, and what is the perturbed value for it ensuring $\epsilon$-LDP (\ie, $\langle j, y_j \rangle$).

\end{itemize}

\subsubsection{Random Sampling Plus Fake Data (RS+FD)} \label{sub:rspfd_sol}

Because the SMP solution discloses the sampled attribute, one can say that it is not fair to all users (\eg, some users will sample age while others will sample disease). 
To address this issue, the recently proposed RS+FD~\cite{Arcolezi2021_rs_fd} solution is composed of two steps, namely local randomization and fake data generation. 
More precisely, each user samples a unique attribute uniformly at random $j=Uniform\left([d]\right)$ and uses an $\epsilon$-LDP protocol to sanitize its value $v_j$. 
Next, for each non-sampled attribute $i \in [d] \setminus \{j\}$, the user generates uniform random fake data following $A_i$. 
Finally, each user sends the (LDP or fake) value of each attribute to the aggregator (\ie, a tuple $\textbf{y}=[y_1,y_2,\ldots,y_d]$). 
In this manner, the sampling result is not disclosed to the aggregator, thus increasing the \textit{uncertainty}. 
For this reason, to satisfy $\epsilon$-LDP, following the parallel composition theorem~\cite{dwork2014algorithmic} and the amplification by sampling result~\cite{Li2012}, RS+FD utilizes an amplified privacy budget $\epsilon'=\ln{\left( d \cdot (e^{\epsilon} - 1) + 1 \right)}$ for the sampled attribute~\cite{Arcolezi2021_rs_fd}. 

With the RS+FD solution, the estimator should remove the bias introduced by the local randomizer $\mathcal{M}$ and uniform fake data. In~\cite{Arcolezi2021_rs_fd}, the authors used GRR and OUE as LDP protocols within the RS+FD solution, which results in RS+FD[GRR], RS+FD[OUE-z] and RS+FD[OUE-r]. 
We briefly recall how these three protocols, generalizing OUE to UE as one can select either SUE or OUE (\cf{} Section~\ref{sub:ue_protocols}) as local randomizers~\cite{Arcolezi2021_rs_fd,Varma2022}. 

For all three protocols, on the \textit{client-side}, each user randomly samples an attribute $j$ and uses $\mathcal{M}$ to sanitize the value $v_j$ with an amplified privacy parameter $\epsilon'=\ln{\left( d \cdot (e^{\epsilon} - 1) + 1 \right)}$. 
Next, the fake data generation procedure and the unbiased estimator for the frequency of each value $v_i \in A_j$, for $i \in [k_j]$, are as follows:

\begin{itemize}
    
\item \textbf{RS+FD[GRR]~\cite{Arcolezi2021_rs_fd}.} 
For each non-sampled attribute $i \in [d] \setminus \{j\}$, the user generates fake data uniformly at random according to the domain size $k_i$. 
On the \textit{server-side}, the unbiased estimator for this protocol is: $\hat{f}(v_i) = \frac{C(v_i) dk_j - n(d - 1 + qk_j)}{nk_j(p-q)}$, in which $C(v_i)$ is the number of times $v_i$ has been reported, $p=\frac{e^{\epsilon'}}{e^{\epsilon'} + k_j - 1}$ and $q=\frac{1-p}{k_j-1}$.
    
\item  \textbf{RS+FD[UE-z]~\cite{Arcolezi2021_rs_fd}.} For each non-sampled attribute $i \in [d] \setminus \{j\}$, the user generates fake data by applying an UE protocol to zero-vectors (\ie, $[0,0,\ldots,0]$) of size $k_i$. 
On the \textit{server-side}, the unbiased estimator for this protocol is: $\hat{f}(v_i) =  \frac{d(C(v_i)  - nq)}{n(p-q)}$, in which $C(v_i)$ is the number of times $v_i$ has been reported and parameters $p$ and $q$ can be selected following the SUE~\cite{rappor} or OUE~\cite{tianhao2017} protocols.
    
\item  \textbf{RS+FD[UE-r]~\cite{Arcolezi2021_rs_fd}.} For each non-sampled attribute $i \in [d] \setminus \{j\}$, the user generates fake data by applying an UE protocol to one-hot-encoded fake data (uniform at random) of size $k_i$. 
On the \textit{server-side}, the unbiased estimator for this protocol is: $\hat{f}(v_i) =  \frac{C(v_i) d k_j - n\left[ qk_j + (p-q)(d-1) + qk_j(d-1)) \right]}{nk_j(p-q)}$, in which $C(v_i)$ is the number of times $v_i$ has been reported and parameters $p$ and $q$ can be selected following the SUE~\cite{rappor} or OUE~\cite{tianhao2017} protocols.

\end{itemize}

\section{System Overview \& Privacy Threats} \label{sec:attacks}

Hereafter, we describe the system and adversary models before presenting our adversarial analyses of SMP and RS+FD.

\subsection{System Overview} \label{sub:system_overview}

We consider the situation in which a (possibly untrusted) server collects users' multidimensional data $d\geq 2$ for frequency estimation under $\epsilon$-LDP guarantees multiple times.
Particularly, in each data collection (\ie, survey), the server can select a different number of attributes. 
For instance, through a mobile app the server may collect private frequency estimation for different users' demographic data and different application usage (\eg, how much time spent on the application, preferred widget, etc). 
Users could be encouraged to share their private data through the exchange of discount coupons, statistics to compare usage with other users, etc. 
For the sake of simplicity, we assume that the set of users $\mathcal{U}$ is unique across all surveys, although this can be relaxed in real-life allowing users to opt-in or opt-out of a given survey. 
We assume that the server uses one of the state-of-the-art LDP solutions (\eg, SMP or RS+FD) to collect one random attribute per user.
Thus, we do not consider the SPL solution in our attacks as all attributes would be collected at once, thus resulting in a low level of utility~\cite{xiao2,wang2019,Arcolezi2021_allomfree,tianhao2017} 

\textbf{Adversary model.} Following the LDP assumptions~\cite{first_ldp,Duchi2013}, we assume that the server knows the users' pseudonymized IDs, \textit{but not their private data or their real identity}. 
This also implies that the server has no knowledge about the real data distributions. 
However, we assume that the server might have some background knowledge $\mathcal{D}_{BK}$ coming from public available source, such as Census data~\cite{census}.
This background knowledge could for instance contain partial or complete profiles of users along with their true identities.
Thus, the adversary could be for example the server itself, an attacker who intercepts the communication between the client and the server (\eg, through a man-in-the-middle attack) or a third-party analyst with whom the server may have shared the collected data.

\subsection{Attacking SMP: Plausible Deniability and Risks of Re-Identification} \label{sub:atk_models_smp}

\textbf{Plausible deniability.} Let $v_y$ be an embarrassing value of $A_j=\{v_y,v_n\}$ (\eg, a value ``Yes'' for an attribute $A_j$ denoting whether someone cheated on their partner). As long as $\Pr \left[ \mathcal{M}(v_y)=v_y \right] < 1$, the user can deny to have $A_j=v_y$~\cite{domingo2018connecting}.

The LDP protocols of Section~\ref{sub:ldp_protocols} are based on RR~\cite{Warner1965}, which provides ``plausible deniability'' for users' reports. 
However, increasing $\epsilon$ to improve utility of LDP protocols compromises the ``plausible deniability'' of the users' reports. 
Indeed, common $\epsilon$ values used daily by users in high-scale industrial systems nowadays range from small $\epsilon\leq1$ to high values $\epsilon \geq 8$~\cite{desfontaines_dp_real_world,linkedin,tang2017privacy}. 
Thus, we conduct an adversarial analysis to the SMP solution (\cf{} Section~\ref{sub:spl_smp_sol}) in which the user randomly samples a single attribute among $d\geq 2$ ones and uses the whole privacy budget $\epsilon$ to report it. 
Consequently, since the whole privacy budget will be allocated to a single attribute, the ``plausible deniability'' for this attribute will be lower, which can lead an attacker to predict the users' true value as the most likely value after randomization (see details in Section~\ref{sub:plausible_deniability}). 
In this setting, in which many surveys are proposed by the server to the same set of users with possibly different number of attributes (\eg, demographic, preference, application usage, \ldots), an attacker knowing the tuple $\langle$sampled attribute, $\epsilon$-LDP report$\rangle$ will be able to profile each user throughout time. 
Therefore, once a partial or complete profile of the target user is built (see details in Sections~\ref{sub:plausible_deniability_uniform},~\ref{sub:plausible_deniability_non_uniform} and~\ref{sub:re_ident_attack_models}), the adversary could use his background knowledge $\mathcal{D}_{BK}$ to possibly re-identify a user within population~\cite{Samarati2001,SWEENEY2002,samarati1998protecting,Machanavajjhala2006,sweeney2015only,Li2007}, possibly also inferring all other available attributes.
The next four subsections analyze the ``plausible deniability'' of LDP protocols in single and multiple collections, and describes the proposed re-identification attack models, respectively.

\subsubsection{Plausible Deniability of LDP protocols} \label{sub:plausible_deniability}

Given a user's true value $v \in A_j$, different LDP protocols $\mathcal{M}$ have different type of output $y_i=\mathcal{M}\left(v, \epsilon \right)$~\cite{tianhao2017}. 
For instance, UE protocols output unary encoded vectors, $\omega$-SS outputs a subset $\Omega$ of $\omega$ non-encoded values and so on (\cf{} Section~\ref{sub:ldp_protocols}). 
Thus, for each user $u_i \in \mathcal{U}$, for $i \in [n]$, given $y_i$, the adversary's goal is to predict $v_i$, which is denoted as $\hat{v}_i$. 
The attacker's accuracy (ACC) for LDP protocols is measured by the number of correct predictions $v = \hat{v}$ over the number of users $n$: $ACC_{FO} (\%) = 100 \cdot \frac{\sum_{i=1}^{n} f\left(v_i, \hat{v}_i \right)}{n}$, in which $f\left(v,\hat{v}\right)=1$ if $v=\hat{v}$ and $0$ otherwise. 
Following the ``plausible deniability'' intuition and the fact that for all LDP protocols the probability $p$ of reporting the true value $v_i$ (or bit $i$) is higher than any other value $v \in A_j \setminus \{v_i\}$, we now describe our attack strategy to each LDP protocol.
By the time of completing this paper, we learned about a recent work showing that the expectation of our attacks could be analytically formalized with the Bayes adversary of~\cite{Gursoy2022}.
We believe this work is complementary to our ``plausible deniability'' attacking interpretation.

\textbf{Plausible Deniability of GRR.} Since no specific encoding is used with GRR, the most likely value after randomization is the user's $u_i$ own true value $v$. Thus, an attacker can assume that the reported value $y$ is the true one (\ie, $\hat{v}=y$), which gives on expectation an $ACC_{GRR} (\%)= 100 \cdot \frac{e^{\epsilon}}{e^{\epsilon} + k_j - 1}$. 

\textbf{Plausible Deniability of OLH.} 
Since the output of OLH for user $u_i$ is the hash function $H_i$ used to hash the user's value $v$ and the hashed value $h_i=H_i\left(v\right)$, the most likely value after randomization is one within the subset of all values $v \in A_j$ that hash to $h_i$ (\ie, $A_{j_{H}}= \{v|v\in A_j, H_i(v) = h_i\}$). 
Thus, the attacker's best guess is a random choice $\hat{v}=Uniform\left( A_{j_{H}} \right)$. 
On expectation~\cite{Gursoy2022}, one achieves: $ACC_{OLH} (\%)=100 \cdot \frac{1}{2 \cdot \max \left( \frac{k_j}{e^{\epsilon}+1}, 1 \right)}$.

\textbf{Plausible Deniability of $\omega$-SS.} Since the output of $\omega$-SS for user $u_i$ is a set $\Omega \subseteq A_j$, the most likely value after randomization is one within the subset $\Omega$. 
Thus, the attacker's best guess is a random choice $\hat{v}=Uniform\left( \Omega \right)$. 
Selecting $\omega = \frac{k_j}{e^{\epsilon}+1}$ in $\omega$-SS~\cite{wang2016mutual,Min2018}, on expectation~\cite{Gursoy2022}, one achieves: $ACC_{\omega\textrm{-}SS} (\%)=100 \cdot \frac{e^{\epsilon}+1}{2 k_j}$.

\textbf{Plausible Deniability of UE protocols.} Since the output of UE protocols for user $u_i$ is a sanitized unary encoded vector $B$ of size $k_j$, there are three possibilities: 1) a single bit $b$ in $B$ is set to 1, in which the attacker's best guess is to predict the bit as the true value as $\hat{v}=B_b$; 2) more than one bit in $B$ is set to 1, in which the attacker's best guess is a random choice of the bits set to 1 as $\hat{v}=Uniform\left( \left\{b | b \in [k_j] \textrm{ if } B_b=1 \right\} \right)$; and 3) no bit in $B$ is set to 1, in which the attacker's best guess is a random choice of the domain $\hat{v}=Uniform\left( A_j \right)$. 
Therefore, on expectation~\cite{Gursoy2022}, the attacker's accuracy for SUE is: $ACC_{SUE} (\%) = 100 \cdot \frac{1}{k_j \left(e^{\epsilon/2}+1 \right)} \cdot \frac{e^{\epsilon/2}}{e^{\epsilon/2}+1}^{k_j - 1} + \sum_{i=1}^{k_j} \frac{e^{\epsilon/2}}{(e^{\epsilon/2}+1)i} \cdot \textrm{Bin} \left (  i-1 ; k_j - 1 , \frac{1}{e^{\epsilon/2}+1} \right )$, in which $\textrm{Bin}(.)$ denotes a Binomial distribution with $k_j-1$ trials, success probability $\frac{1}{e^{\epsilon/2}+1}$ and exactly $i-1$ successes. On the other hand, on expectation~\cite{Gursoy2022}, the attacker's accuracy for OUE is: $ACC_{OUE} (\%) = 100 \cdot \frac{1}{2 k_j} \cdot \frac{e^{\epsilon}}{e^{\epsilon}+1}^{k_j - 1} + \sum_{i=1}^{k_j} \frac{1}{2i} \cdot \textrm{Bin} \left (  i-1 ; k_j - 1 , \frac{1}{e^{\epsilon}+1} \right )$. 

\subsubsection{Plausible Deniability on Multiple Data Collections: Uniform Privacy Metric} \label{sub:plausible_deniability_uniform}

When collecting multidimensional data $d\geq 2$ with the SMP solution multiple times, the server could implement that all users sample attributes without replacement.
This way, each user will randomly select a new attribute in each data collection (\ie, survey), ensuring a \textit{uniform privacy metric across all users}. 
Since for all LDP protocols the expected $ACC_{FO}$ depends on $\epsilon$ and $k_j$, our analysis focuses on a generic LDP protocol here. 
Therefore, depending on the LDP protocol, the expected ACC with uniform privacy metric after $\#\textrm{surveys}=d$, denoted as $ACC_{FO}^{U}$, now follows: 

\begin{equation} \label{eq:uniform_priv_metric}
    ACC_{FO}^{U} (\%) = 100 \cdot \prod_{j=1}^{d} ACC_{FO} \left ( \epsilon, k_j \right) \textrm{.}
\end{equation}

Since each survey is independent and users sample without replacement, Eq.~\eqref{eq:uniform_priv_metric} represents the expected probability of accurately profiling users with exactly $d$ attributes.

\subsubsection{Plausible Deniability on Multiple Data Collections: Non-Uniform Privacy Metric} \label{sub:plausible_deniability_non_uniform}

On the other hand, when collecting multidimensional data $d\geq 2$ with the SMP solution multiple times, the server can allow users to sample attributes with replacements in each data collection. 
In case of a repeated attribute, the user can report the previous randomized value (\emph{a.k.a.} memoization~\cite{rappor,microsoft,Arcolezi2021_allomfree}).
This way, \textit{users will have a non-uniform privacy metric}. 
Depending on the LDP protocol, the expected ACC with non-uniform privacy metric after $\#\textrm{surveys}=d$, denoted as $ACC_{FO}^{NU}$, now follows: 

\begin{equation} \label{eq:non_uniform_priv_metric}
    ACC_{FO}^{NU} (\%) = 100 \cdot \prod_{j=1}^{d} \frac{d + 1 - j}{d} ACC_{FO} \left ( \epsilon, k_j \right) \textrm{.}
\end{equation}

Since each survey is independent but attributes are sampled with replacement, Eq.~\eqref{eq:non_uniform_priv_metric} denotes the overall adversary's accuracy only considering users that reports a different attribute in each survey (\ie, of accurately profiling users with exactly $d$ attributes). 
Thus, in this setting, users can also end-up with partial profiles.

\textbf{Analytical analysis of expected ACC.} 
In Fig.~\ref{fig:analytical_ACC_uniform_VS_non_uniform}, we illustrate the expected $ACC_{FO}^{U}$ following Eq.~\eqref{eq:uniform_priv_metric} and the $ACC_{FO}^{NU}$ following Eq.~\eqref{eq:non_uniform_priv_metric} of each LDP protocol with the following parameters (taken from Section~\ref{sec:results}): $\epsilon=[1, 2,\ldots,9, 10]$, $d=3$, $\textbf{k}=[74, 7, 16]$, and $\# \textrm{surveys}=d$. 
From Fig.~\ref{fig:analytical_ACC_uniform_VS_non_uniform} (a), one can notice that GRR, $\omega$-SS and SUE have the highest attacker's accuracy, which would enable an adversary to accurately infer a \textbf{complete profile} after $\# \textrm{surveys}=d$.
Allowing users to have non-uniform privacy metrics in the plot (b), minimizes the attacker's accuracy to infer complete profiles as the probability of selecting different attributes in all $d$ surveys is $\frac{d!}{d^d}$.
Note that the expected $ACC_{FO}$ in both Eqs.~\eqref{eq:uniform_priv_metric} and~\eqref{eq:non_uniform_priv_metric} decreases with the $\#\textrm{surveys}$ since the probability of accurately inferring the users' true value is independent in each survey.

\begin{figure}[!ht]
\begin{subfigure}{.5\columnwidth}
  \centering
  \includegraphics[width=0.9\linewidth]{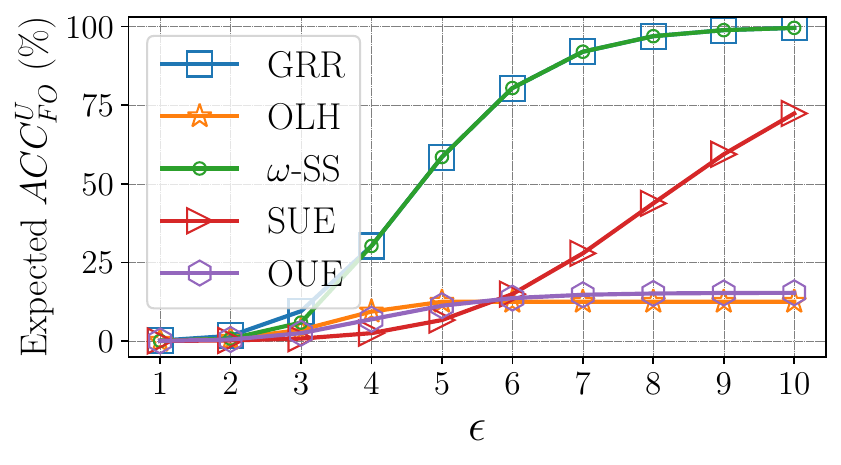}
  \caption{Uniform privacy metric.}
\end{subfigure}%
\begin{subfigure}{.5\columnwidth}
  \centering
  \includegraphics[width=0.9\linewidth]{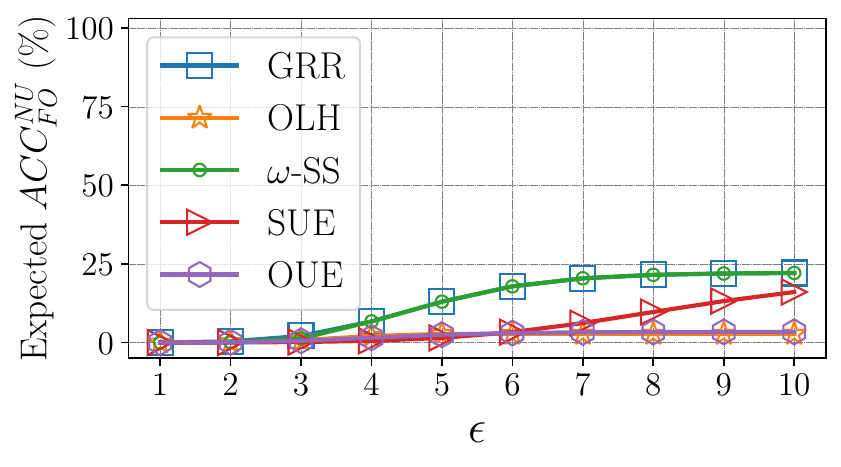}
  \caption{Non-uniform privacy metric.}
\end{subfigure}
\caption{Analytical attacker's accuracy when collecting multidimensional data ($d=3$) with the SMP solution multiple times ($\#surveys=3$) with attributes' domain size $\textbf{k}=[74, 7, 16]$: (a) uniform privacy metric across users with Eq.~\eqref{eq:uniform_priv_metric} and (b) non-uniform privacy metric across users with Eq.~\eqref{eq:non_uniform_priv_metric}.}
\label{fig:analytical_ACC_uniform_VS_non_uniform}
\end{figure}

\subsubsection{Re-Identification Attack Models} \label{sub:re_ident_attack_models}

Following the system overview of Section~\ref{sub:system_overview}, we consider two re-identification attack models: \textbf{full-knowledge re-identification (FK-RI)} and \textbf{partial knowledge re-identification (PK-RI)}, that we detail in the following. 
The first FK-RI model considers that the attacker has access to the complete background knowledge $\mathcal{D}_{BK}$ to re-identify the target user. 
The latter PK-RI model considers that the attacker only has access to a subset $\mathcal{D}_{PK} \subseteq \mathcal{D}_{BK}$ for her re-identification attack.
The re-identification success of both FK-RI and PK-RI models will depend on the results of Sections~\ref{sub:plausible_deniability_uniform} and~\ref{sub:plausible_deniability_non_uniform} to accurately profile the target user, which is impacted by the LDP protocol considered. 

In particular, after $\# \textrm{surveys}$, the attacker will have a profile $\textbf{y}_i$ of at most $\# \textrm{surveys}$ sanitized values for the target user $u_i \in \mathcal{U}$. 
The number of attributes inferred per target user depends on the setting used (\ie, uniform or non-uniform privacy metrics). 
Therefore, the re-identification attack starts with a \textit{matching algorithm} $\mathcal{R}$, which takes as input the sanitized profile $\textbf{y}_i$ and the background knowledge $\mathcal{D}_{BK}$ (or $\mathcal{D}_{PK}$ for PK-RI), and outputs a score $c_i \in \mathbb{R}$. 
More precisely, the score $c_i$ measures the distance between the target $\textbf{y}_i$ and all samples $\textbf{r} \in \mathcal{D}_{BK}$. 
Since the LDP protocols from Section~\ref{sub:ldp_protocols} do not have a notion of ``distance'' when randomizing a value, when an attribute in $\textbf{y}_i \neq \textbf{r}$ the distance is 1 and 0 otherwise. 
A smaller distance between $\textbf{y}_i$ and a profile in $\mathcal{D}_{BK}$ indicates that is highly likely that $\textbf{y}_i$ has been re-identified through the uniqueness combination of $\# \textrm{surveys}$ attributes~\cite{Samarati2001,SWEENEY2002,samarati1998protecting,Machanavajjhala2006,sweeney2015only,Li2007}. 
Finally, a \textit{decision algorithm} $\mathcal{G}$ takes as input the computed distances and outputs a list of top-$k$ possible profiles (or IDs) in $\mathcal{D}_{BK}$ that corresponds to the target user $u_i \in \mathcal{U}$. 
The attacker's re-identification accuracy (RID-ACC) is measured by the number of correct re-identification $u_{id}=\hat{u}_{id}$ over the number of users $n$: $RID\textrm{-}ACC (\%) = 100 \cdot \frac{\sum_{i=1}^{n} f\left(u_{id_{i}}, \hat{u}_{id_{i}} \right)}{n}$, in which $f\left(u_{id},\hat{u}_{id}\right)=1$ if $u_{id}=\hat{u}_{id}$ and $0$ otherwise.
The attacker's RID-ACC depends on the accuracy of partially or completely profiling the target user (\ie, as measured by Eqs.~\eqref{eq:uniform_priv_metric} and~\eqref{eq:non_uniform_priv_metric}) and the ``uniqueness'' of users with respect to the collected attributes (unknown by the server) and in the background knowledge $\mathcal{D}_{BK}$.

\subsection{Attacking RS+FD: Uncovering the Sampled Attribute ($\rightarrow$ SMP)} \label{sub:atk_models_rspfd}

Because the objective of the RS+FD solution is to hide the LDP value among fake data~\cite{Arcolezi2021_rs_fd}, discovering the sampled attribute of each user would convert RS+FD into the SMP solution again. 
Even more, unlike SMP (and SPL), RS+FD utilizes an amplified $\epsilon' > \epsilon$, which decreases the ``plausible deniability'' of the user's report (\cf{} Section~\ref{sub:plausible_deniability}) and could thus be leveraged for re-identification attacks (\cf{} Section~\ref{sub:re_ident_attack_models}) under multiple data collections.

For instance, consider the scenario in which a given user $u \in \mathcal{U}$ whose sampled attribute is $t \in [d]$ produces an RS+FD's output tuple as $\textbf{y}=[y_1, y_2,\ldots, y_d]$. 
In this situation, the \textbf{baseline classification model} is just a random guess $\hat{t}=Uniform([d])$. 
In addition, we propose a \textbf{classifier learning setting} in which an attacker aims to train a classifier over a learning dataset $\textbf{D}_l=\{(\textbf{y}_i, t_i) | \textrm{  } i \in [r]\}$ of $r$ rows and $c=d+1$ columns. 
That is, for each user $u_i$, $\textbf{y}_i$ is the output tuple of the RS+FD solution (LDP/fake values, \ie, a full profile of $d$ attributes) and $t_i$ is the sampled attribute (target is a class within $[d]$). 
\textbf{Because the sampled attribute $t_i$ of users should be unknown to the attacker}, in this work, we propose three settings to build a learning dataset $\textbf{D}_l$, which depends on the attack model. 
In all these settings, we assume that the attacker has the knowledge of the privacy budget $\epsilon$ and the LDP protocol used by users with the RS+FD solution.
Finally, the attacker's attribute inference accuracy (AIF-ACC) is measured by the number of correct predictions $t=\hat{t}$ over the number of users in the testing dataset $n_t$: $AIF\textrm{-}ACC (\%) = 100 \cdot \frac{\sum_{i=1}^{n_t} f\left(t_i, \hat{t}_i \right)}{n_t}$, in which $f\left(t,\hat{t}\right)=1$ if $t=\hat{t}$ and $0$ otherwise.

\subsubsection{No Knowledge: Training a Classifier Over Synthetic Profiles}

With no knowledge of the real sampled attribute of the $n$ users $u \in \mathcal{U}$ and after aggregating users' LDP data, an attacker could use the estimated frequencies $\hat{\textbf{f}}=[\hat{f_1}, \hat{f_2}, \ldots, \hat{f_d}]$ to generate $s$ \textbf{synthetic profiles} $\textbf{s}_i=[s_1, s_2, \ldots, s_d]$, for $i \in [s]$, \ie, mimic the real profiles with one value per attribute. 
Afterwards, for all $s$ synthetic profiles, the attacker could follow the same protocol used by the real users (\emph{i.e}, RS+FD with an LDP protocol) to generate the learning set $\textbf{D}_l$. 
Notice that the attacker has full control over the training set size $s$, which can be seen as a trade-off between computational costs (\ie, generating $s$ synthetic profiles and use as training set) and the attacker's AIF-ACC. 
In this \textbf{no knowledge (NK)} model, the testing set $\textbf{D}_t$ is composed of all the real RS+FD's sanitized tuples $\textbf{y}$ of users $u \in \mathcal{U}$, and the objective is to accurately classify their sampled attribute $t \in [d]$. 

\subsubsection{Partial-Knowledge: Training a Classifier Over Real (Known) Profiles}

This second setting considers the scenario in which the attacker has knowledge about the sampled attribute of $n_{pk} < n$ real users, \ie, the subset $\mathcal{U}_{pk} \subset \mathcal{U}$\footnote{If $\mathcal{U}_{pk} \subseteq \mathcal{U}$, this will correspond to a full-knowledge model in which the adversary has knowledge of all users' sampled attribute (\ie, SMP solution).}. 
This setting corresponds in situations in which some users disclose the sampled attribute by preference (\eg, less ``sensitive'' attributes) or due to security breaches. 
In this \textbf{partial-knowledge (PK)} model, the learning set $\textbf{D}_l$ depends on the number of (compromised) profiles $n_{pk}$ the attacker has access to and the testing set $\textbf{D}_t$ has $n - n_{pk}$ sanitized tuples $\textbf{y}$ of users $u \in \mathcal{U} \setminus \mathcal{U}_{pk}$, in which the objective is to accurately classify their sampled attribute $t \in [d]$.

\subsubsection{Partial-Knowledge Plus Synthetic Profiles}

This last setting combines both NK and PK models, in which the attacker has knowledge about the sampled attribute of $n_{pk} < n$ real users and augments the subset $\mathcal{U}_{pk} \subset \mathcal{U}$ with $s$ synthetic profiles. 
In this \textbf{hybrid model (HM)}, the learning set $\textbf{D}_l$ is dependent on both the number of synthetic profiles $s$ the attacker generates and the number of (compromised) profiles $n_{pk}$ the attacker has access to. 
Similarly to the PK model, the testing set $\textbf{D}_t$ has $n - n_{pk}$ sanitized tuples $\textbf{y}$ of users $u \in \mathcal{U} \setminus \mathcal{U}_{pk}$, and the goal is to accurately classify their sampled attribute $t \in [d]$.

\section{Experimental Evaluation} \label{sec:results}

In this section, we introduce the general setup of our experiments. 
Next, we present the experimental setting and results on the risks of re-identification of the SMP solution. 
Afterwards, we describe the setup of experiments carried out to uncover the sampled attribute of the RS+FD solution. 
Finally, we detail the experimental setting and results on the risks of re-identification of the RS+FD solution.

\subsection{Experimental Setup} \label{sub:setup}

\noindent\textbf{Environment.} All algorithms were implemented in Python 3.
In all experiments, we report the results averaged over 20 runs.

\noindent\textbf{Datasets.} For ease of reproducibility, we conduct our experiments on two census-based multidimensional and open datasets. 

\begin{itemize}
\item  \textbf{ACSEmployement.} This dataset is generated from the Folktables Python package~\cite{ding2021retiring} that provides access to datasets derived from the US Census. 
We have selected the ``Montana'' state only, which results in $n=10,336$ samples with $d=18$ discrete attributes (target included) and domain size $\textbf{k}=[92, 25, 5, 2, 2, 9, 4, 5, 5, 4, 2, 18, 2, 2, 3, 9, 3, 6]$.
    
\item  \textbf{Adult.} This is a classical dataset from the UCI ML repository~\cite{uci} with $n=45,222$ samples after cleaning. 
We selected $d=10$ attributes (``age'',	``workclass'', ``education'',	``marital-status'', ``occupation'', ``relationship'', ``race'', ``sex'', ``native-country'' and ``salary'') with domain size $\textbf{k}=[74, 7, 16, 7, 14, 6, 5, 2, 41, 2]$, respectively. 
\end{itemize}

\subsection{Re-identification Risk of the SMP Solution} \label{sub:results_re_ident_smp}

\noindent\textbf{Methods evaluated.} We consider for evaluation all five LDP protocols described in Section~\ref{sub:ldp_protocols}: GRR, OLH, $\omega$-SS, SUE and OUE. 

\noindent\textbf{Privacy protection.} We vary the privacy budget in the interval $\epsilon=[1, 2,\ldots, 9, 10]$, which corresponds to values used by industry nowadays~\cite{linkedin,desfontaines_dp_real_world} and experiments found in the LDP attacking literature with single~\cite{chatzikokolakis2020bayes,Murakami2021,Gursoy2022} and multiple~\cite{Gadotti2022} collections.

\noindent\textbf{Attack performance metric.} We measure the quality of the re-identification attack with the attacker's re-identification accuracy (RID-ACC) metric, which corresponds to how many times the user is correctly re-identified in the top-$k$ groups, for top-$k$ $\in \{1, 10\}$.

\noindent\textbf{Baseline.} For each top-$k$, the baseline re-identification model follows top-$k$ random guesses (\ie, $\hat{u}_{id}=Uniform([n])$) without replacement with expected RID-ACC: top-$k/n$.

\noindent\textbf{Experimental evaluation.} We set $\# \textrm{surveys}=5$, in which each survey $sv \in [\# \textrm{surveys}]$, has a different number of attributes $d_{sv}=Uniform\left(\frac{d}{2},\ldots,d\right)$ (\ie, with at least $\frac{d}{2}$ attributes). 
The attributes are also selected at random per survey. 
Due to space constraints, we only present here the experiments with the FK-RI model (\cf{} Section~\ref{sub:re_ident_attack_models}), considering the $d$-dimensional dataset as background knowledge $\mathcal{D}_{BK}$, and with the uniform privacy metric setting from Section~\ref{sub:plausible_deniability_uniform}.
Finally, we measure the attacker's RID-ACC after $\# \textrm{surveys}\geq 2$, which results in the inferred profile of each user having respectively $2, 3, 4 \textrm{ or }5$ attributes, to be used for the re-identification attack.

\noindent\textbf{Results.} Fig.~\ref{fig:reident_smp} illustrates the attacker's RID-ACC metric on the Adult dataset for top-$k$ re-identification using the SMP solution, the FK-RI model with uniform $\epsilon$-LDP privacy metric across users, by varying the LDP protocol and the number of surveys. 
Additional results with all LDP protocols, Adult and ACSEmployement datasets, FK-RI and PK-RI models, uniform and non-uniform privacy metric settings as well as with the relaxed LDP metric of~\cite{Murakami2021} are presented in Appendix~\ref{appC:add_re_ident_smp}.

\begin{figure*}[!ht]
\begin{subfigure}{.5\textwidth}
  \centering
  \includegraphics[width=0.9\linewidth]{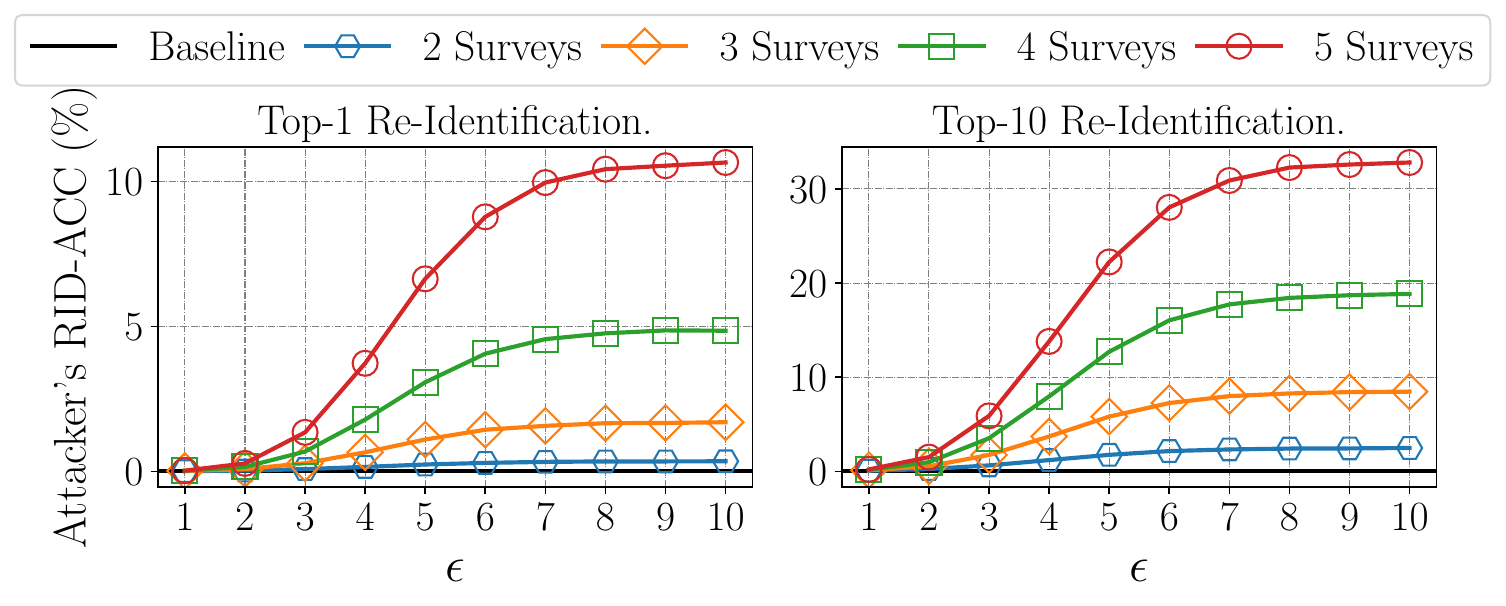}
  \caption{Re-identification risk of the GRR~\cite{kairouz2016discrete,kairouz2016extremal} protocol.}
\end{subfigure}%
\begin{subfigure}{.5\textwidth}
  \centering
  \includegraphics[width=0.9\linewidth]{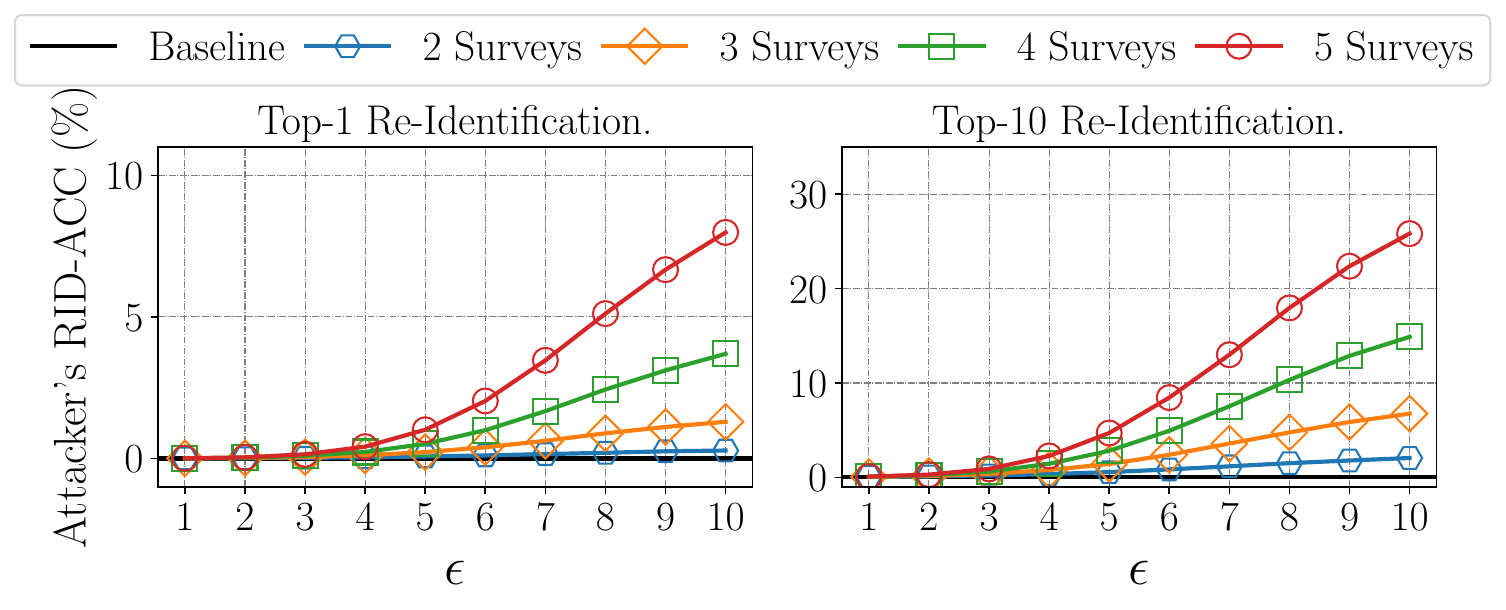}
  \caption{Re-identification risk of the SUE (\emph{a.k.a.} RAPPOR)~\cite{rappor} protocol.}
\end{subfigure}
\\
\begin{subfigure}{.5\textwidth}
  \centering
  \includegraphics[width=0.9\linewidth]{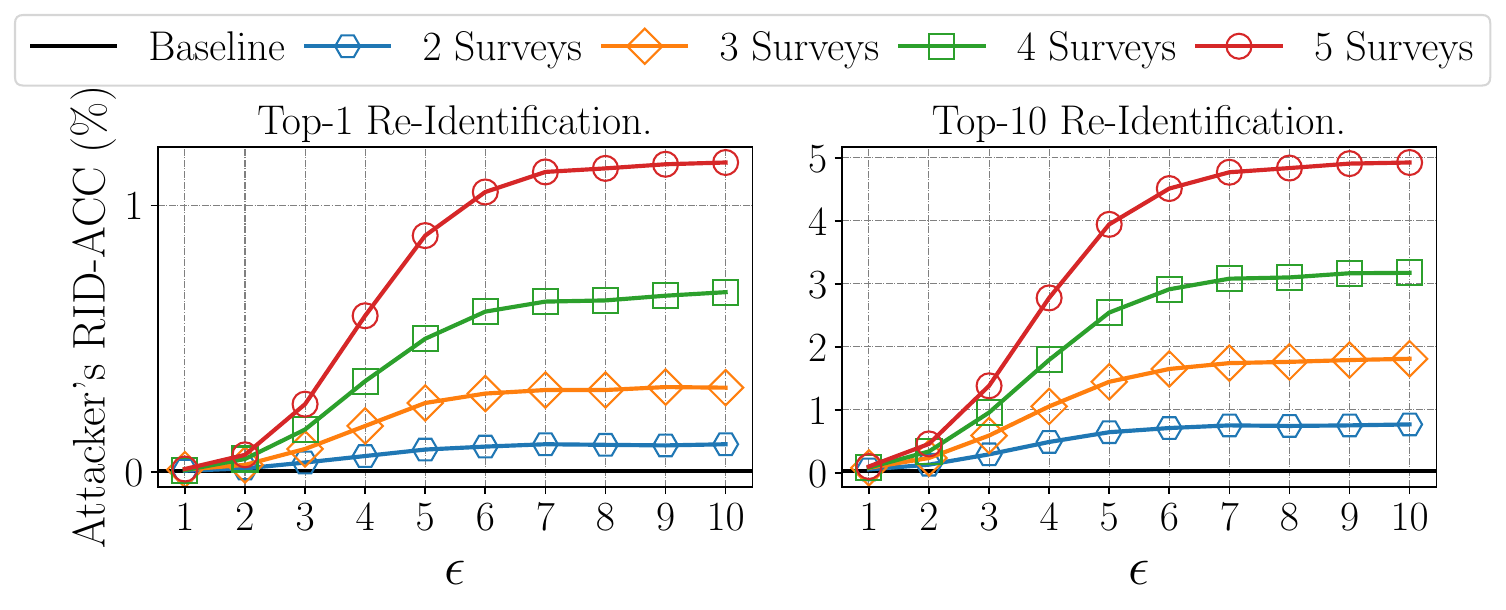}
  \caption{Re-identification risk of the OLH~\cite{tianhao2017} protocol.}
\end{subfigure}%
\begin{subfigure}{.5\textwidth}
  \centering
  \includegraphics[width=0.9\linewidth]{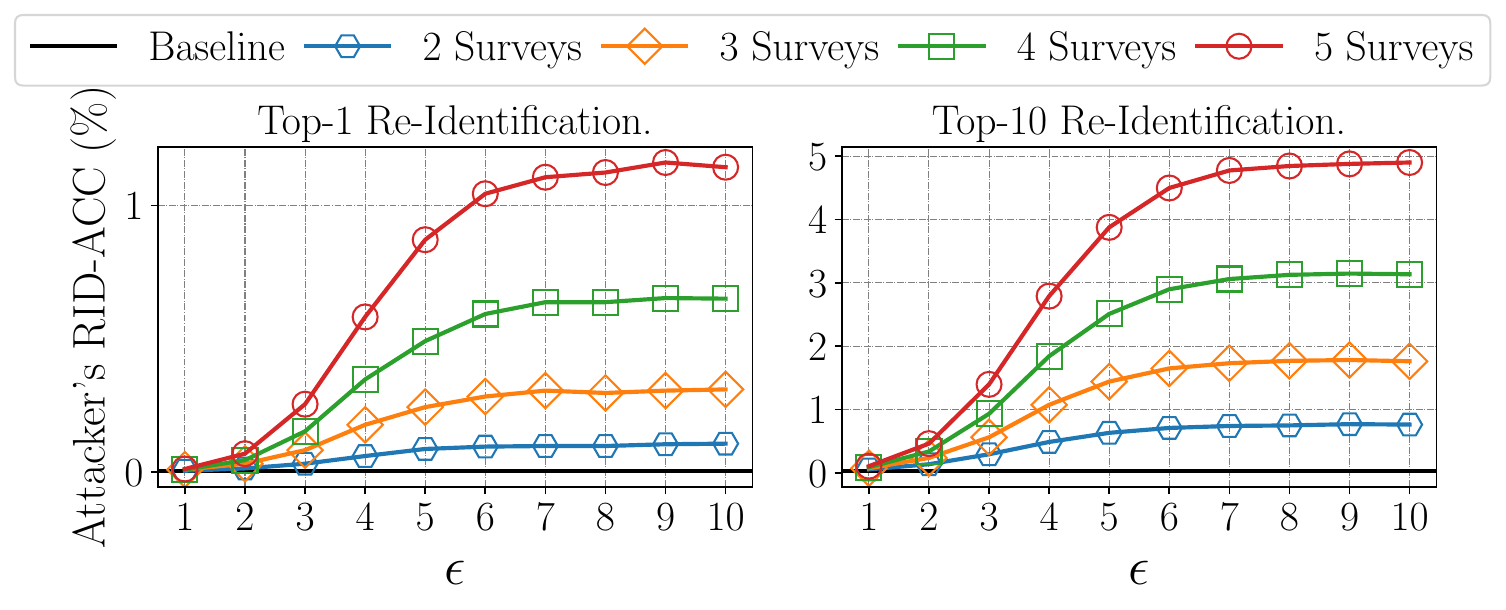}
  \caption{Re-identification risk of the OUE~\cite{tianhao2017} protocol.}
\end{subfigure}
\caption{Attacker's re-identification accuracy (RID-ACC) on the Adult dataset for top-$k$ re-identification on using the SMP solution, the full knowledge FK-RI model with uniform $\epsilon$-LDP privacy metric across users, and by varying the LDP protocol and the number of surveys (\ie, data collections). Omitted results for the $\omega$-SS protocol~\cite{wang2016mutual,Min2018} is due to similarity to plot (a).}
\label{fig:reident_smp}
\end{figure*}

\noindent\textbf{Analysis.} In general, the experimental results of Fig.~\ref{fig:reident_smp} match the numerical results of the expected values from Fig.~\ref{fig:analytical_ACC_uniform_VS_non_uniform}.
From Fig.~\ref{fig:reident_smp}, one can observe that our re-identification attacks present significant improvement over a random baseline model that has $RID\textrm{-}ACC \ll 1 \%$ (\ie, $\textrm{top-}k/n$). 
For instance, with a single shot (\ie, top-$1$), the attacker's RID-ACC is already significant for GRR (and $\omega$-SS) and SUE after about $\# \textrm{surveys} \geq 4$, with at most $\sim 10\%$ of RID-ACC.
In comparison, both OUE and OLH protocols have about 10x less RID-ACC, (\ie, at most $\sim 1\%$ of RID-ACC for top-$1$). 
On the other hand, when there is a set of top-$10$ profiles, the adversary achieves $RID\textrm{-}ACC \geq 2.5 \%$ for GRR (and $\omega$-SS) after only 2 surveys with an upper bound of about $33\%$ of RID-ACC after 5 surveys. 
Though with slightly smaller RID-ACC, the SUE protocol also achieves about $28\%$ of RID-ACC after 5 surveys, and both OUE and OLH are upper bounded by about $5\%$ of RID-ACC.
Although the user is not uniquely re-identified, this still represents a threat due to the possibility of performing, \eg, homogeneity attacks~\cite{Cohen2022,Machanavajjhala2006,Li2007}. 

Overall, these ``high'' re-identification rates may be explained by many factors. 
First, the combination of multiple attributes within the Adult dataset leads to several unique people or small groups of people (this is also the case for the ACSEmployement dataset in Fig. 9 of Appendix~\ref{appC:add_re_ident_smp}).
Additionally, the uniform privacy metric setting require the users to always sample a new attribute, increasing the privacy leakage.
In a more realistic scenario, the non-uniform privacy metric setting minimizes the RID-ACC (see Fig. 11 of Appendix~\ref{appC:add_re_ident_smp}) as already shown in Fig.~\ref{fig:analytical_ACC_uniform_VS_non_uniform}.
Furthermore, the FK-RI model allows the attacker to use the whole background knowledge $\mathcal{D}_{BK}$ to match the inferred profiles. 
For instance, the attacker's RID-ACC metric decreased by almost half when considering the PK-RI model (\cf{} Fig. 10 of Appendix~\ref{appC:add_re_ident_smp}) since there are fewer attributes as background information to use for the matching algorithm $\mathcal{R}$ (see Section~\ref{sub:re_ident_attack_models}).
Lastly, we used the same dataset for private data collection and as (partial) background knowledge. 
A different set of experiments could mix demographic attributes and (synthetic) application usage in each survey, limiting the number of demographic attributes per user to constitute a profile.

\subsection{Uncovering the Sampled Attribute of the RS+FD Solution ($\rightarrow$ SMP)} \label{sub:results_att_inf_rspfd}

\noindent\textbf{Classifier.} We use the state-of-the-art XGBoost~\cite{XGBoost} algorithm to predict the sampled attribute of users in a multiclass classification framework (\ie, $d$ attributes) with default parameters.

\noindent\textbf{Methods evaluated.} We consider for evaluation five protocols within the RS+FD solution from Section~\ref{sub:rspfd_sol}, namely RS+FD[GRR], RS+FD[SUE-z], RS+FD[SUE-r], RS+FD[OUE-z] and RS+FD[OUE-r]. 

\noindent\textbf{Metrics.} Similar to Section~\ref{sub:results_re_ident_smp}, we vary the privacy budget in the interval $\epsilon=[1, 2,\ldots,9, 10]$. 
Besides, we use the attacker's attribute inference accuracy (AIF-ACC) metric to measure the quality of the attack, which corresponds to how many times the attacker can correctly predict the users’ sampled attribute.

\noindent\textbf{Baseline.} The baseline classification model is a random guess $\hat{t}=Uniform([d])$ with expected AIF-ACC: $1/d$.

\noindent\textbf{Experimental evaluation.} All five protocols are evaluated with the three settings of Section~\ref{sub:atk_models_rspfd}, namely No Knowledge (NK), Partial-Knowledge (PK) and Hybrid Model (HM). 
For the NK model, we vary the number of synthetic profiles $s$ the attacker generates in the interval $s=[1n, 3n, 5n]$. For the PK model, we vary the number of compromised profiles $n_{pk}$ the attacker has access to in the interval $n_{pk}=[0.1n, 0.3n, 0.5n]$. 
Finally, for the HM setting, we combined both intervals, \ie, $(s,n_{pk})=[(1n, 0.1n),(3n, 0.3n),(5n, 0.5n)]$.

\noindent\textbf{Results.} Fig.~\ref{fig:attack_rspfd} illustrates the attacker's AIF-ACC metric on the ACSEmployement dataset with the three attack models (\ie, NK, PK and HM) and all five protocols (\ie, RS+FD[GRR], RS+FD[SUE-z], RS+FD[OUE-z], RS+FD[SUE-r] and RS+FD[OUE-r]), varying $\epsilon$, the number of synthetic profiles $s$ and the number of compromised profiles $n_{pk}$. 
Additional results (Adult and Nursery datasets~\cite{uci}) are presented in Appendix~\ref{appD:add_RSpFD}.

\noindent\textbf{Analysis.} From Fig.~\ref{fig:attack_rspfd}, one can notice that the proposed attack models, namely, NK, PK and HM present significant 2-20 fold increments in the attacker's AIF-ACC over the Baseline model. 
Surprisingly, even under an NK model in which the attacker has access only to the estimated frequencies satisfying $\epsilon$-LDP, generating $s=[1n,3n,5n]$ synthetic profiles to train a classifier provides higher attacker's AIF-ACC than having compromised $n_{pk}=0.5n$ profiles in the PK model. 
On the other hand, increasing the number of synthetic profiles $s$ that the attacker generates in the NK model has less impact than increasing the number of compromised profiles $n_{pk}$ that the attacker has access to in the PK model. 
Due to this, results for both NK and HM models are quite similar. 

In this adversarial analysis, the attacker's AIF-ACC now depends on both the LDP protocol and how fake data are generated. 
For the former (\ie, \textit{different LDP protocols}), the difference between RS+FD[GRR] and RS+FD[UE-r] protocols lies in the encoding and randomization steps, which directly affects the attacker's AIF-ACC with a difference of about $5\%$ favoring the RS+FD[GRR] protocol. 
Since GRR requires no particular encoding, there is less noise compared to a randomized unary encoded vector. 
Furthermore, with respect to \textit{different fake data generation procedures}, when fake data are generated with a uniformly random (encoded) value (\ie, RS+FD[GRR] and RS+FD[UE-r]), the attacker's AIF-ACC is upper-bounded by about $25\%$. 
On the other hand, generating fake data through applying a UE protocol on zero-vectors led to an attacker's AIF-ACC of about $50\%$ with RS+FD[OUE-z] and almost $100\%$ with RS+FD[SUE-z] when $\epsilon=10$. 
This high accuracy with RS+FD[UE-z] protocols is because there is only one parameter to perturb each bit when generating fake data, \ie, $\Pr[0\rightarrow 1]=q$ (\cf{} Section~\ref{sub:ue_protocols}). 
When using different UE protocols, the randomization parameters $p$ and $q$ (\cf{} Section~\ref{sub:ue_protocols}) also influence the attacker's AIF-ACC, which led RS+FD[SUE] protocols to have lower attacker's AIF-ACC when $\epsilon$ is small, but higher attacker's AIF-ACC in low privacy regimes. 

Lastly, we remark that due to the original formulation of RS+FD in~\cite{Arcolezi2021_rs_fd} to generate fake data uniformly at random, a classifier was able to learn the sampled attribute from the users, as the distribution of the attributes was not always uniform with the ACSEmployement dataset. 
Nevertheless, when the attributes follow uniform-like distribution, none of the three attack models NK, PK or HM achieves a meaningful increment over the Baseline model (\cf{} results with the Nursery dataset~\cite{uci} in Appendix~\ref{appD:add_RSpFD}).

\begin{figure*}
\begin{subfigure}{.33\textwidth}
  \centering
  \includegraphics[width= 0.712\linewidth]{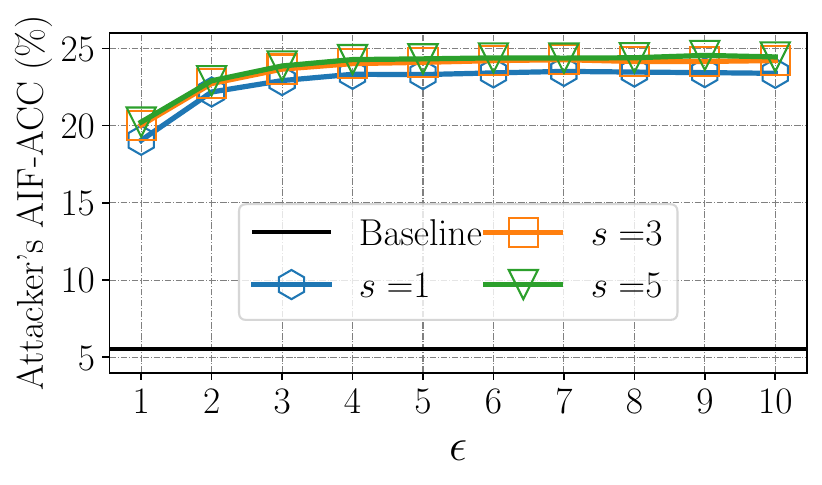}
  \caption{NK model with RS+FD[GRR] protocol.}
\end{subfigure}%
\begin{subfigure}{.33\textwidth}
  \centering
  \includegraphics[width= 0.712\linewidth]{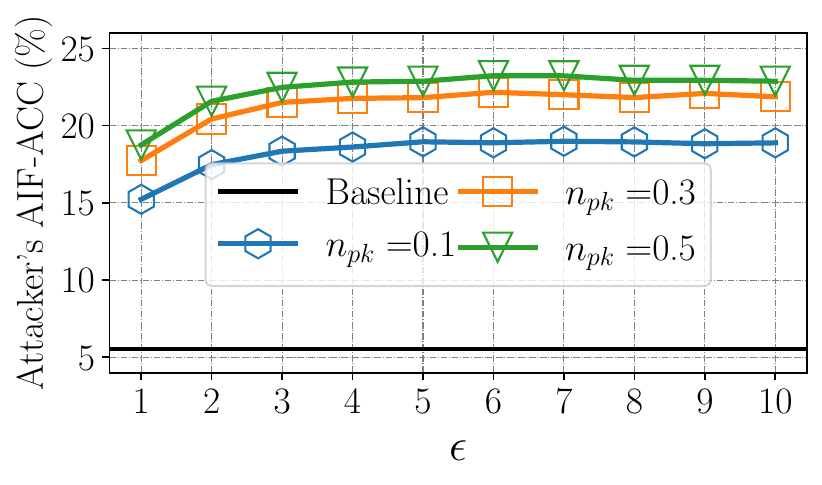}
  \caption{PK model with RS+FD[GRR] protocol.}
\end{subfigure}
\begin{subfigure}{.33\textwidth}
  \centering
  \includegraphics[width= 0.712\linewidth]{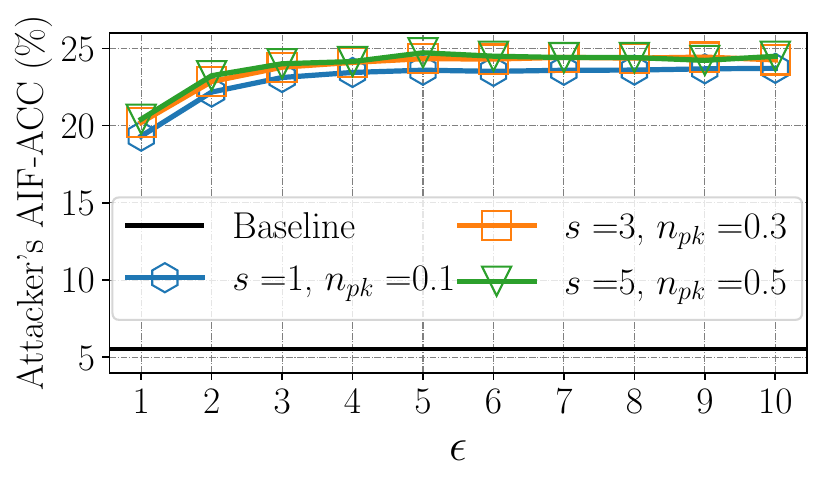}
  \caption{Hybrid model with RS+FD[GRR] protocol.}
\end{subfigure}
\\
\begin{subfigure}{.33\textwidth}
  \centering
  \includegraphics[width= 0.712\linewidth]{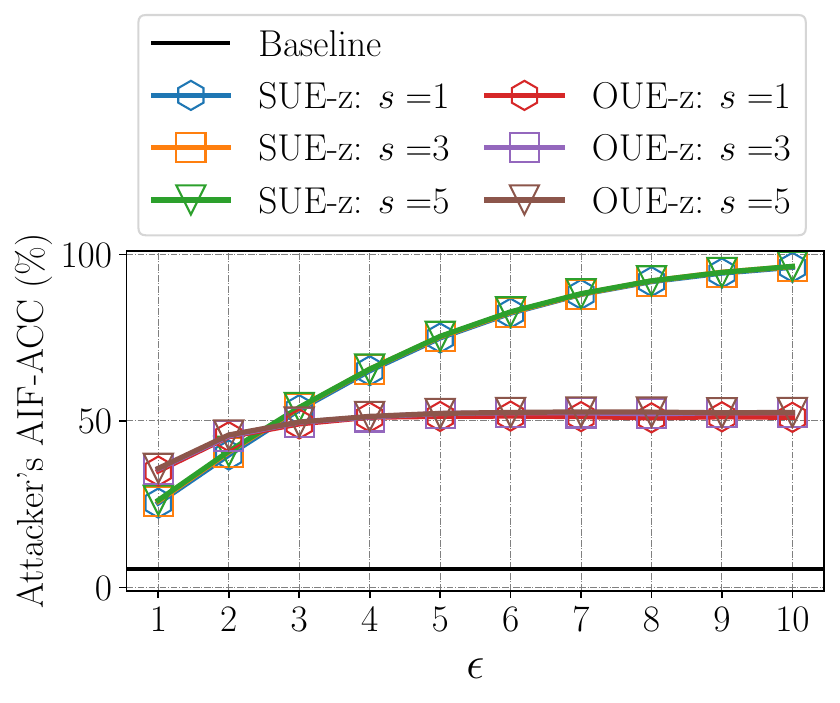}
  \caption{NK model with RS+FD[UE-z] protocols.}
\end{subfigure}%
\begin{subfigure}{.33\textwidth}
  \centering
  \includegraphics[width= 0.712\linewidth]{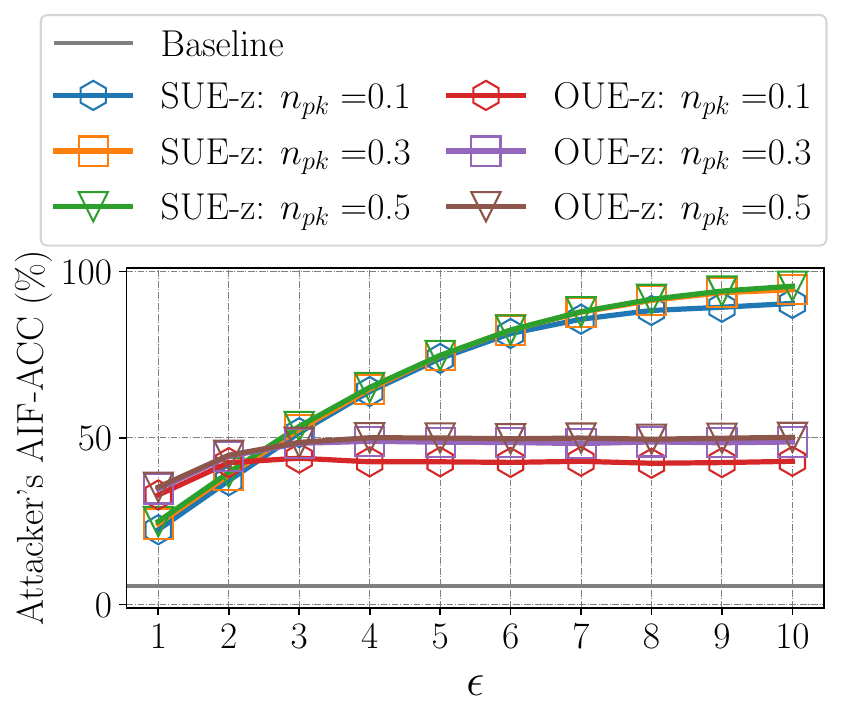}
  \caption{PK model with RS+FD[UE-z] protocols.}
\end{subfigure}
\begin{subfigure}{.33\textwidth}
  \centering
  \includegraphics[width= 0.712\linewidth]{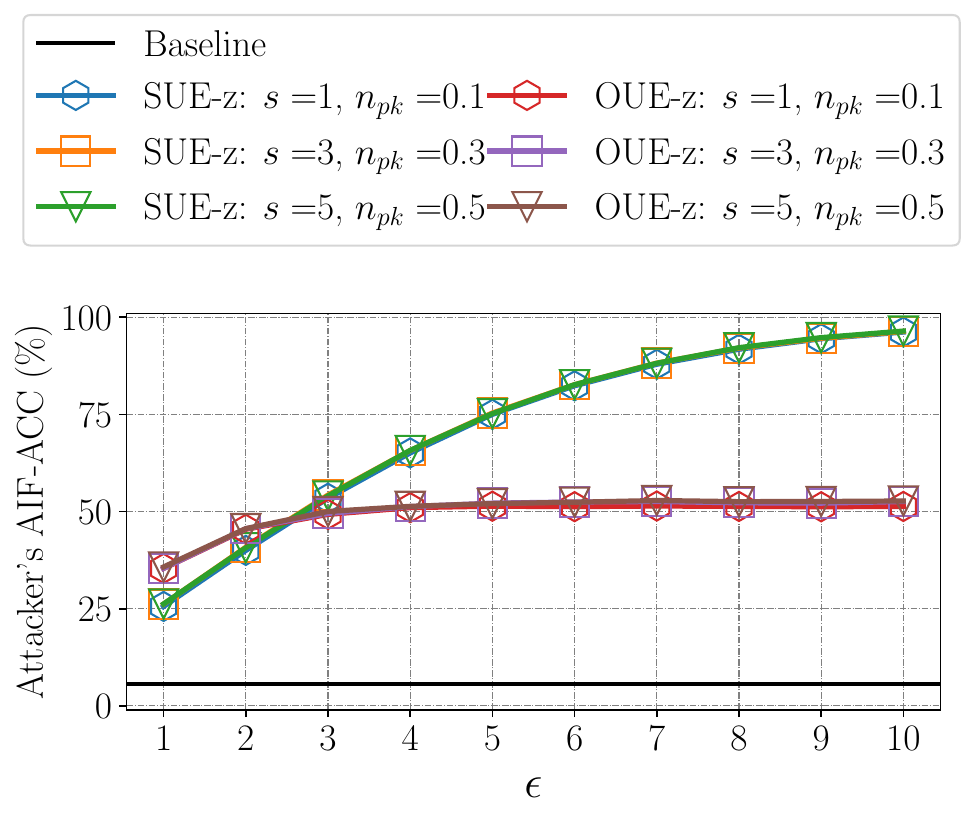}
  \caption{Hybrid model with RS+FD[UE-z] protocols.}
\end{subfigure}
\\
\begin{subfigure}{.33\textwidth}
  \centering
  \includegraphics[width= 0.712\linewidth]{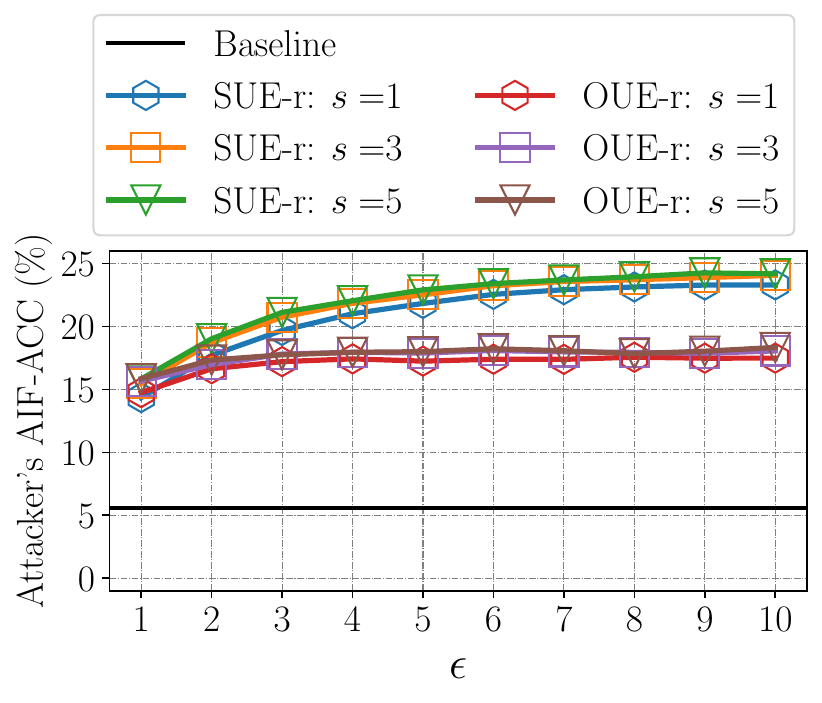}
  \caption{NK model with RS+FD[UE-r] protocols.}
\end{subfigure}%
\begin{subfigure}{.33\textwidth}
  \centering
  \includegraphics[width= 0.712\linewidth]{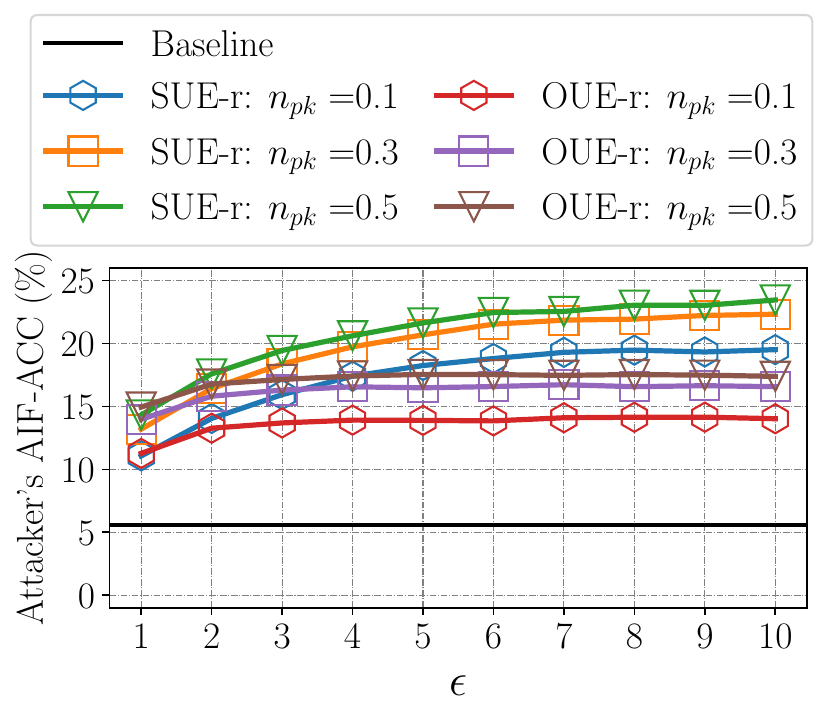}
  \caption{PK model with RS+FD[UE-r] protocols.}
\end{subfigure}
\begin{subfigure}{.33\textwidth}
  \centering
  \includegraphics[width= 0.712\linewidth]{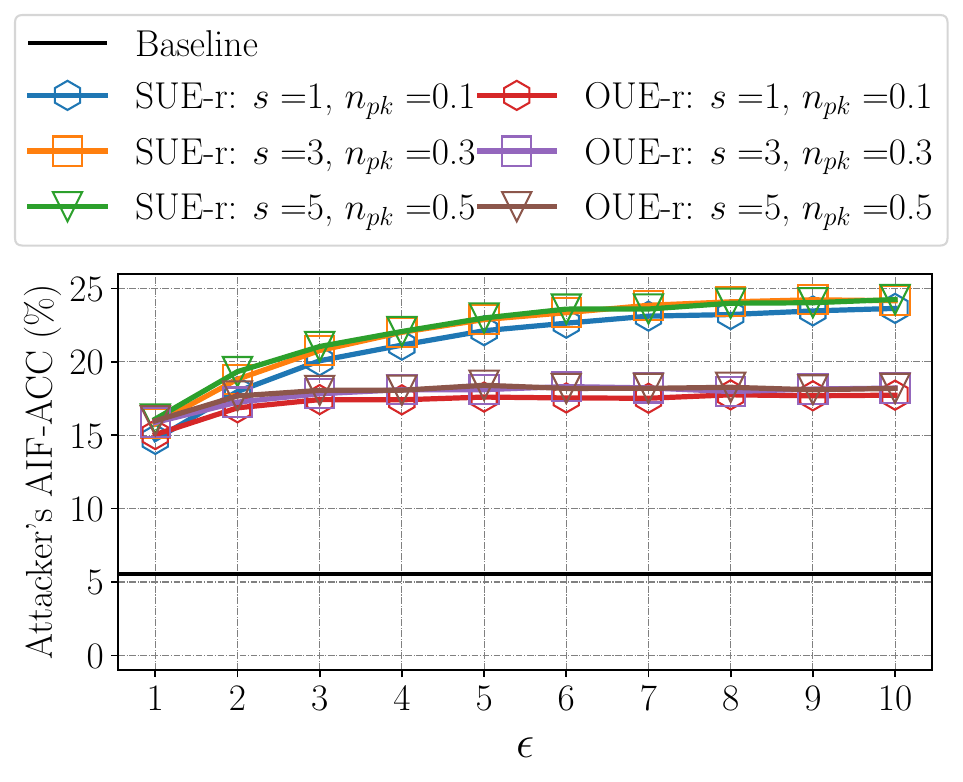}
  \caption{Hybrid model with RS+FD[UE-r] protocols.}
\end{subfigure}
\caption{Attacker's AIF-ACC on the ACSEmployement dataset with three attack models (\ie, NK, PK and hybrid) and five protocols (\ie, RS+FD[GRR], RS+FD[SUE-z], RS+FD[OUE-z], RS+FD[SUE-r] and RS+FD[OUE-r]), varying $\epsilon$, the number of synthetic profiles $s$ the attacker generates and the number of compromised profiles $n_{pk}$ the attacker has access to.}
\label{fig:attack_rspfd}
\end{figure*}

\subsection{Re-identification Risk of the RS+FD Solution} \label{sub:results_re_ident_rspfd}

In this section, we experiment with multiple data collections following the RS+FD solution to measure the attacker's RID-ACC. 
We follow a similar \textbf{experimental evaluation} of Section~\ref{sub:results_re_ident_smp} with the addition of the attribute's inference attack (\cf{} Section~\ref{sub:results_att_inf_rspfd}) in each data collection (\ie, survey). 
To this end, we use the NK model by generating $s=1n$ profiles as accuracy did not substantially increased with higher $s$ (\cf{} Fig.~\ref{fig:attack_rspfd}).
We selected the RS+FD[GRR]~\cite{Arcolezi2021_rs_fd} protocol as it provides an intermediate guarantee between RS+FD[UE-r] (lower bound) and RS+FD[UE-z] (upper bound) protocols. 
We only evaluated the FK-RI model with $\mathcal{D}_{BK}$ and uniform $\epsilon$-LDP privacy metric across users (\ie, users select a new attribute for each survey) as they led to higher re-identification rates using the SMP solution.

\noindent\textbf{Results.} Fig.~\ref{fig:reident_rspfd} illustrates the attacker's RID-ACC metric on the Adult dataset for top-$k$ re-identification using the FK-RI model and the RS+FD[GRR] protocol and by varying the uniform $\epsilon$-LDP privacy metric and the number of surveys. 

\begin{figure}[h!]
    \centering
    \includegraphics[width=0.9\linewidth]{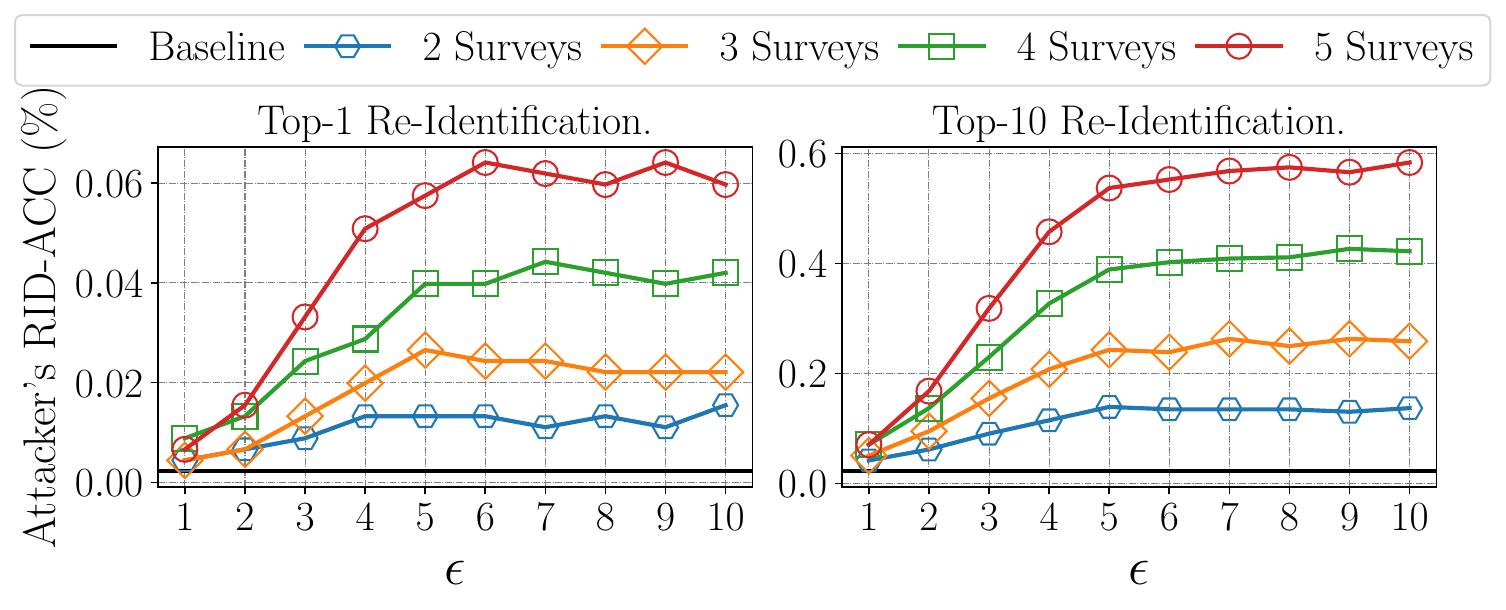}
    \caption{Attacker's re-identification accuracy (RID-ACC) on the Adult dataset for top-$k$ re-identification using the FK-RI model and the RS+FD[GRR] protocol and by varying the uniform $\epsilon$-LDP privacy metric and the number of surveys.}
    \label{fig:reident_rspfd}
\end{figure}

\noindent\textbf{Analysis.} From Fig.~\ref{fig:reident_rspfd}, one can note that the re-identification rates with RS+FD has drastically decreased in comparison with the results of the SMP solution in Fig.~\ref{fig:reident_smp}. 
Re-identification attacks on the RS+FD solution are not trivial, as the attacker has no guarantee that the predicted attribute is correct.
Indeed, from Fig. 14 in Appendix~\ref{appD:add_RSpFD}, the attacker's AIF-ACC on the Adult dataset with the RS+FD[GRR] protocol is upper bounded in $40\%$, which leads to chained errors when profiling a target user in multiple collections. 
For instance, the attacker's RID-ACC for the top-$1$ group is nearly equal the random Baseline model. 
Even for the top-$10$ group the attacker's RID-ACC has meaningful improvement over the Baseline model. 
These results with the RS+FD[GRR] protocol indicates that RS+FD is already (to some extent) a countermeasure to re-identification attacks, except for RS+FD[SUE-z] in which the attacker can predict the attribute with high confidence with high $\epsilon$.

\section{Countermeasure} \label{sec:countermeasure}

As shown in Section~\ref{sub:results_re_ident_rspfd}, the RS+FD solution already provides some resistance to re-identification attacks. 
Thus, we now present an improvement of the RS+FD solution and the experimental results.

\subsection{Random Sampling Plus Realistic Fake Data} \label{sub:RSpRFD}

As briefly described in Section~\ref{sub:multi_freq_est}, the client-side of RS+FD~\cite{Arcolezi2021_rs_fd} is split into two steps (\ie, local randomization and \textit{uniform} fake data generation). 
We now present an improvement of RS+FD, which we call Random Sampling Plus Realistic Fake Data (RS+RFD) as fake data will follow (potentially prior) \textit{non-uniform} distributions. 
For instance, several demographic attributes have national statistics released by the Census~\cite{census} the previous year. 
Therefore, more ``realistic'' profiles can be generated by users to counter the inference of the sampled attribute and consequently the risk of re-identification.

\noindent\textbf{Client-Side.} Alg.~\ref{alg:rsprfd} displays the pseudocode of our RS+RFD solution at the client-side. 
The input of RS+RFD is the user's true tuple of values $\textbf{v} = [v_1,v_2,\ldots, v_d]$, the domain size of attributes $\textbf{k}=[k_1,k_2,\ldots,k_d]$, the attributes' prior distributions $\mathbf{\tilde{f}}=[\tilde{f}_1,\tilde{f}_2,\ldots,\tilde{f}_d]$ (transmitted by the server in advance), the privacy parameter $\epsilon$ and a local randomizer $\mathcal{M}$. 
The output is a tuple $\textbf{y}=[y_1,y_2,\ldots,y_d]$ of values (LDP and fake). 
In Alg.~\ref{alg:rsprfd}, line 6, \textit{Sample} means a random sample is generated following prior $\tilde{f}_i$ of the attribute $i \in [d] \setminus \{j\}$.

\begin{algorithm}[!ht]

\caption{\underline{R}andom \underline{S}ampling \underline{plus} \underline{R}ealistic \underline{F}ake \underline{D}ata (RS+RFD)}
\label{alg:rsprfd}
\begin{algorithmic}[1]
\small
\Statex \textbf{Input :} tuple $\textbf{v} = [v_1,v_2,\ldots, v_d]$, domain size of attributes $\textbf{k}=[k_1,k_2,\ldots,k_d]$, prior distribution of attributes $\mathbf{\tilde{f}}=[\tilde{f}_1,\tilde{f}_2,\ldots,\tilde{f}_d]$, privacy parameter $\epsilon$ and local randomizer $\mathcal{M}$. 
\Statex \textbf{Output :} sanitized tuple $\textbf{y}=[y_1,y_2,\ldots,y_d]$.

\State $\epsilon' = \ln{\left( d \cdot (e^{\epsilon} - 1) + 1 \right)}$ \Comment{Amplification by sampling~\cite{Li2012}} 

\State $j \gets Uniform([d])$ \Comment{Selection of attribute to sanitize}

\State $B_j \gets \texttt{Encode}(v_j)$ \Comment{Encode (if needed)}

\State $y_j \gets \mathcal{M}(B_j, k_j, \epsilon')$ \Comment{Sanitize data of the sampled attribute}

\State \textbf{for} $i \in [d] \setminus \{j\}$ \textbf{do}\Comment{For each non-sampled attributes} 

\State  \hskip1em $y_i \gets \textrm{Sample}(\{1,\ldots,k_i\}, \tilde{f}_i) $ \Comment{Generate one fake data}

\State \textbf{end for}

\Statex \textbf{return :} $\textbf{y}=[y_1,y_2,\ldots,y_d]$ \Comment{Sanitized tuple}
\end{algorithmic}
\end{algorithm}

\noindent\textbf{Server-Side.} The aggregator performs multiple frequency estimation on the collected data by removing bias introduced by the local randomizer $\mathcal{M}$ and fake data. 
The new estimators of using RS+RFD with GRR or UE-based protocols (\eg, SUE~\cite{rappor} or OUE~\cite{tianhao2017}) as local randomizer $\mathcal{M}$ in Alg.~\ref{alg:rsprfd} is presented in the following.
For each attribute $j\in[d]$, the aggregator estimates $\hat{f}(v_i)$ for the frequency of each value $i \in [k_j]$ as:

\begin{itemize}
    
    \item \textbf{RS+RFD[GRR].} The RS+RFD[GRR] estimator is: 
    
    \begin{equation} \label{eq:est_rspfd_grr}
        \hat{f}_{\textrm{GRR}}(v_i) = \frac{ d C(v_i) - n \left (q + (d - 1) \tilde{f}_{j}(v_i) \right )}{n(p-q)} \textrm{,}
    \end{equation}
    
    \noindent in which $C(v_i)$ is the number of times $v_i$ has been reported, $\tilde{f}_{j}(v_i)$ is the prior distribution of value $v_i \in A_j$, $\epsilon'=\ln{\left( d \cdot (e^{\epsilon} - 1) + 1 \right)}$, $p=\frac{e^{\epsilon'}}{e^{\epsilon'} + k_j - 1}$ and $q=\frac{1-p}{k_j-1}$. 
    The probability tree of the RS+RFD[GRR] protocol, the proof that the estimator in Eq.~\eqref{eq:est_rspfd_grr} is unbiased and its variance computation are provided in Appendix~\ref{appA:RSpFD_GRR}.
    
    \item \textbf{RS+RFD[UE-r].} 
    Similar to the RS+FD[UE-r] protocol in Section~\ref{sub:rspfd_sol}, in Line 6 of Alg.~\ref{alg:rsprfd}, for each non-sampled attribute $i$, for $i \in [d] \setminus \{j\}$, the user generates fake data by applying an UE protocol to encoded random data following prior distribution $\tilde{f}_i$. 
    The RS+RFD[UE-r] estimator is: 
        
    \begin{equation} \label{eq:est_rspfd_ue}
        \hat{f}_{\textrm{UE-r}}(v_i) = \frac{ d C(v_i) - n \left ( q+ (p-q) (d-1)  \tilde{f}_{j}(v_i) + q (d-1)  \right )}{n(p-q)} \textrm{,}
    \end{equation}
    
    \noindent in which $C(v_i)$ is the number of times $v_i$ has been reported, $\epsilon'=\ln{\left( d \cdot (e^{\epsilon} - 1) + 1 \right)}$ and $\tilde{f}_{j}(v_i)$ is the prior distribution of value $v_i \in A_j$. 
    Parameters $p$ and $q$ can be selected following the SUE~\cite{rappor} protocol ($p=\frac{e^{\epsilon'/2}}{e^{\epsilon'/2} + 1}$ and $q=\frac{1}{e^{\epsilon'/2} + 1}$) or OUE~\cite{tianhao2017} protocol ($p=\frac{1}{2}$ and $q=\frac{1}{e^{\epsilon'} + 1}$).
    The probability tree of the RS+RFD[UE-r] protocol, the proof that the estimator in Eq.~\eqref{eq:est_rspfd_ue} is unbiased and its variance calculation is provided in Appendix~\ref{appB:RSpFD_UE}.
\end{itemize}

\noindent\textbf{Privacy analysis.} Similar to the RS+FD solution~\cite{Arcolezi2021_rs_fd}, let $\mathcal{M}$ be any existing LDP mechanism, Alg.~\ref{alg:rsprfd} satisfies $\epsilon$-LDP, in a way that $\epsilon'=\ln{\left( d \cdot (e^{\epsilon} - 1) + 1 \right)}$, in which $d$ is the number of attributes.

\noindent\textbf{Limitations.} Besides known limits of the RS+FD solution~\cite{Arcolezi2021_rs_fd,Varma2022}, RS+RFD adds a limitation on being dependent on the underlying prior distributions $\mathbf{\tilde{f}}$ to generate realistic fake data. 
Yet, many demographic attributes have Census data~\cite{census} and other attributes' priors can be defined following domain expert knowledge.

\subsection{Experimental Results} \label{sub:results_rsprfd}

In this section, we present the general setup of experiments with the RS+RFD solution, which includes: the frequency estimation of multiple attributes and the inference attack of the sampled attribute.

\subsubsection{General Experimental Setup} \label{sub:setup_countermeasure}

We use the ACSEmployement dataset described in Section~\ref{sub:setup}.

\noindent\textbf{Prior distribution.} 
To simulate ``Correct'' prior distributions $\tilde{\textbf{f}}=[\tilde{f}_1, \tilde{f}_2, \ldots, \tilde{f}_d]$ to be used to generate realistic fake data with RS+RFD, we perturb the real frequency of each attribute $j\in[d]$ with the standard Laplace mechanism~\cite{Dwork2006,Dwork2006DP,dwork2014algorithmic} in centralized DP satisfying $\epsilon=0.1/d$ (\ie, split $\epsilon=0.1$ by $d$ attributes).
In addition, to simulate an ``Incorrect'' scenario in which prior distributions are wrongly specified, we use Dirichlet distributions with parameter $\mathbf{1}$.

\noindent\textbf{Methods evaluated.} We consider for evaluation three protocols within the RS+RFD solution from Section~\ref{sub:RSpRFD}, namely, RS+RFD[GRR], RS+RFD[SUE-r] and RS+RFD[OUE-r].

\noindent \subsubsection{Frequency Estimation of Multiple Attributes} \label{sub:mult_freq_rsrfd}

We compare the results of our RS+RFD protocols with their respective version within the RS+FD~\cite{Arcolezi2021_rs_fd} solution, \ie, RS+FD[GRR], RS+FD[SUE-r] and RS+FD[OUE-r] (\cf{} Section~\ref{sub:rspfd_sol}). 

\noindent\textbf{Evaluation metrics.} To compare with~\cite{Arcolezi2021_rs_fd}, we vary $\epsilon$ in the interval $\epsilon=[\ln{(2)},\ln{(3)},\ldots,\ln{(7)}]$ and we measure the quality of the estimated frequencies with the averaged mean squared error metric: $MSE_{avg} = \frac{1}{d} \sum_{j \in [d]} \frac{1}{|A_j|} \sum_{v \in A_j}(f(v) - \hat{f}(v) )^2$.

\noindent\textbf{Results.} Fig.~\ref{fig:mult_freq_results} illustrates for all methods the $MSE_{avg}$ metric ($y$-axis) according to the privacy parameter $\epsilon$ (x-axis) for both ``Correct'' and ``Incorrect'' priors. 
Additional empirical and analytical results with the Adult dataset are provided in Appendix~\ref{appE:add_rsprfd}.

\noindent\textbf{Analysis.} For the ``Correct'' prior, one can observe that the $MSE_{avg}$ metric of our proposed RS+RFD protocols consistently and significantly outperform the utility of their respective version within the RS+FD solution. 
The intuition is that since random noise is drawn from realistic prior distributions, the fake data also contributes to the estimation of the attribute. 
Indeed, even with ``Incorrect'' priors, our RS+RFD protocols still outperform the RS+FD protocols, with the exception of RS+RFD[OUE-r] with similar utility RS+RD[OUE-r] in low privacy regimes.
On the other hand, when random noise follows uniform distributions, as with RS+FD, fake data can only increase the estimation of non-correct items. 

\begin{figure}[!htb]
\begin{subfigure}{.5\columnwidth}
  \centering
  \includegraphics[width=0.9\linewidth]{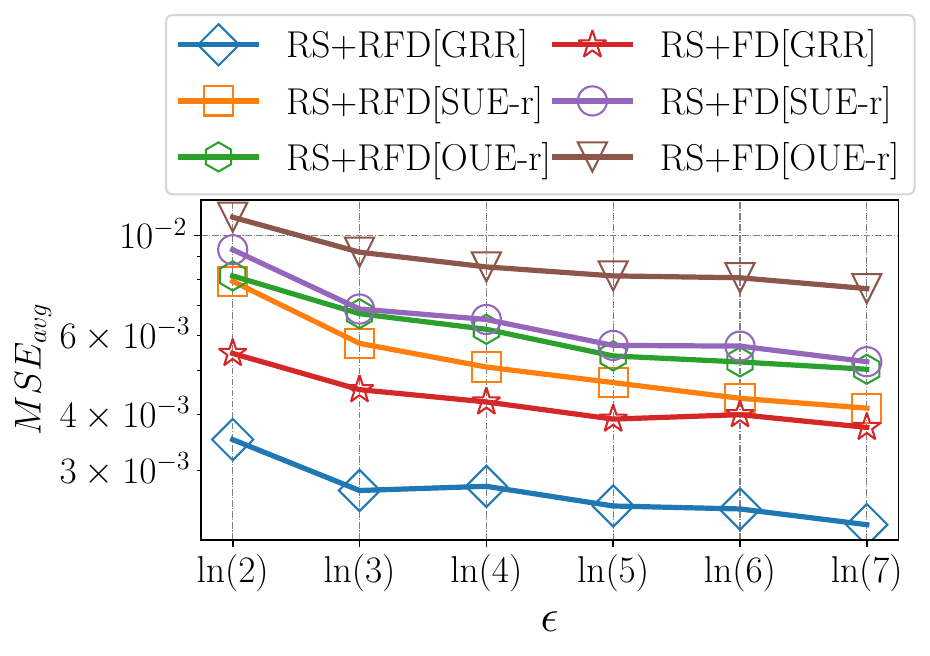}
  \caption{``Correct'' priors.}
\end{subfigure}%
\begin{subfigure}{.5\columnwidth}
  \centering
  \includegraphics[width=0.9\linewidth]{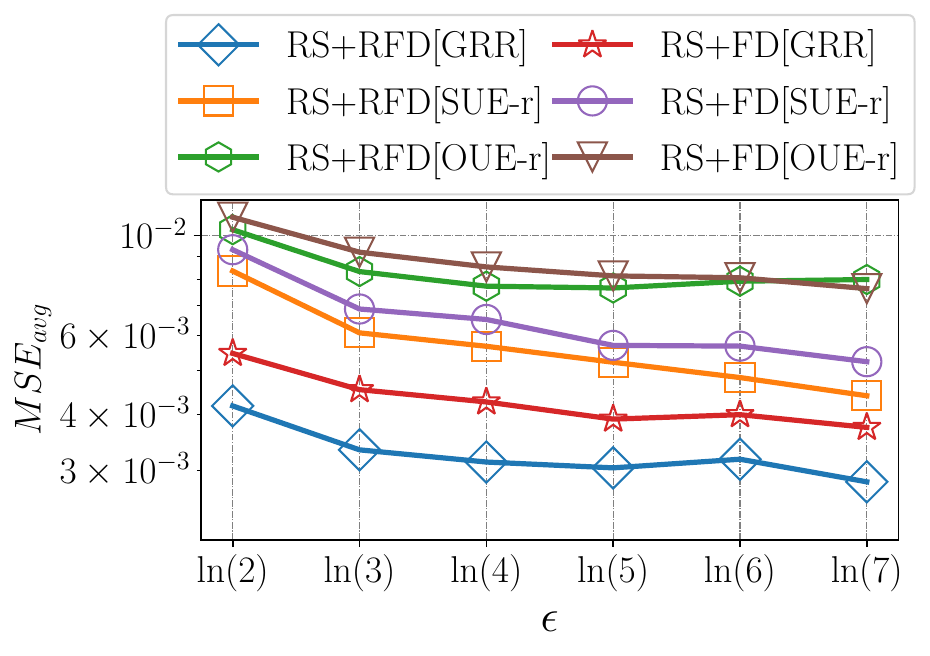}
  \caption{``Incorrect'' priors.}
\end{subfigure}
\caption{Averaged MSE metric varying $\epsilon$ for (a) ``Correct'' and (b)``Incorrect'' priors for multidimensional frequency estimation with the RS+RFD and RS+FD solutions.}
\label{fig:mult_freq_results}
\end{figure}

\subsubsection{Uncovering the Sampled Attribute of the RS+RFD Solution ($\rightarrow$ SMP)} \label{sub:att_inf_rsprfd}

This section follows similar parameters (dataset, $\epsilon$ range and attacker's AIF-ACC metric) used in the experiments of Section~\ref{sub:results_att_inf_rspfd}. 

\noindent\textbf{Results.} Fig.~\ref{fig:attack_rsprfd} illustrates the attacker's AIF-ACC metric on the ACSEmployement dataset with three attack models (\ie, NK, PK and hybrid) and our three protocols (\ie, RS+RFD[GRR], RS+RFD[SUE-r] and RS+RFD[OUE-r] with ``Correct'' priors), varying $\epsilon$, the number of synthetic profiles $s$ and the number of compromised profiles $n_{pk}$. 
Further results with ``Incorrect'' priors are in Appendix~\ref{appE:add_rsprfd}.

\begin{figure*}[!ht]
\begin{subfigure}{.33\textwidth}
  \centering
  \includegraphics[width= 0.712\linewidth]{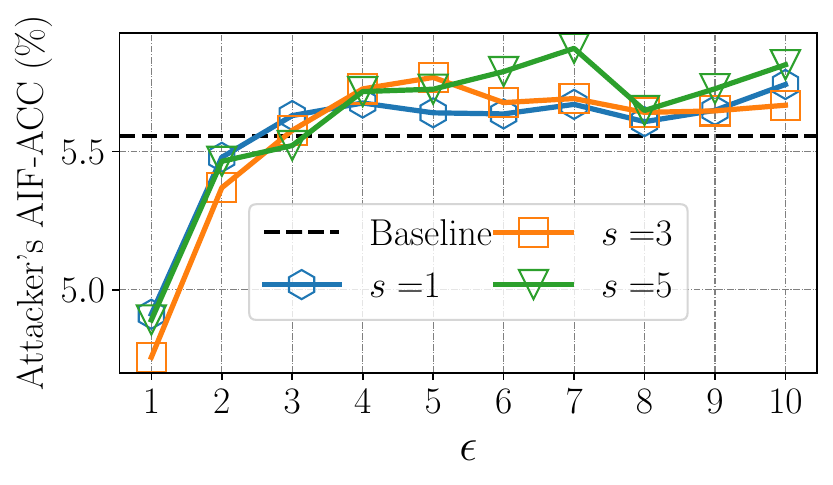}
  \caption{NK model with RS+RFD[GRR] protocol.}
\end{subfigure}%
\begin{subfigure}{.33\textwidth}
  \centering
  \includegraphics[width= 0.712\linewidth]{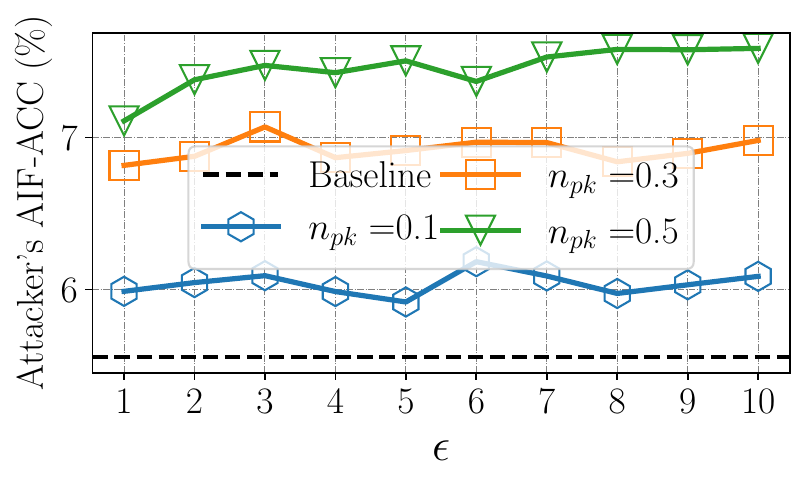}
  \caption{PK model with RS+RFD[GRR] protocol.}
\end{subfigure}
\begin{subfigure}{.33\textwidth}
  \centering
  \includegraphics[width= 0.712\linewidth]{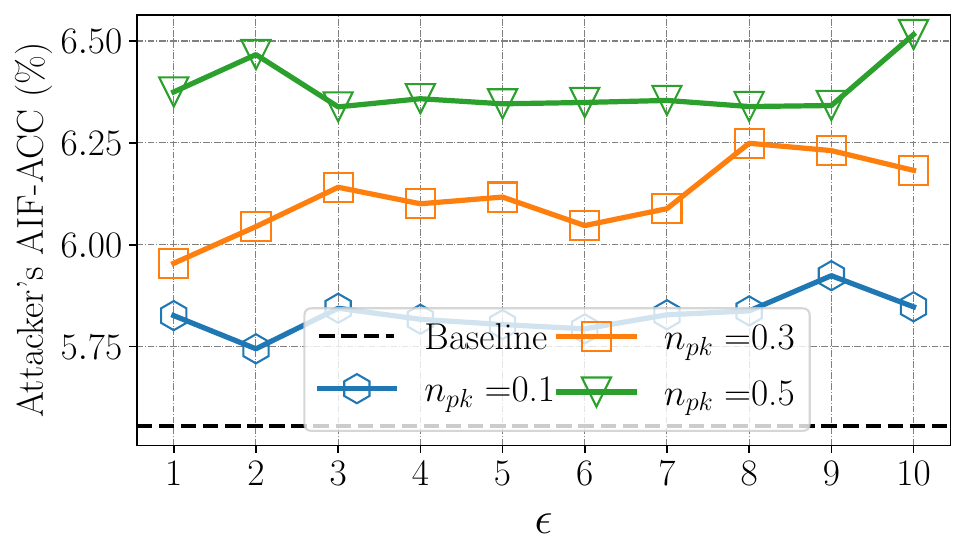}
  \caption{Hybrid model with RS+RFD[GRR] protocol.}
\end{subfigure}
\\
\begin{subfigure}{.33\textwidth}
  \centering
  \includegraphics[width= 0.712\linewidth]{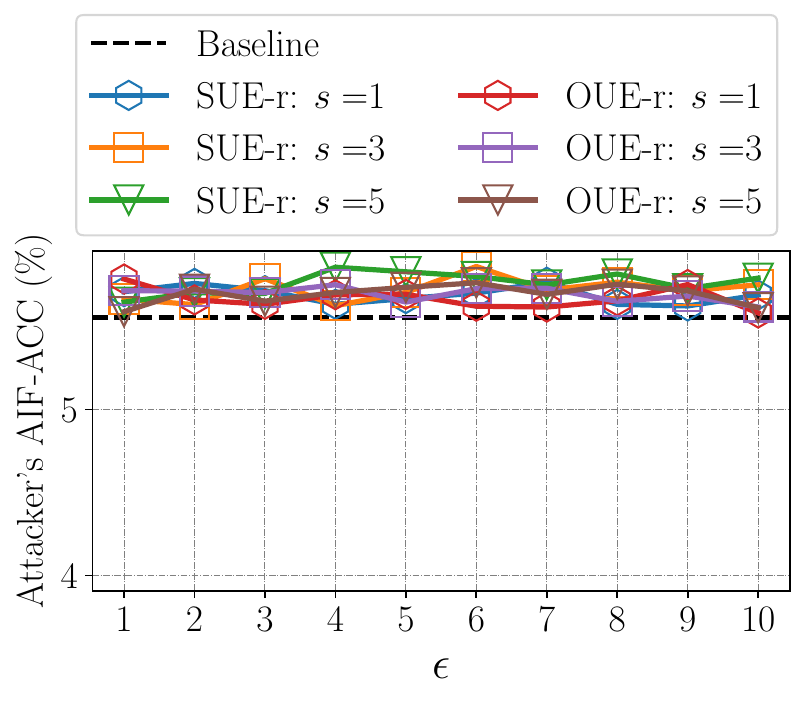}
  \caption{NK model with RS+RFD[UE-r] protocols.}
\end{subfigure}%
\begin{subfigure}{.33\textwidth}
  \centering
  \includegraphics[width= 0.712\linewidth]{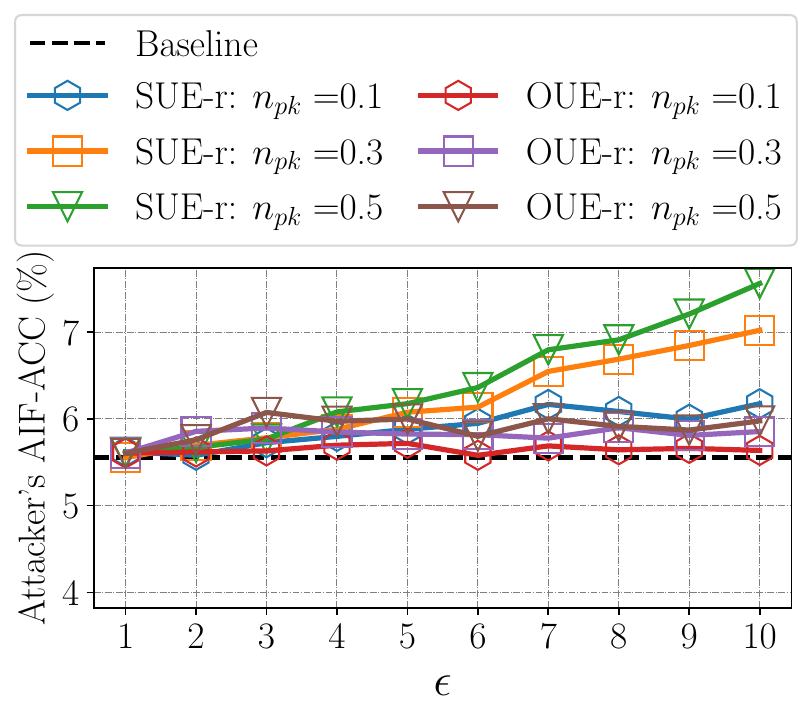}
  \caption{PK model with RS+RFD[UE-r] protocols.}
\end{subfigure}
\begin{subfigure}{.33\textwidth}
  \centering
  \includegraphics[width= 0.712\linewidth]{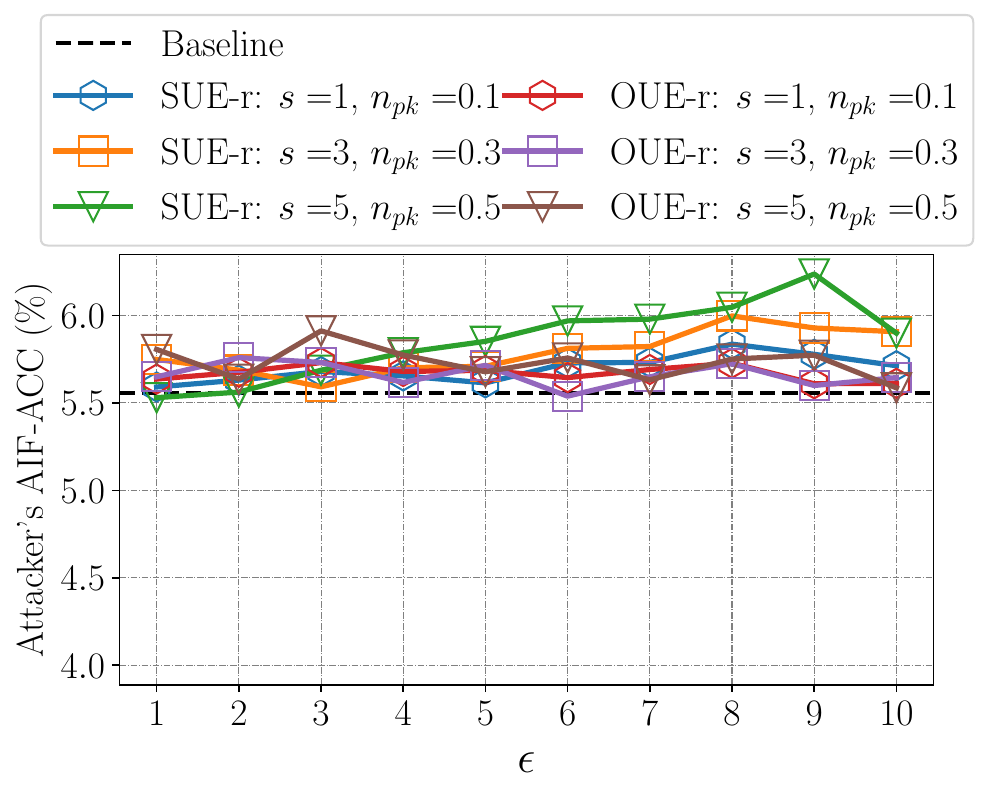}
  \caption{Hybrid model with RS+RFD[UE-r] protocols.}
\end{subfigure}
\caption{Attacker's AIF-ACC on the ACSEmployement dataset with three attack models (\ie, NK, PK and hybrid) and our three protocols (\ie, RS+RFD[GRR], RS+RFD[SUE-r] and RS+RFD[OUE-r] with ``Correct'' priors), varying $\epsilon$, the number of synthetic profiles $s$ the attacker generates and the number of compromised profiles $n_{pk}$ the attacker has access to.}
\label{fig:attack_rsprfd}
\end{figure*}

\noindent\textbf{Analysis.} We highlight that the non-stability in the plots of Fig.~\ref{fig:attack_rsprfd} is due to different sources of randomness: $\epsilon=0.1$-DP for ``Correct'' prior distributions $\mathbf{\tilde{f}}$, $\epsilon$-LDP randomization, fake data generation and the XGBoost algorithm. 
From Fig.~\ref{fig:attack_rsprfd}, one can remark that our RS+RFD protocols considerably decrease the attacker's AIF-ACC when comparing with their respective RS+FD version in Fig.~\ref{fig:attack_rspfd}. 
In contrast with the results of Section~\ref{sub:results_att_inf_rspfd}, the results with the PK model has higher attacker's AIF-ACCs than the NK model.
This is intuitive since the attacker gained ``real'' information of the sampled attribute, increasing the attacker's AIF-ACC as the number of compromised profiles $n_{pk}$ gets higher. 
Nevertheless, for all three NK, PK and HM models, the accuracy gain over a random Baseline model is still minor, highlighting the benefits of our RS+RFD proposal.

\section{Discussion} \label{sec:disc}

In brief, we identified and evaluated empirically two threats to users' privacy when collecting multidimensional data with the state-of-the-art solutions SMP and RS+FD, namely re-identification attack and inference of the sampled attribute.
These threats are generic to any LDP protocol and can be modelled by extending the ``plausible deniability'' attack analysis of Section~\ref{sub:plausible_deniability}.
Hereafter, we summarize the key findings that can be used by practitioners and help substantiate the main claims of this paper.

Regarding the SMP solution, in our experiments, the GRR and $\omega$-SS protocol had the highest RID-ACC as the probability of accurately inferring the user's full profile was higher with relatively small $k_j$ values (see also Fig.~\ref{fig:analytical_ACC_uniform_VS_non_uniform}).
With other protocols, such as OLH and OUE, which are the current state-of-the-art for preserving utility~\cite{tianhao2017}, the adversary cannot accurately infer the profile of users when using $\epsilon$-LDP as privacy model, which leads to lower re-identification risks (see Fig.~\ref{fig:reident_smp} (c) and (d)). 
On the other hand, as shown in Appendix~\ref{appC:add_re_ident_smp}, when using the relaxed version of LDP from~\cite{Murakami2021}, the RID-ACC increases considerably for both OLH and OUE protocols.
Though we only experimented with $\# \textrm{surveys}\leq 5$, we believe that more data collections can lead to higher RID-ACC as long as the profile is accurately inferred.
Yet, under standard sequential composition~\cite{dwork2014algorithmic}, the overall privacy loss is excessive when using high values for $\epsilon$, but we have also considered them due to their use in practical deployments~\cite{apple,tang2017privacy} and similar experiments found in~\cite{Gadotti2022} (though with higher $\# \textrm{surveys} \in \{7, 30, 90, 180\}$).

On the other hand, when using the RS+FD to ``hide'' the sampled attribute, the utility-oriented protocol RS+FD[UE-z] has the highest AIF-ACC due to generating fake data with zero-vectors and we recommend not using it in practice.
Even with the RS+FD[GRR] or RS+FD[UE-r] protocols the attacker's AIF-ACC is considerably greater than a random guess.
Yet, since there are chained errors in multiple collections on accurately predicting the sampled attribute and on inferring the user's value, the RS+FD considerably minimizes the risks of re-identification presented by the SMP solution.

Overall, though some LDP protocols minimized the RID-ACC or AIF-ACC in our experiments (see the main body and Appendices~\ref{appC:add_re_ident_smp},~\ref{appD:add_RSpFD} and~\ref{appE:add_rsprfd}), they did not fully mitigate the risks when increasing $\epsilon$ as done in practice to get more accurate estimations. 
This means they still allow a small portion of users to leak more information than others and corroborate with DP consensus of using $\epsilon \leq 1$. 

Therefore, considering the setting described in Section~\ref{sub:system_overview}, the overall recommendation when using the SMP solution is to select: the standard $\epsilon$-LDP as privacy model, the OUE and/or OLH protocols (depending on $k_j$ due to communication costs~\cite{tianhao2017}), the non-uniform privacy metric setting (\ie, allowing users to sample with replacement and enforce memoization~\cite{microsoft,rappor,Arcolezi2021_allomfree}) and to keep $\epsilon \leq 1$.
On the other hand, when using the RS+FD solution, even when no prior is available, we highly recommend the proposed version in this paper, \ie, RS+RFD with non-uniform fake data.

\section{Related Work} \label{sec:rel_work}

The literature on the local DP model has largely explored the issue of improving the utility of LDP protocols~\cite{wang2019,Arcolezi2021_rs_fd,xiao2,tianhao2017,Wang2018,Wang2020_post_process,kairouz2016discrete,kairouz2016extremal,apple,rappor,Arcolezi2021_allomfree,microsoft,Duchi2013,Varma2022}. 
Recently, a few works have started to design attacks on LDP protocols. 
Some authors focused on maliciously modifying the estimated statistic on the server through targeted or untargeted attacks~\cite{Cheu2021,cao2021data,wu2021poisoning,li2022fine}. 
To counter such kinds of attacks, some works~\cite{Ambainis2004,Kato2021} investigated cryptography-based approaches. 

These targeted or untargeted attacks raise awareness of potential \textbf{security vulnerabilities} of LDP protocols. 
However, these attacks do not aim to attack \textbf{users' privacy} as initially investigated in~\cite{chatzikokolakis2020bayes,Murakami2021,Gursoy2022,Gadotti2022} and in this work.
For instance, Chatzikokolakis \etal~\cite{chatzikokolakis2020bayes} proposed the Bayes security measure to quantify the expected gain over a random guess of an adversary that observes a report of the RR protocol.
Similar to~\cite{chatzikokolakis2020bayes}, this paper provides a ``plausible deniability'' attacking interpretation of five state-of-the-art LDP protocols to infer the user's true value by observing an LDP report.
In an independent and concurrent work, Gursoy \etal~\cite{Gursoy2022} proposed a formalized Bayes adversary for the same attack, which was referenced in Section~\ref{sub:plausible_deniability} to give the expected accuracy of our analyses.
Besides, we extended our attack to multiple collections of multidimensional data in Sections~\ref{sub:plausible_deniability_uniform} and~\ref{sub:plausible_deniability_non_uniform}, which were proposed to account for the consequent risks of re-identification~\cite{Murakami2021,Narayanan2008,Murakami2017,Gambs2014} in Section~\ref{sub:re_ident_attack_models}. 
Re-identification risks in the LDP model for single-frequency estimation were first investigated by Murakami and Takahashi~\cite{Murakami2021}. 
However, different from~\cite{Murakami2021} that focused on a single attribute (\eg, location traces), our work considers multiple attributes being collected multiple times.
Regarding multiple collections, Gadotti \etal~\cite{Gadotti2022} introduced pool inference attacks to LDP protocols for \textit{single-frequency estimation} in a way that an adversary can infer the user’s preferred pool (\eg, skin tone used in emojis).

On the other hand, Arcolezi \etal~\cite{Arcolezi2021_rs_fd} introduced the RS+FD solution focusing only on the utility of the protocols, which was also later studied in~\cite{Varma2022}. 
In this work, we are the first to propose three attack models to the RS+FD solution, showing it is possible to distinguish the $\epsilon$-LDP report from fake data. 
Consequently, in multiple collections, we also show that RS+FD is still subject to (reduced) re-identification risks. 
We thus proposed an improvement of the RS+FD solution that generates non-uniform fake data (\ie, RS+RFD of Section~\ref{sub:RSpRFD}) and can serve as a countermeasure solution.

\section{Conclusion and perspectives} \label{sec:conclusion}

In this paper, we studied privacy threats against LDP protocols for multidimensional data following two state-of-the-art solutions for frequency estimation of multiple attributes, \ie, SMP and RS+FD~\cite{Arcolezi2021_rs_fd}. 
On the one hand, we presented inference attacks based on ``plausible deniability''~\cite{Warner1965} of five widely used LDP protocols (\ie, GRR~\cite{kairouz2016discrete,kairouz2016extremal}, OLH~\cite{tianhao2017}, $\omega$-SS~\cite{wang2016mutual,Min2018}, RAPPOR~\cite{rappor} and OUE~\cite{tianhao2017}) under multiple collections following the SMP solution. 
This analysis also empirically clarifies the risks of re-identification when an attacker is able to build complete and/or partial profiles of users and can correlate them with prior knowledge. 

In addition, we introduced three attack models to infer the sampled attribute of the RS+FD~\cite{Arcolezi2021_rs_fd} solution, which allowed us to still reconstruct complete and/or partial profiles of users and lead to re-identification (although to a much lesser extent than the SMP solution). 
Finally, we proposed a refinement to the RS+FD solution, called RS+RFD that improves both utility and privacy.
That is, in our experiments, RS+RFD minimized the estimation error in comparison with the RS+FD solution, as well as almost fully mitigated the inference of the sampled attribute attack. 

Though we identified and investigated two privacy threats for LDP protocols for multidimensional data in single and multiple data collections, these are not unique and we believe that our work opens new avenues of research in this direction.
For future work, we suggest and aim to formalize the re-identification risks considering different LDP and $d$-privacy~\cite{Chatzikokolakis2013,Alvim2018,Wang2017} protocols, the number of collections, the number of attributes and the ``uniqueness'' of users in a given dataset.
Such a formalization will allow to design other countermeasure solutions beyond RS+FD~\cite{Arcolezi2021_rs_fd} and our RS+RFD.


\begin{acks}
The authors deeply thank the anonymous reviewers for their insightful suggestions.
This work was partially supported by the ERC project HYPATIA with grant agreement Nº 835294 and by the EIPHI-BFC Graduate School (contract "ANR-17-EURE-0002"). 
Sébastien Gambs is supported by the Canada Research Chair program as well as a Discovery Grant from NSERC. 
All computations were performed on the ``M\'esocentre de Calcul de Franche-Comt\'e''.
\end{acks}

\bibliographystyle{ACM-Reference-Format}
\bibliography{references.bib}


\begin{thebibliography}{56}


\ifx \showCODEN    \undefined \def \showCODEN     #1{\unskip}     \fi
\ifx \showDOI      \undefined \def \showDOI       #1{#1}\fi
\ifx \showISBNx    \undefined \def \showISBNx     #1{\unskip}     \fi
\ifx \showISBNxiii \undefined \def \showISBNxiii  #1{\unskip}     \fi
\ifx \showISSN     \undefined \def \showISSN      #1{\unskip}     \fi
\ifx \showLCCN     \undefined \def \showLCCN      #1{\unskip}     \fi
\ifx \shownote     \undefined \def \shownote      #1{#1}          \fi
\ifx \showarticletitle \undefined \def \showarticletitle #1{#1}   \fi
\ifx \showURL      \undefined \def \showURL       {\relax}        \fi
\providecommand\bibfield[2]{#2}
\providecommand\bibinfo[2]{#2}
\providecommand\natexlab[1]{#1}
\providecommand\showeprint[2][]{arXiv:#2}

\bibitem[\protect\citeauthoryear{Abowd}{Abowd}{2018}]%
        {census}
\bibfield{author}{\bibinfo{person}{John~M. Abowd}.}
  \bibinfo{year}{2018}\natexlab{}.
\newblock \showarticletitle{The {U.S.} Census Bureau Adopts Differential
  Privacy}. In \bibinfo{booktitle}{\emph{Proceedings of the 24th {ACM} {SIGKDD}
  International Conference on Knowledge Discovery {\&} Data Mining}}.
  \bibinfo{publisher}{{ACM}}.
\newblock
\urldef\tempurl%
\url{https://doi.org/10.1145/3219819.3226070}
\showDOI{\tempurl}


\bibitem[\protect\citeauthoryear{Alvim, Chatzikokolakis, Palamidessi, and
  Pazii}{Alvim et~al\mbox{.}}{2018}]%
        {Alvim2018}
\bibfield{author}{\bibinfo{person}{Mario Alvim}, \bibinfo{person}{Konstantinos
  Chatzikokolakis}, \bibinfo{person}{Catuscia Palamidessi}, {and}
  \bibinfo{person}{Anna Pazii}.} \bibinfo{year}{2018}\natexlab{}.
\newblock \showarticletitle{Invited Paper: Local Differential Privacy on Metric
  Spaces: Optimizing the Trade-Off with Utility}. In
  \bibinfo{booktitle}{\emph{2018 {IEEE} 31st Computer Security Foundations
  Symposium ({CSF})}}. \bibinfo{publisher}{{IEEE}}.
\newblock
\urldef\tempurl%
\url{https://doi.org/10.1109/csf.2018.00026}
\showDOI{\tempurl}


\bibitem[\protect\citeauthoryear{Ambainis, Jakobsson, and Lipmaa}{Ambainis
  et~al\mbox{.}}{2004}]%
        {Ambainis2004}
\bibfield{author}{\bibinfo{person}{Andris Ambainis}, \bibinfo{person}{Markus
  Jakobsson}, {and} \bibinfo{person}{Helger Lipmaa}.}
  \bibinfo{year}{2004}\natexlab{}.
\newblock \showarticletitle{Cryptographic Randomized Response Techniques}.
\newblock In \bibinfo{booktitle}{\emph{Public Key Cryptography {\textendash}
  {PKC} 2004}}. \bibinfo{publisher}{Springer Berlin Heidelberg},
  \bibinfo{pages}{425--438}.
\newblock
\urldef\tempurl%
\url{https://doi.org/10.1007/978-3-540-24632-9_31}
\showDOI{\tempurl}


\bibitem[\protect\citeauthoryear{Arcolezi, Couchot, Al~Bouna, and
  Xiao}{Arcolezi et~al\mbox{.}}{2021}]%
        {Arcolezi2021_rs_fd}
\bibfield{author}{\bibinfo{person}{H\'{e}ber~H. Arcolezi},
  \bibinfo{person}{Jean-Fran\c{c}ois Couchot}, \bibinfo{person}{Bechara
  Al~Bouna}, {and} \bibinfo{person}{Xiaokui Xiao}.}
  \bibinfo{year}{2021}\natexlab{}.
\newblock \showarticletitle{Random Sampling Plus Fake Data: Multidimensional
  Frequency Estimates With Local Differential Privacy}. In
  \bibinfo{booktitle}{\emph{Proceedings of the 30th ACM International
  Conference on Information \& Knowledge Management}} (Virtual Event,
  Queensland, Australia) \emph{(\bibinfo{series}{CIKM '21})}.
  \bibinfo{publisher}{Association for Computing Machinery},
  \bibinfo{address}{New York, NY, USA}, \bibinfo{pages}{47–57}.
\newblock
\showISBNx{9781450384469}
\urldef\tempurl%
\url{https://doi.org/10.1145/3459637.3482467}
\showDOI{\tempurl}


\bibitem[\protect\citeauthoryear{Arcolezi, Couchot, Bouna, and Xiao}{Arcolezi
  et~al\mbox{.}}{2022}]%
        {Arcolezi2021_allomfree}
\bibfield{author}{\bibinfo{person}{H{\'{e}}ber~H. Arcolezi},
  \bibinfo{person}{Jean-Fran{\c{c}}ois Couchot}, \bibinfo{person}{Bechara~Al
  Bouna}, {and} \bibinfo{person}{Xiaokui Xiao}.}
  \bibinfo{year}{2022}\natexlab{}.
\newblock \showarticletitle{Improving the utility of locally differentially
  private protocols for longitudinal and multidimensional frequency estimates}.
\newblock \bibinfo{journal}{\emph{Digital Communications and Networks}}
  (\bibinfo{year}{2022}).
\newblock
\showISSN{2352-8648}
\urldef\tempurl%
\url{https://doi.org/10.1016/j.dcan.2022.07.003}
\showDOI{\tempurl}


\bibitem[\protect\citeauthoryear{Cao, Jia, and Gong}{Cao et~al\mbox{.}}{2021}]%
        {cao2021data}
\bibfield{author}{\bibinfo{person}{Xiaoyu Cao}, \bibinfo{person}{Jinyuan Jia},
  {and} \bibinfo{person}{Neil~Zhenqiang Gong}.}
  \bibinfo{year}{2021}\natexlab{}.
\newblock \showarticletitle{Data Poisoning Attacks to Local Differential
  Privacy Protocols}. In \bibinfo{booktitle}{\emph{30th {USENIX} Security
  Symposium ({USENIX} Security 21)}}. \bibinfo{publisher}{USENIX Association},
  \bibinfo{pages}{947--964}.
\newblock
\showISBNx{978-1-939133-24-3}


\bibitem[\protect\citeauthoryear{Chatzikokolakis, Andr{\'e}s, Bordenabe, and
  Palamidessi}{Chatzikokolakis et~al\mbox{.}}{2013}]%
        {Chatzikokolakis2013}
\bibfield{author}{\bibinfo{person}{Konstantinos Chatzikokolakis},
  \bibinfo{person}{Miguel~E. Andr{\'e}s}, \bibinfo{person}{Nicol{\'a}s~Emilio
  Bordenabe}, {and} \bibinfo{person}{Catuscia Palamidessi}.}
  \bibinfo{year}{2013}\natexlab{}.
\newblock \showarticletitle{Broadening the Scope of Differential Privacy Using
  Metrics}. In \bibinfo{booktitle}{\emph{Privacy Enhancing Technologies}},
  \bibfield{editor}{\bibinfo{person}{Emiliano De~Cristofaro} {and}
  \bibinfo{person}{Matthew Wright}} (Eds.). \bibinfo{publisher}{Springer Berlin
  Heidelberg}, \bibinfo{address}{Berlin, Heidelberg}, \bibinfo{pages}{82--102}.
\newblock
\showISBNx{978-3-642-39077-7}
\urldef\tempurl%
\url{https://doi.org/10.1007/978-3-642-39077-7_5}
\showDOI{\tempurl}


\bibitem[\protect\citeauthoryear{Chatzikokolakis, Cherubin, Palamidessi, and
  Troncoso}{Chatzikokolakis et~al\mbox{.}}{2023}]%
        {chatzikokolakis2020bayes}
\bibfield{author}{\bibinfo{person}{Konstantinos Chatzikokolakis},
  \bibinfo{person}{Giovanni Cherubin}, \bibinfo{person}{Catuscia Palamidessi},
  {and} \bibinfo{person}{Carmela Troncoso}.} \bibinfo{year}{2023}\natexlab{}.
\newblock \showarticletitle{Bayes Security: A Not So Average Metric}. In
  \bibinfo{booktitle}{\emph{2023 IEEE 36th Computer Security Foundations
  Symposium (CSF)}}. \bibinfo{publisher}{IEEE Computer Society},
  \bibinfo{pages}{159--177}.
\newblock
\showISSN{2374-8303}
\urldef\tempurl%
\url{https://doi.org/10.1109/CSF57540.2023.00011}
\showDOI{\tempurl}


\bibitem[\protect\citeauthoryear{Chen and Guestrin}{Chen and Guestrin}{2016}]%
        {XGBoost}
\bibfield{author}{\bibinfo{person}{Tianqi Chen} {and} \bibinfo{person}{Carlos
  Guestrin}.} \bibinfo{year}{2016}\natexlab{}.
\newblock \showarticletitle{{XGBoost: A Scalable Tree Boosting System}}. In
  \bibinfo{booktitle}{\emph{Proceedings of the 22nd {ACM} {SIGKDD}
  International Conference on Knowledge Discovery and Data Mining}}.
  \bibinfo{publisher}{{ACM}}.
\newblock
\urldef\tempurl%
\url{https://doi.org/10.1145/2939672.2939785}
\showDOI{\tempurl}


\bibitem[\protect\citeauthoryear{Cheu, Smith, and Ullman}{Cheu
  et~al\mbox{.}}{2021}]%
        {Cheu2021}
\bibfield{author}{\bibinfo{person}{Albert Cheu}, \bibinfo{person}{Adam Smith},
  {and} \bibinfo{person}{Jonathan Ullman}.} \bibinfo{year}{2021}\natexlab{}.
\newblock \showarticletitle{Manipulation Attacks in Local Differential
  Privacy}. In \bibinfo{booktitle}{\emph{2021 {IEEE} Symposium on Security and
  Privacy ({SP})}}. \bibinfo{publisher}{{IEEE}}.
\newblock
\urldef\tempurl%
\url{https://doi.org/10.1109/sp40001.2021.00001}
\showDOI{\tempurl}


\bibitem[\protect\citeauthoryear{Cohen}{Cohen}{2022}]%
        {Cohen2022}
\bibfield{author}{\bibinfo{person}{Aloni Cohen}.}
  \bibinfo{year}{2022}\natexlab{}.
\newblock \showarticletitle{Attacks on Deidentification{\textquoteright}s
  Defenses}. In \bibinfo{booktitle}{\emph{31st USENIX Security Symposium
  (USENIX Security 22)}}. \bibinfo{publisher}{USENIX Association},
  \bibinfo{address}{Boston, MA}, \bibinfo{pages}{1469--1486}.
\newblock
\showISBNx{978-1-939133-31-1}


\bibitem[\protect\citeauthoryear{Cormode, Maddock, and Maple}{Cormode
  et~al\mbox{.}}{2021}]%
        {Cormode2021}
\bibfield{author}{\bibinfo{person}{Graham Cormode}, \bibinfo{person}{Samuel
  Maddock}, {and} \bibinfo{person}{Carsten Maple}.}
  \bibinfo{year}{2021}\natexlab{}.
\newblock \showarticletitle{Frequency estimation under local differential
  privacy}.
\newblock \bibinfo{journal}{\emph{Proceedings of the {VLDB} Endowment}}
  \bibinfo{volume}{14}, \bibinfo{number}{11} (\bibinfo{date}{July}
  \bibinfo{year}{2021}), \bibinfo{pages}{2046--2058}.
\newblock
\urldef\tempurl%
\url{https://doi.org/10.14778/3476249.3476261}
\showDOI{\tempurl}


\bibitem[\protect\citeauthoryear{Desfontaines}{Desfontaines}{2021}]%
        {desfontaines_dp_real_world}
\bibfield{author}{\bibinfo{person}{Damien Desfontaines}.}
  \bibinfo{year}{2021}\natexlab{}.
\newblock \bibinfo{title}{A list of real-world uses of differential privacy}.
\newblock
\newblock
\newblock
\shownote{Available online:
  \url{https://desfontain.es/privacy/real-world-differential-privacy.html}
  (accessed on 27 May 2022).}


\bibitem[\protect\citeauthoryear{Ding, Kulkarni, and Yekhanin}{Ding
  et~al\mbox{.}}{2017}]%
        {microsoft}
\bibfield{author}{\bibinfo{person}{Bolin Ding}, \bibinfo{person}{Janardhan
  Kulkarni}, {and} \bibinfo{person}{Sergey Yekhanin}.}
  \bibinfo{year}{2017}\natexlab{}.
\newblock \showarticletitle{Collecting Telemetry Data Privately}.
\newblock In \bibinfo{booktitle}{\emph{Advances in Neural Information
  Processing Systems 30}}, \bibfield{editor}{\bibinfo{person}{I.~Guyon},
  \bibinfo{person}{U.~V. Luxburg}, \bibinfo{person}{S.~Bengio},
  \bibinfo{person}{H.~Wallach}, \bibinfo{person}{R.~Fergus},
  \bibinfo{person}{S.~Vishwanathan}, {and} \bibinfo{person}{R.~Garnett}}
  (Eds.). \bibinfo{publisher}{Curran Associates, Inc.},
  \bibinfo{pages}{3571--3580}.
\newblock


\bibitem[\protect\citeauthoryear{Ding, Hardt, Miller, and Schmidt}{Ding
  et~al\mbox{.}}{2021}]%
        {ding2021retiring}
\bibfield{author}{\bibinfo{person}{Frances Ding}, \bibinfo{person}{Moritz
  Hardt}, \bibinfo{person}{John Miller}, {and} \bibinfo{person}{Ludwig
  Schmidt}.} \bibinfo{year}{2021}\natexlab{}.
\newblock \showarticletitle{Retiring adult: New datasets for fair machine
  learning}.
\newblock \bibinfo{journal}{\emph{Advances in Neural Information Processing
  Systems}}  \bibinfo{volume}{34} (\bibinfo{year}{2021}).
\newblock


\bibitem[\protect\citeauthoryear{Domingo-Ferrer and Soria-Comas}{Domingo-Ferrer
  and Soria-Comas}{2018}]%
        {domingo2018connecting}
\bibfield{author}{\bibinfo{person}{Josep Domingo-Ferrer} {and}
  \bibinfo{person}{Jordi Soria-Comas}.} \bibinfo{year}{2018}\natexlab{}.
\newblock \showarticletitle{Connecting randomized response, post-randomization,
  differential privacy and t-closeness via deniability and permutation}.
\newblock \bibinfo{journal}{\emph{arXiv preprint arXiv:1803.02139}}
  (\bibinfo{year}{2018}).
\newblock


\bibitem[\protect\citeauthoryear{Dua and Graff}{Dua and Graff}{2017}]%
        {uci}
\bibfield{author}{\bibinfo{person}{Dheeru Dua} {and} \bibinfo{person}{Casey
  Graff}.} \bibinfo{year}{2017}\natexlab{}.
\newblock \bibinfo{title}{{UCI} Machine Learning Repository}.
\newblock
\newblock
\newblock
\shownote{Available online: \url{http://archive.ics.uci.edu/ml} (accessed on 12
  January 2023).}


\bibitem[\protect\citeauthoryear{Duchi, Jordan, and Wainwright}{Duchi
  et~al\mbox{.}}{2013}]%
        {Duchi2013}
\bibfield{author}{\bibinfo{person}{John~C. Duchi}, \bibinfo{person}{Michael~I.
  Jordan}, {and} \bibinfo{person}{Martin~J. Wainwright}.}
  \bibinfo{year}{2013}\natexlab{}.
\newblock \showarticletitle{Local Privacy and Statistical Minimax Rates}. In
  \bibinfo{booktitle}{\emph{2013 {IEEE} 54th Annual Symposium on Foundations of
  Computer Science}}. \bibinfo{publisher}{{IEEE}}.
\newblock
\urldef\tempurl%
\url{https://doi.org/10.1109/focs.2013.53}
\showDOI{\tempurl}


\bibitem[\protect\citeauthoryear{Dwork}{Dwork}{2006}]%
        {Dwork2006DP}
\bibfield{author}{\bibinfo{person}{Cynthia Dwork}.}
  \bibinfo{year}{2006}\natexlab{}.
\newblock \showarticletitle{Differential Privacy}. In
  \bibinfo{booktitle}{\emph{Automata, Languages and Programming}},
  \bibfield{editor}{\bibinfo{person}{Michele Bugliesi}, \bibinfo{person}{Bart
  Preneel}, \bibinfo{person}{Vladimiro Sassone}, {and} \bibinfo{person}{Ingo
  Wegener}} (Eds.). \bibinfo{publisher}{Springer Berlin Heidelberg},
  \bibinfo{address}{Berlin, Heidelberg}, \bibinfo{pages}{1--12}.
\newblock
\showISBNx{978-3-540-35908-1}


\bibitem[\protect\citeauthoryear{Dwork, McSherry, Nissim, and Smith}{Dwork
  et~al\mbox{.}}{2006}]%
        {Dwork2006}
\bibfield{author}{\bibinfo{person}{Cynthia Dwork}, \bibinfo{person}{Frank
  McSherry}, \bibinfo{person}{Kobbi Nissim}, {and} \bibinfo{person}{Adam
  Smith}.} \bibinfo{year}{2006}\natexlab{}.
\newblock \showarticletitle{Calibrating Noise to Sensitivity in Private Data
  Analysis}.
\newblock In \bibinfo{booktitle}{\emph{Theory of Cryptography}}.
  \bibinfo{publisher}{Springer Berlin Heidelberg}, \bibinfo{pages}{265--284}.
\newblock
\urldef\tempurl%
\url{https://doi.org/10.1007/11681878\_14}
\showDOI{\tempurl}


\bibitem[\protect\citeauthoryear{Dwork, Roth, et~al\mbox{.}}{Dwork
  et~al\mbox{.}}{2014}]%
        {dwork2014algorithmic}
\bibfield{author}{\bibinfo{person}{Cynthia Dwork}, \bibinfo{person}{Aaron
  Roth}, {et~al\mbox{.}}} \bibinfo{year}{2014}\natexlab{}.
\newblock \showarticletitle{The algorithmic foundations of differential
  privacy}.
\newblock \bibinfo{journal}{\emph{Foundations and Trends{\textregistered} in
  Theoretical Computer Science}} \bibinfo{volume}{9}, \bibinfo{number}{3--4}
  (\bibinfo{year}{2014}), \bibinfo{pages}{211--407}.
\newblock


\bibitem[\protect\citeauthoryear{Emre~Gursoy, Liu, Chow, Truex, and
  Wei}{Emre~Gursoy et~al\mbox{.}}{2022}]%
        {Gursoy2022}
\bibfield{author}{\bibinfo{person}{M. Emre~Gursoy}, \bibinfo{person}{Ling Liu},
  \bibinfo{person}{Ka-Ho Chow}, \bibinfo{person}{Stacey Truex}, {and}
  \bibinfo{person}{Wenqi Wei}.} \bibinfo{year}{2022}\natexlab{}.
\newblock \showarticletitle{An Adversarial Approach to Protocol Analysis and
  Selection in Local Differential Privacy}.
\newblock \bibinfo{journal}{\emph{IEEE Transactions on Information Forensics
  and Security}}  \bibinfo{volume}{17} (\bibinfo{year}{2022}),
  \bibinfo{pages}{1785--1799}.
\newblock
\urldef\tempurl%
\url{https://doi.org/10.1109/TIFS.2022.3170242}
\showDOI{\tempurl}


\bibitem[\protect\citeauthoryear{Erlingsson, Pihur, and Korolova}{Erlingsson
  et~al\mbox{.}}{2014}]%
        {rappor}
\bibfield{author}{\bibinfo{person}{\'Ulfar Erlingsson}, \bibinfo{person}{Vasyl
  Pihur}, {and} \bibinfo{person}{Aleksandra Korolova}.}
  \bibinfo{year}{2014}\natexlab{}.
\newblock \showarticletitle{{RAPPOR}: Randomized Aggregatable
  Privacy-Preserving Ordinal Response}. In
  \bibinfo{booktitle}{\emph{Proceedings of the 2014 ACM SIGSAC Conference on
  Computer and Communications Security}} (Scottsdale, Arizona, USA).
  \bibinfo{publisher}{ACM}, \bibinfo{address}{New York, NY, USA},
  \bibinfo{pages}{1054--1067}.
\newblock
\urldef\tempurl%
\url{https://doi.org/10.1145/2660267.2660348}
\showDOI{\tempurl}


\bibitem[\protect\citeauthoryear{Gadotti, Houssiau, Annamalai, and
  de~Montjoye}{Gadotti et~al\mbox{.}}{2022}]%
        {Gadotti2022}
\bibfield{author}{\bibinfo{person}{Andrea Gadotti}, \bibinfo{person}{Florimond
  Houssiau}, \bibinfo{person}{Meenatchi Sundaram Muthu~Selva Annamalai}, {and}
  \bibinfo{person}{Yves-Alexandre de Montjoye}.}
  \bibinfo{year}{2022}\natexlab{}.
\newblock \showarticletitle{Pool Inference Attacks on Local Differential
  Privacy: Quantifying the Privacy Guarantees of Apple{\textquoteright}s Count
  Mean Sketch in Practice}. In \bibinfo{booktitle}{\emph{31st USENIX Security
  Symposium (USENIX Security 22)}}. \bibinfo{publisher}{USENIX Association},
  \bibinfo{address}{Boston, MA}, \bibinfo{pages}{501--518}.
\newblock
\showISBNx{978-1-939133-31-1}


\bibitem[\protect\citeauthoryear{Gambs, Killijian, and {Núñez del Prado
  Cortez}}{Gambs et~al\mbox{.}}{2014}]%
        {Gambs2014}
\bibfield{author}{\bibinfo{person}{Sébastien Gambs},
  \bibinfo{person}{Marc-Olivier Killijian}, {and} \bibinfo{person}{Miguel
  {Núñez del Prado Cortez}}.} \bibinfo{year}{2014}\natexlab{}.
\newblock \showarticletitle{De-anonymization attack on geolocated data}.
\newblock \bibinfo{journal}{\emph{J. Comput. System Sci.}}
  \bibinfo{volume}{80}, \bibinfo{number}{8} (\bibinfo{year}{2014}),
  \bibinfo{pages}{1597--1614}.
\newblock
\showISSN{0022-0000}
\urldef\tempurl%
\url{https://doi.org/10.1016/j.jcss.2014.04.024}
\showDOI{\tempurl}


\bibitem[\protect\citeauthoryear{Kairouz, Bonawitz, and Ramage}{Kairouz
  et~al\mbox{.}}{2016a}]%
        {kairouz2016discrete}
\bibfield{author}{\bibinfo{person}{Peter Kairouz}, \bibinfo{person}{Keith
  Bonawitz}, {and} \bibinfo{person}{Daniel Ramage}.}
  \bibinfo{year}{2016}\natexlab{a}.
\newblock \showarticletitle{Discrete distribution estimation under local
  privacy}. In \bibinfo{booktitle}{\emph{International Conference on Machine
  Learning}}. PMLR, \bibinfo{pages}{2436--2444}.
\newblock


\bibitem[\protect\citeauthoryear{Kairouz, Oh, and Viswanath}{Kairouz
  et~al\mbox{.}}{2016b}]%
        {kairouz2016extremal}
\bibfield{author}{\bibinfo{person}{Peter Kairouz}, \bibinfo{person}{Sewoong
  Oh}, {and} \bibinfo{person}{Pramod Viswanath}.}
  \bibinfo{year}{2016}\natexlab{b}.
\newblock \showarticletitle{Extremal mechanisms for local differential
  privacy}.
\newblock \bibinfo{journal}{\emph{The Journal of Machine Learning Research}}
  \bibinfo{volume}{17}, \bibinfo{number}{1} (\bibinfo{year}{2016}),
  \bibinfo{pages}{492--542}.
\newblock


\bibitem[\protect\citeauthoryear{Kasiviswanathan, Lee, Nissim, Raskhodnikova,
  and Smith}{Kasiviswanathan et~al\mbox{.}}{2008}]%
        {first_ldp}
\bibfield{author}{\bibinfo{person}{Shiva~Prasad Kasiviswanathan},
  \bibinfo{person}{Homin~K. Lee}, \bibinfo{person}{Kobbi Nissim},
  \bibinfo{person}{Sofya Raskhodnikova}, {and} \bibinfo{person}{Adam Smith}.}
  \bibinfo{year}{2008}\natexlab{}.
\newblock \showarticletitle{What Can We Learn Privately?}. In
  \bibinfo{booktitle}{\emph{2008 49th Annual IEEE Symposium on Foundations of
  Computer Science}}. \bibinfo{pages}{531--540}.
\newblock
\urldef\tempurl%
\url{https://doi.org/10.1109/FOCS.2008.27}
\showDOI{\tempurl}


\bibitem[\protect\citeauthoryear{Kato, Cao, and Yoshikawa}{Kato
  et~al\mbox{.}}{2021}]%
        {Kato2021}
\bibfield{author}{\bibinfo{person}{Fumiyuki Kato}, \bibinfo{person}{Yang Cao},
  {and} \bibinfo{person}{Masatoshi Yoshikawa}.}
  \bibinfo{year}{2021}\natexlab{}.
\newblock \showarticletitle{Preventing Manipulation Attack in Local
  Differential Privacy Using Verifiable Randomization Mechanism}.
\newblock In \bibinfo{booktitle}{\emph{Data and Applications Security and
  Privacy {XXXV}}}. \bibinfo{publisher}{Springer International Publishing},
  \bibinfo{pages}{43--60}.
\newblock
\urldef\tempurl%
\url{https://doi.org/10.1007/978-3-030-81242-3_3}
\showDOI{\tempurl}


\bibitem[\protect\citeauthoryear{Li, Li, and Venkatasubramanian}{Li
  et~al\mbox{.}}{2007}]%
        {Li2007}
\bibfield{author}{\bibinfo{person}{Ninghui Li}, \bibinfo{person}{Tiancheng Li},
  {and} \bibinfo{person}{Suresh Venkatasubramanian}.}
  \bibinfo{year}{2007}\natexlab{}.
\newblock \showarticletitle{t-Closeness: Privacy Beyond k-Anonymity and
  l-Diversity}. In \bibinfo{booktitle}{\emph{2007 {IEEE} 23rd International
  Conference on Data Engineering}}. \bibinfo{publisher}{{IEEE}}.
\newblock
\urldef\tempurl%
\url{https://doi.org/10.1109/icde.2007.367856}
\showDOI{\tempurl}


\bibitem[\protect\citeauthoryear{Li, Qardaji, and Su}{Li et~al\mbox{.}}{2012}]%
        {Li2012}
\bibfield{author}{\bibinfo{person}{Ninghui Li}, \bibinfo{person}{Wahbeh
  Qardaji}, {and} \bibinfo{person}{Dong Su}.} \bibinfo{year}{2012}\natexlab{}.
\newblock \showarticletitle{On sampling, anonymization, and differential
  privacy or, k-anonymization meets differential privacy}. In
  \bibinfo{booktitle}{\emph{Proceedings of the 7th {ACM} Symposium on
  Information, Computer and Communications Security - {ASIACCS}
  {\textquotesingle}12}}. \bibinfo{publisher}{{ACM} Press}.
\newblock
\urldef\tempurl%
\url{https://doi.org/10.1145/2414456.2414474}
\showDOI{\tempurl}


\bibitem[\protect\citeauthoryear{Li, Gong, Li, Sun, and Li}{Li
  et~al\mbox{.}}{2022}]%
        {li2022fine}
\bibfield{author}{\bibinfo{person}{Xiaoguang Li},
  \bibinfo{person}{Neil~Zhenqiang Gong}, \bibinfo{person}{Ninghui Li},
  \bibinfo{person}{Wenhai Sun}, {and} \bibinfo{person}{Hui Li}.}
  \bibinfo{year}{2022}\natexlab{}.
\newblock \showarticletitle{Fine-grained Poisoning Attacks to Local
  Differential Privacy Protocols for Mean and Variance Estimation}.
\newblock \bibinfo{journal}{\emph{arXiv preprint arXiv:2205.11782}}
  (\bibinfo{year}{2022}).
\newblock


\bibitem[\protect\citeauthoryear{Machanavajjhala, Gehrke, Kifer, and
  Venkitasubramaniam}{Machanavajjhala et~al\mbox{.}}{2006}]%
        {Machanavajjhala2006}
\bibfield{author}{\bibinfo{person}{A. Machanavajjhala}, \bibinfo{person}{J.
  Gehrke}, \bibinfo{person}{D. Kifer}, {and} \bibinfo{person}{M.
  Venkitasubramaniam}.} \bibinfo{year}{2006}\natexlab{}.
\newblock \showarticletitle{L-diversity: privacy beyond k-anonymity}. In
  \bibinfo{booktitle}{\emph{22nd International Conference on Data Engineering
  ({ICDE}{\textquotesingle}06)}}. \bibinfo{publisher}{{IEEE}}.
\newblock
\urldef\tempurl%
\url{https://doi.org/10.1109/icde.2006.1}
\showDOI{\tempurl}


\bibitem[\protect\citeauthoryear{Murakami, Kanemura, and Hino}{Murakami
  et~al\mbox{.}}{2017}]%
        {Murakami2017}
\bibfield{author}{\bibinfo{person}{Takao Murakami}, \bibinfo{person}{Atsunori
  Kanemura}, {and} \bibinfo{person}{Hideitsu Hino}.}
  \bibinfo{year}{2017}\natexlab{}.
\newblock \showarticletitle{Group Sparsity Tensor Factorization for
  Re-Identification of Open Mobility Traces}.
\newblock \bibinfo{journal}{\emph{IEEE Transactions on Information Forensics
  and Security}} \bibinfo{volume}{12}, \bibinfo{number}{3}
  (\bibinfo{year}{2017}), \bibinfo{pages}{689--704}.
\newblock
\urldef\tempurl%
\url{https://doi.org/10.1109/TIFS.2016.2631952}
\showDOI{\tempurl}


\bibitem[\protect\citeauthoryear{Murakami and Takahashi}{Murakami and
  Takahashi}{2021}]%
        {Murakami2021}
\bibfield{author}{\bibinfo{person}{Takao Murakami} {and} \bibinfo{person}{Kenta
  Takahashi}.} \bibinfo{year}{2021}\natexlab{}.
\newblock \showarticletitle{Toward Evaluating Re-identification Risks in the
  Local Privacy Model}.
\newblock \bibinfo{journal}{\emph{Transactions on Data Privacy}}
  \bibinfo{volume}{14}, \bibinfo{number}{3} (\bibinfo{year}{2021}),
  \bibinfo{pages}{79--116}.
\newblock


\bibitem[\protect\citeauthoryear{Narayanan and Shmatikov}{Narayanan and
  Shmatikov}{2008}]%
        {Narayanan2008}
\bibfield{author}{\bibinfo{person}{Arvind Narayanan} {and}
  \bibinfo{person}{Vitaly Shmatikov}.} \bibinfo{year}{2008}\natexlab{}.
\newblock \showarticletitle{Robust De-anonymization of Large Sparse Datasets}.
  In \bibinfo{booktitle}{\emph{2008 IEEE Symposium on Security and Privacy (sp
  2008)}}. \bibinfo{pages}{111--125}.
\newblock
\urldef\tempurl%
\url{https://doi.org/10.1109/SP.2008.33}
\showDOI{\tempurl}


\bibitem[\protect\citeauthoryear{Nguy{\^e}n, Xiao, Yang, Hui, Shin, and
  Shin}{Nguy{\^e}n et~al\mbox{.}}{2016}]%
        {xiao2}
\bibfield{author}{\bibinfo{person}{Th{\^o}ng~T Nguy{\^e}n},
  \bibinfo{person}{Xiaokui Xiao}, \bibinfo{person}{Yin Yang},
  \bibinfo{person}{Siu~Cheung Hui}, \bibinfo{person}{Hyejin Shin}, {and}
  \bibinfo{person}{Junbum Shin}.} \bibinfo{year}{2016}\natexlab{}.
\newblock \showarticletitle{Collecting and analyzing data from smart device
  users with local differential privacy}.
\newblock \bibinfo{journal}{\emph{arXiv preprint arXiv:1606.05053}}
  (\bibinfo{year}{2016}).
\newblock


\bibitem[\protect\citeauthoryear{Ren, Yu, Yu, Yang, Member, Yang, Mccann, Yu,
  and Fellow}{Ren et~al\mbox{.}}{2018}]%
        {Ren2018}
\bibfield{author}{\bibinfo{person}{Xuebin Ren}, \bibinfo{person}{Chia-mu Yu},
  \bibinfo{person}{Weiren Yu}, \bibinfo{person}{Shusen Yang},
  \bibinfo{person}{Senior Member}, \bibinfo{person}{Xinyu Yang},
  \bibinfo{person}{Julie~A Mccann}, \bibinfo{person}{Philip~S Yu}, {and}
  \bibinfo{person}{Life Fellow}.} \bibinfo{year}{2018}\natexlab{}.
\newblock \showarticletitle{{LoPub : High-Dimensional Crowdsourced Data}}.
\newblock  \bibinfo{volume}{13}, \bibinfo{number}{9} (\bibinfo{year}{2018}),
  \bibinfo{pages}{2151--2166}.
\newblock
\urldef\tempurl%
\url{https://doi.org/10.1109/TIFS.2018.2812146}
\showDOI{\tempurl}


\bibitem[\protect\citeauthoryear{Rogers, Subramaniam, Peng, Durfee, Lee,
  Kancha, Sahay, and Ahammad}{Rogers et~al\mbox{.}}{2021}]%
        {linkedin}
\bibfield{author}{\bibinfo{person}{Ryan Rogers}, \bibinfo{person}{Subbu
  Subramaniam}, \bibinfo{person}{Sean Peng}, \bibinfo{person}{David Durfee},
  \bibinfo{person}{Seunghyun Lee}, \bibinfo{person}{Santosh~Kumar Kancha},
  \bibinfo{person}{Shraddha Sahay}, {and} \bibinfo{person}{Parvez Ahammad}.}
  \bibinfo{year}{2021}\natexlab{}.
\newblock \showarticletitle{LinkedIn’s Audience Engagements {API}: A Privacy
  Preserving Data Analytics System at Scale}.
\newblock \bibinfo{journal}{\emph{Journal of Privacy and Confidentiality}}
  \bibinfo{volume}{11}, \bibinfo{number}{3} (\bibinfo{date}{Dec.}
  \bibinfo{year}{2021}).
\newblock
\urldef\tempurl%
\url{https://doi.org/10.29012/jpc.782}
\showDOI{\tempurl}


\bibitem[\protect\citeauthoryear{Samarati}{Samarati}{2001}]%
        {Samarati2001}
\bibfield{author}{\bibinfo{person}{P. Samarati}.}
  \bibinfo{year}{2001}\natexlab{}.
\newblock \showarticletitle{Protecting respondents identities in microdata
  release}.
\newblock \bibinfo{journal}{\emph{{IEEE}Transactions on Knowledge and Data
  Engineering}} \bibinfo{volume}{13}, \bibinfo{number}{6}
  (\bibinfo{year}{2001}), \bibinfo{pages}{1010--1027}.
\newblock
\urldef\tempurl%
\url{https://doi.org/10.1109/69.971193}
\showDOI{\tempurl}


\bibitem[\protect\citeauthoryear{Samarati and Sweeney}{Samarati and
  Sweeney}{1998}]%
        {samarati1998protecting}
\bibfield{author}{\bibinfo{person}{Pierangela Samarati} {and}
  \bibinfo{person}{Latanya Sweeney}.} \bibinfo{year}{1998}\natexlab{}.
\newblock \showarticletitle{Protecting privacy when disclosing information:
  k-anonymity and its enforcement through generalization and suppression}.
\newblock  (\bibinfo{year}{1998}).
\newblock


\bibitem[\protect\citeauthoryear{Sweeney}{Sweeney}{2002}]%
        {SWEENEY2002}
\bibfield{author}{\bibinfo{person}{Latanya Sweeney}.}
  \bibinfo{year}{2002}\natexlab{}.
\newblock \showarticletitle{k-Anonymity: A Model for Protecting Privacy}.
\newblock \bibinfo{journal}{\emph{International Journal of Uncertainty,
  Fuzziness and Knowledge-Based Systems}} \bibinfo{volume}{10},
  \bibinfo{number}{05} (\bibinfo{date}{Oct.} \bibinfo{year}{2002}),
  \bibinfo{pages}{557--570}.
\newblock
\urldef\tempurl%
\url{https://doi.org/10.1142/s0218488502001648}
\showDOI{\tempurl}


\bibitem[\protect\citeauthoryear{Sweeney}{Sweeney}{2015}]%
        {sweeney2015only}
\bibfield{author}{\bibinfo{person}{Latanya Sweeney}.}
  \bibinfo{year}{2015}\natexlab{}.
\newblock \showarticletitle{Only you, your doctor, and many others may know}.
\newblock \bibinfo{journal}{\emph{Technology Science}}
  \bibinfo{volume}{2015092903}, \bibinfo{number}{9} (\bibinfo{year}{2015}),
  \bibinfo{pages}{29}.
\newblock


\bibitem[\protect\citeauthoryear{Tang, Korolova, Bai, Wang, and Wang}{Tang
  et~al\mbox{.}}{2017}]%
        {tang2017privacy}
\bibfield{author}{\bibinfo{person}{Jun Tang}, \bibinfo{person}{Aleksandra
  Korolova}, \bibinfo{person}{Xiaolong Bai}, \bibinfo{person}{Xueqiang Wang},
  {and} \bibinfo{person}{Xiaofeng Wang}.} \bibinfo{year}{2017}\natexlab{}.
\newblock \showarticletitle{Privacy loss in apple's implementation of
  differential privacy on macos 10.12}.
\newblock \bibinfo{journal}{\emph{arXiv preprint arXiv:1709.02753}}
  (\bibinfo{year}{2017}).
\newblock


\bibitem[\protect\citeauthoryear{Team}{Team}{2017}]%
        {apple}
\bibfield{author}{\bibinfo{person}{Apple Differential~Privacy Team}.}
  \bibinfo{year}{2017}\natexlab{}.
\newblock \bibinfo{title}{Learning with privacy at scale}.
\newblock
  \bibinfo{howpublished}{\url{https://docs-assets.developer.apple.com/ml-research/papers/learning-with-privacy-at-scale.pdf}}.
\newblock
\newblock
\shownote{Online; accessed 11 December 2021.}


\bibitem[\protect\citeauthoryear{Varma, Chauhan, and Singh}{Varma
  et~al\mbox{.}}{2022}]%
        {Varma2022}
\bibfield{author}{\bibinfo{person}{Gatha Varma}, \bibinfo{person}{Ritu
  Chauhan}, {and} \bibinfo{person}{Dhananjay Singh}.}
  \bibinfo{year}{2022}\natexlab{}.
\newblock \showarticletitle{Sarve: synthetic data and local differential
  privacy for private frequency estimation}.
\newblock \bibinfo{journal}{\emph{Cybersecurity}} \bibinfo{volume}{5},
  \bibinfo{number}{1} (\bibinfo{year}{2022}), \bibinfo{pages}{1--20}.
\newblock
\urldef\tempurl%
\url{https://doi.org/10.1186/s42400-022-00129-6}
\showDOI{\tempurl}


\bibitem[\protect\citeauthoryear{Wang, Xiao, Yang, Zhao, Hui, Shin, Shin, and
  Yu}{Wang et~al\mbox{.}}{2019}]%
        {wang2019}
\bibfield{author}{\bibinfo{person}{Ning Wang}, \bibinfo{person}{Xiaokui Xiao},
  \bibinfo{person}{Yin Yang}, \bibinfo{person}{Jun Zhao},
  \bibinfo{person}{Siu~Cheung Hui}, \bibinfo{person}{Hyejin Shin},
  \bibinfo{person}{Junbum Shin}, {and} \bibinfo{person}{Ge Yu}.}
  \bibinfo{year}{2019}\natexlab{}.
\newblock \showarticletitle{Collecting and Analyzing Multidimensional Data with
  Local Differential Privacy}. In \bibinfo{booktitle}{\emph{2019 {IEEE} 35th
  International Conference on Data Engineering ({ICDE})}}.
  \bibinfo{publisher}{{IEEE}}.
\newblock
\urldef\tempurl%
\url{https://doi.org/10.1109/icde.2019.00063}
\showDOI{\tempurl}


\bibitem[\protect\citeauthoryear{Wang, Huang, Wang, Nie, Xu, Yang, Li, and
  Qiao}{Wang et~al\mbox{.}}{2016}]%
        {wang2016mutual}
\bibfield{author}{\bibinfo{person}{Shaowei Wang}, \bibinfo{person}{Liusheng
  Huang}, \bibinfo{person}{Pengzhan Wang}, \bibinfo{person}{Yiwen Nie},
  \bibinfo{person}{Hongli Xu}, \bibinfo{person}{Wei Yang},
  \bibinfo{person}{Xiang-Yang Li}, {and} \bibinfo{person}{Chunming Qiao}.}
  \bibinfo{year}{2016}\natexlab{}.
\newblock \showarticletitle{Mutual information optimally local private discrete
  distribution estimation}.
\newblock \bibinfo{journal}{\emph{arXiv preprint arXiv:1607.08025}}
  (\bibinfo{year}{2016}).
\newblock


\bibitem[\protect\citeauthoryear{Wang, Nie, Wang, Xu, Yang, and Huang}{Wang
  et~al\mbox{.}}{2017b}]%
        {Wang2017}
\bibfield{author}{\bibinfo{person}{Shaowei Wang}, \bibinfo{person}{Yiwen Nie},
  \bibinfo{person}{Pengzhan Wang}, \bibinfo{person}{Hongli Xu},
  \bibinfo{person}{Wei Yang}, {and} \bibinfo{person}{Liusheng Huang}.}
  \bibinfo{year}{2017}\natexlab{b}.
\newblock \showarticletitle{Local private ordinal data distribution
  estimation}. In \bibinfo{booktitle}{\emph{{IEEE} {INFOCOM} 2017 - {IEEE}
  Conference on Computer Communications}}. \bibinfo{publisher}{{IEEE}}.
\newblock
\urldef\tempurl%
\url{https://doi.org/10.1109/infocom.2017.8056977}
\showDOI{\tempurl}


\bibitem[\protect\citeauthoryear{Wang, Blocki, Li, and Jha}{Wang
  et~al\mbox{.}}{2017a}]%
        {tianhao2017}
\bibfield{author}{\bibinfo{person}{Tianhao Wang}, \bibinfo{person}{Jeremiah
  Blocki}, \bibinfo{person}{Ninghui Li}, {and} \bibinfo{person}{Somesh Jha}.}
  \bibinfo{year}{2017}\natexlab{a}.
\newblock \showarticletitle{Locally Differentially Private Protocols for
  Frequency Estimation}. In \bibinfo{booktitle}{\emph{26th {USENIX} Security
  Symposium ({USENIX} Security 17)}}. \bibinfo{publisher}{{USENIX}
  Association}, \bibinfo{address}{Vancouver, BC}, \bibinfo{pages}{729--745}.
\newblock
\showISBNx{978-1-931971-40-9}


\bibitem[\protect\citeauthoryear{Wang, Li, and Jha}{Wang et~al\mbox{.}}{2018}]%
        {Wang2018}
\bibfield{author}{\bibinfo{person}{Tianhao Wang}, \bibinfo{person}{Ninghui Li},
  {and} \bibinfo{person}{Somesh Jha}.} \bibinfo{year}{2018}\natexlab{}.
\newblock \showarticletitle{Locally Differentially Private Frequent Itemset
  Mining}. In \bibinfo{booktitle}{\emph{2018 {IEEE} Symposium on Security and
  Privacy ({SP})}}. \bibinfo{publisher}{{IEEE}}.
\newblock
\urldef\tempurl%
\url{https://doi.org/10.1109/sp.2018.00035}
\showDOI{\tempurl}


\bibitem[\protect\citeauthoryear{Wang, Lopuhaa-Zwakenberg, Li, Skoric, and
  Li}{Wang et~al\mbox{.}}{2020}]%
        {Wang2020_post_process}
\bibfield{author}{\bibinfo{person}{Tianhao Wang}, \bibinfo{person}{Milan
  Lopuhaa-Zwakenberg}, \bibinfo{person}{Zitao Li}, \bibinfo{person}{Boris
  Skoric}, {and} \bibinfo{person}{Ninghui Li}.}
  \bibinfo{year}{2020}\natexlab{}.
\newblock \showarticletitle{Locally Differentially Private Frequency Estimation
  with Consistency}. In \bibinfo{booktitle}{\emph{Proceedings 2020 Network and
  Distributed System Security Symposium}}. \bibinfo{publisher}{Internet
  Society}.
\newblock
\urldef\tempurl%
\url{https://doi.org/10.14722/ndss.2020.24157}
\showDOI{\tempurl}


\bibitem[\protect\citeauthoryear{Warner}{Warner}{1965}]%
        {Warner1965}
\bibfield{author}{\bibinfo{person}{Stanley~L. Warner}.}
  \bibinfo{year}{1965}\natexlab{}.
\newblock \showarticletitle{Randomized Response: A Survey Technique for
  Eliminating Evasive Answer Bias}.
\newblock \bibinfo{journal}{\emph{J. Amer. Statist. Assoc.}}
  \bibinfo{volume}{60}, \bibinfo{number}{309} (\bibinfo{date}{March}
  \bibinfo{year}{1965}), \bibinfo{pages}{63--69}.
\newblock
\urldef\tempurl%
\url{https://doi.org/10.1080/01621459.1965.10480775}
\showDOI{\tempurl}


\bibitem[\protect\citeauthoryear{Wu, Cao, Jia, and Gong}{Wu
  et~al\mbox{.}}{2022}]%
        {wu2021poisoning}
\bibfield{author}{\bibinfo{person}{Yongji Wu}, \bibinfo{person}{Xiaoyu Cao},
  \bibinfo{person}{Jinyuan Jia}, {and} \bibinfo{person}{Neil~Zhenqiang Gong}.}
  \bibinfo{year}{2022}\natexlab{}.
\newblock \showarticletitle{Poisoning Attacks to Local Differential Privacy
  Protocols for {Key-Value} Data}. In \bibinfo{booktitle}{\emph{31st USENIX
  Security Symposium (USENIX Security 22)}}. \bibinfo{publisher}{USENIX
  Association}, \bibinfo{address}{Boston, MA}, \bibinfo{pages}{519--536}.
\newblock
\showISBNx{978-1-939133-31-1}


\bibitem[\protect\citeauthoryear{Ye and Barg}{Ye and Barg}{2018}]%
        {Min2018}
\bibfield{author}{\bibinfo{person}{Min Ye} {and} \bibinfo{person}{Alexander
  Barg}.} \bibinfo{year}{2018}\natexlab{}.
\newblock \showarticletitle{Optimal Schemes for Discrete Distribution
  Estimation Under Locally Differential Privacy}.
\newblock \bibinfo{journal}{\emph{IEEE Transactions on Information Theory}}
  \bibinfo{volume}{64}, \bibinfo{number}{8} (\bibinfo{year}{2018}),
  \bibinfo{pages}{5662--5676}.
\newblock
\urldef\tempurl%
\url{https://doi.org/10.1109/TIT.2018.2809790}
\showDOI{\tempurl}


\bibitem[\protect\citeauthoryear{Zhang, Wang, Li, He, and Chen}{Zhang
  et~al\mbox{.}}{2018}]%
        {Zhang2018}
\bibfield{author}{\bibinfo{person}{Zhikun Zhang}, \bibinfo{person}{Tianhao
  Wang}, \bibinfo{person}{Ninghui Li}, \bibinfo{person}{Shibo He}, {and}
  \bibinfo{person}{Jiming Chen}.} \bibinfo{year}{2018}\natexlab{}.
\newblock \showarticletitle{{CALM: Consistent adaptive local marginal for
  marginal release under local differential privacy}}.
\newblock \bibinfo{journal}{\emph{Proceedings of the ACM Conference on Computer
  and Communications Security}} (\bibinfo{year}{2018}),
  \bibinfo{pages}{212--229}.
\newblock
\showISBNx{9781450356930}
\showISSN{15437221}
\urldef\tempurl%
\url{https://doi.org/10.1145/3243734.3243742}
\showDOI{\tempurl}


\end{thebibliography}

\appendix

\section{RS+RFD with GRR} \label{appA:RSpFD_GRR}

Visually, Fig.~\ref{fig:prob_tree_rsrfd_grr} illustrates the probability tree of the RS+RFD[GRR] protocol (cf. Section~\ref{sub:rspfd_sol}).

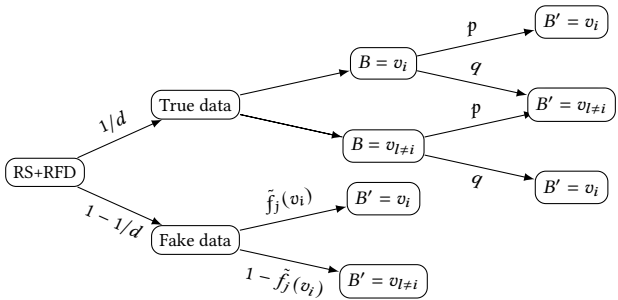
\begin{figure}[ht]
\centering
\tikzset{
  treenode/.style = {shape=rectangle, rounded corners,
                     draw,align=center,
                     top color=white},
  root/.style     = {treenode},
  env/.style      = {treenode},
  dummy/.style    = {circle,draw}
}
\tikzstyle{level 1}=[level distance=2cm, sibling distance=1.8cm]
\tikzstyle{level 2}=[level distance=2.5cm, sibling distance=1.1cm]

\begin{tikzpicture}
  [
    grow                    = right,
    edge from parent/.style = {draw, -latex},
    every node/.style       = {font=\footnotesize},
    sloped
  ]
  \node [root] {RS+RFD}
    child { node [env] {Fake data}
        child { node [env] {$B'=v_{l\neq i}$}
          edge from parent node [below] {$1-\tilde{f}_j(v_i)$} }
        child { node [env] {$B'=v_i$}
          edge from parent node [above] {$\tilde{f}_j(v_i)$} }
        edge from parent node [below] {$1-1/d$} }
    child { node [env] {True data}
        child { node [env] {$B=v_{l\neq i}$}  
         child { node [env] {$B'=v_i$}
            edge from parent node [below] {$q$}}
            child { node [env] {$B'=v_i$}
            edge from parent node [above] {$p$}}
            edge from parent node [below] {}
            edge from parent node [below] {}}
        child { node [env] {$B=v_i$} 
            child { node [env] {$B'=v_{l\neq i}$}
            edge from parent node [above] {$q$}}
            child { node [env] {$B'=v_i$}
            edge from parent node [above] {$p$}}
            edge from parent node [above] {}}
        edge from parent node [above] {$1/d$}};
\end{tikzpicture}
\caption{Probability tree for the RS+FD[GRR] protocol.} 
\label{fig:prob_tree_rsrfd_grr}
\end{figure}

\begin{theorem} \label{theo:est_grr} For $j\in[1,d]$, the estimation result $\hat{f}_{\textrm{GRR}}(v_i)$ in Eq.~\eqref{eq:est_rspfd_grr} is an unbiased estimation of $f (v_i)$ for any value $v_i \in A_j$.
\end{theorem}

\begin{proof}

\begin{equation*}
\begin{aligned}
    \mathbb{E}\left[ \hat{f}_{\textrm{GRR}}(v_i) \right] &= \mathbb{E}\left[ \frac{ d C(v_i) - n \left (q + (d - 1) \tilde{f}_{j}(v_i) \right )}{n(p-q)} \right] \\
    &= \frac{d }{n(p-q)} \mathbb{E}\left[ C(v_i) \right] -  \frac{ (d - 1) \tilde{f}_{j}(v_i) + q}{(p-q)}  \textrm{.}
\end{aligned}
\end{equation*}

On expectation, the number of times that $v_i$ is reported is:

\begin{equation*}
\begin{aligned}
    \mathbb{E}\left[ C(v_i) \right] &= \frac{1}{d} \left( p n f (v_i) + q (n - n f (v_i))\right)  + n \frac{(d-1) \tilde{f}_{j}(v_i)}{d}\\
    &= \frac{n}{d} \left(f (v_i)(p-q) + q   + (d-1) \tilde{f}_{j}(v_i) \right)   \textrm{.}
\end{aligned}
\end{equation*}

Therefore,

\begin{equation*}
    \mathbb{E}\left[ \hat{f}_{\textrm{GRR}}(v_i) \right] = f(v_i) \textrm{.}
\end{equation*}
\end{proof}

\begin{theorem} \label{theo:variance_grr} The variance of the estimation in Eq.~\eqref{eq:est_rspfd_grr} is:

\begin{equation}\label{var:rs+rfd_grr}
\begin{gathered}
    \operatorname{VAR}\left[ \hat{f}_{\textrm{GRR}}(v_i) \right] = \frac{d^2 \gamma (1-\gamma)}{n (p-q)^2} \textrm{, where} \\
    \gamma = \frac{1}{d} \left( q + f(v_i) (p-q) + (d-1) \tilde{f}_{j}(v_i) \right ) \textrm{.}
\end{gathered}
\end{equation}

\end{theorem}

\begin{proof}
Thanks to Eq.~\eqref{eq:est_rspfd_grr} we have

\begin{equation*}
\operatorname{VAR}\left[ \hat{f}_{\textrm{GRR}}(v_i) \right] = 
\frac{\operatorname{VAR}\left[ C(v_i) \right] d^2}{n^2 (p-q)^2}  \textrm{.}
\end{equation*}

Since $C(v_i)$ is the number of times value $v_i$ is observed, it can be defined as $C(v_i) = \sum_{u=1}^n X_u$ where $X_u$ is equal to 1 if the user $u$, 
$1 \le u \le n$ reports value $v_i$, and 0 otherwise. We thus have 
$
\operatorname{VAR}\left[ C(v_i) \right] 
= \sum_{u=1}^n \operatorname{VAR}\left[ X_u \right] 
= n \operatorname{VAR}\left[ X \right]$, since all the users are independent. \textcolor{black}{According to the probability tree in Fig.~\ref{fig:prob_tree_rsrfd_grr}, }
\[
\Pr\left[X = 1\right] = \Pr\left[X^2 = 1\right] = \gamma = \frac{1}{d} \left( q + f(v_i) (p-q) + (d-1) \tilde{f}_{j}(v_i) \right )  \textrm{.}
\]
We thus have $\operatorname{VAR}\left[ X \right]= \gamma - \gamma^2 = \gamma(1 - \gamma) $ and, finally,

\begin{equation*} \label{var:generic}
\operatorname{VAR}\left[ \hat{f}_{\textrm{GRR}}(v_i) \right] =
\frac{d^2 \gamma (1-\gamma)}{n (p-q)^2}.
\end{equation*}
\end{proof}

\section{RS+RFD with UE-r Protocols} \label{appB:RSpFD_UE}

Visually, Fig.~\ref{fig:prob_tree_rsrfd_ue_r} illustrates the probability tree of the RS+RFD[UE-r] protocol (cf. Section~\ref{sub:rspfd_sol}).

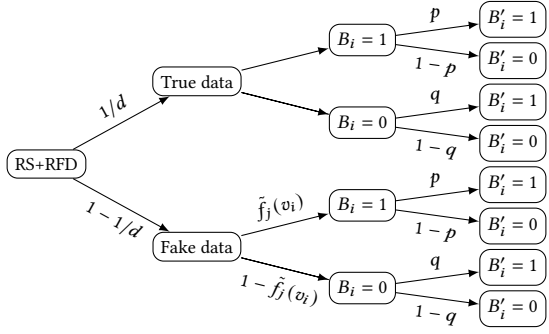
\begin{figure}[ht]
\centering
\tikzset{
  treenode/.style = {shape=rectangle, rounded corners,
                     draw,align=center,
                     top color=white},
  root/.style     = {treenode},
  env/.style      = {treenode},
  dummy/.style    = {circle,draw}
}
\tikzstyle{level 1}=[level distance=2cm, sibling distance=2.2cm]
\tikzstyle{level 2}=[level distance=2.2cm, sibling distance=1.1cm]
\tikzstyle{level 3}=[level distance=2cm, sibling distance=0.55cm]

\begin{tikzpicture}
  [
    grow                    = right,
    edge from parent/.style = {draw, -latex},
    every node/.style       = {font=\footnotesize},
    sloped
  ]
  \node [root] {RS+RFD}
    child { node [env] {Fake data}
        child { node [env] {$B_i=0$}  
         child { node [env] {$B_i'=0$}
            edge from parent node [below] {$1-q$}}
            child { node [env] {$B_i'=1$}
            edge from parent node [above] {$q$}}
            edge from parent node [below] {$1-\tilde{f}_j(v_i)$}
            edge from parent node [below] {}}
        child { node [env] {$B_i=1$} 
            child { node [env] {$B_i'=0$}
            edge from parent node [below] {$1-p$}}
            child { node [env] {$B_i'=1$}
            edge from parent node [above] {$p$}}
            edge from parent node [above] {$\tilde{f}_j(v_i)$}}
        edge from parent node [below] {$1-1/d$} }
    child { node [env] {True data}
        child { node [env] {$B_i=0$}  
         child { node [env] {$B_i'=0$}
            edge from parent node [below] {$1-q$}}
            child { node [env] {$B_i'=1$}
            edge from parent node [above] {$q$}}
            edge from parent node [below] {}
            edge from parent node [below] {}}
        child { node [env] {$B_i=1$} 
            child { node [env] {$B_i'=0$}
            edge from parent node [below] {$1-p$}}
            child { node [env] {$B_i'=1$}
            edge from parent node [above] {$p$}}
            edge from parent node [above] {}}
        edge from parent node [above] {$1/d$}};
\end{tikzpicture}
\caption{Probability tree for the RS+FD[UE-r] protocol.} 
\label{fig:prob_tree_rsrfd_ue_r}
\end{figure}

\begin{theorem} \label{theo:est_oue_r} For $j\in[1,d]$, the estimation result $\hat{f}_{\textrm{UE-R}}(v_i)$ in Eq.~\eqref{eq:est_rspfd_ue} is an unbiased estimation of $f (v_i)$ for any value $v_i \in A_j$.
\end{theorem}

\begin{proof}
\begin{equation*}
\begin{aligned}
     \mathbb{E}\left[ \hat{f}_{\textrm{UE-R}}(v_i) \right] &=  \mathbb{E}\left[\frac{d C(v_i) - n \left ( q + (p-q) (d-1) \tilde{f}_{j}(v_i) + q (d-1)  \right )}{n(p-q)}\right] \\
     &= \frac{d}{n(p-q)} \mathbb{E}\left[ C(v_i) \right] - \frac{(p-q)(d-1) \tilde{f}_{j}(v_i) + q + q (d-1) }{(p-q)} \textrm{.}
\end{aligned}
\end{equation*}

On expectation, the number of times that $v_i$ is reported is:

\begin{equation*}
     \mathbb{E}\left[ C(v_i) \right] = \frac{n}{d} \left( f(v_i)(p-q) + q )\right)  + \frac{n(d-1)}{d} \left(\tilde{f}_{j}(v_i) (p - q) + q \right) \textrm{.}
\end{equation*}

Therefore,

\[
\mathbb{E}\left[ \hat{f}_{\textrm{UE-R}}(v_i) \right]  =  f(v_i) \textrm{.}
\] 
\end{proof}

\begin{theorem} \label{theo:variance_oue_r} The variance of the estimation in Eq.~\eqref{eq:est_rspfd_ue} is:

\begin{equation}\label{var:rs+rfd_ue_r}
\begin{gathered}
    \operatorname{VAR} \left[ \hat{f}_{\textrm{UE-R}}(v_i) \right] = \frac{d^2 \gamma (1-\gamma)}{n (p-q)^2} \textrm{, where} \\
    \gamma = \frac{1}{d} \left(f(v_i) (p-q) + q + (d-1) \left(\tilde{f}_{j}(v_i) (p - q) + q \right) \right) \textrm{.}
\end{gathered}
\end{equation}

\end{theorem}

The proof of Theorem~\ref{theo:variance_oue_r} follows the proof of Theorem~\ref{theo:variance_grr} and is omitted here. \textcolor{black}{In this case, $\gamma$ follows the probability tree in Fig.~\ref{fig:prob_tree_rsrfd_ue_r}.}

\section{Additional Results for Section~\ref{sub:results_re_ident_smp}} \label{appC:add_re_ident_smp}

This section provides additional results for the risks of re-identification on collecting multidimensional data with the SMP solution. 
In addition to the standard LDP (\cf{} Definition~\ref{def:ldp}), we also use a recent relaxation of LDP known as PIE (Personal Information Entropy)~\cite{Murakami2021}, which aims at quantifying the re-identification risks in the local model. 

\subsection{PIE privacy model}
PIE is defined as the mutual information between user $U$ (random variable representing a user in $\mathcal{U}$) and perturbed data $Y$ (random variable representing perturbed data) as $\textrm{PIE}=I\left(U;Y\right) \textrm{ } (bits)$. 
As $I\left(U;Y\right)$ approaches 0, almost no information about user $U$ can be obtained through the perturbed data $Y$. 
Based on this observation, the authors in~\cite{Murakami2021} defined the privacy metric $(\mathcal{U},\alpha)$-PIE privacy that guarantees that the PIE is upper bounded by a parameter $\alpha$ for a set of users $\mathcal{U}$. 
More formally, let $p_{U,V}$ be the joint distribution of $U$ and $V$ (random variable representing personal data), $\Psi$ be a finite set of all humans and $\mathcal{U}\subseteq \Psi$ be a finite set of users reporting attribute $A_j$ of size $k_j$, the definition of $(\mathcal{U},\alpha)$-PIE privacy is:

\begin{definition} [$(\mathcal{U},\alpha)$-PIE privacy~\cite{Murakami2021}]\label{def:pie_privacy} Let $\mathcal{U}\subseteq \Psi$ and $\alpha \in \mathbb{R}_{\geq 0}$. An obfuscation mechanism $\mathcal{M}$ provides $(\mathcal{U},\alpha)$-PIE privacy if

\begin{equation} \label{eq:pie_privacy}
    \underset{p_{U,V}}{\textrm{sup}} \textrm{ } I(U;Y) \leq \alpha \textrm{ } (bits) \textrm{.}
\end{equation}

\end{definition}

Since the inequality in Eq.~\eqref{eq:pie_privacy} holds for any distribution $p_{U,V}$, the PIE is upper bounded by $\alpha$ irrespective of the adversary's background knowledge~\cite{Murakami2021}. 
The parameter $\alpha$ plays a role similar to the privacy budget $\epsilon$ in LDP and can be selected by fixing the lowest possible Bayes error probability $\beta_{U|S}$ (given a score vector $S$) as:

\begin{corollary} \label{coro:bayes_error} (Bayes error and PIE privacy~\cite{Murakami2021}). 
Let $\mathcal{U}\subseteq \Psi$ and $\alpha \in \mathbb{R}_{\geq 0}$. 
If an obfuscation mechanism $\mathcal{M}$ provides $(\mathcal{U},\alpha)$-PIE privacy and if $U$ is uniformly distributed (\ie, one data item), then

\begin{equation} \label{eq:bayes_error}
    \beta_{U|S} \geq 1 - \frac{\alpha + 1}{\log_2 (n)} \textrm{.}
\end{equation}

\end{corollary}

Lastly, the relationship between LDP and PIE is:

\begin{proposition} (LDP and PIE~\cite{Murakami2021}). 
\label{prop:pie_ldp} If an obfuscation mechanism $\mathcal{M}$ provides $\epsilon$-LDP, then it provides ($(\mathcal{U},\alpha)$-PIE privacy) for any $\mathcal{U}\subseteq \Psi$ such that $|\mathcal{U}|=n$, where

\begin{equation} \label{eq:alpha}
    \alpha = \textrm{min} \left\{ \epsilon \log_2 (e), \epsilon^2 \log_2 (e), \log_2 (n), \log_2 (k_j) \right\} \textrm{.}
\end{equation}
\end{proposition}

As stated in~\cite{Murakami2021}, Proposition~\ref{prop:pie_ldp} holds for any LDP protocol.

\subsection{Experiments} 
We follow a similar experimental evaluation to the one described in Section~\ref{sub:results_re_ident_smp} and we vary the following:

\begin{itemize}
    
    \item \textbf{Dataset.} Adult~\cite{uci} and ACSEmployement~\cite{ding2021retiring} datasets.
    
    \item \textbf{LDP protocol.} GRR~\cite{kairouz2016discrete,kairouz2016extremal}, OLH~\cite{tianhao2017}, $\omega$-SS~\cite{wang2016mutual,Min2018}, SUE (a.k.a. Basic One-Time RAPPOR~\cite{rappor}) and OUE~\cite{tianhao2017}.
    
    \item \textbf{Re-identification model.} Full knowledge (FK-RI) and partial knowledge (PK-RI) models (cf. Section~\ref{sub:re_ident_attack_models}). The background knowledge $\mathcal{D}_{BK}$ and $\mathcal{D}_{PK} \subseteq \mathcal{D}_{BK}$ (a random subset with at least $\frac{d}{2}$ attributes) is the own Adult and ACSEmployement datasets.
    
    \item \textbf{Privacy metric.} Local differential privacy (with $\epsilon=[1,2,\ldots,9,10]$) and $\alpha$-PIE privacy~\cite{Murakami2021} (with $\beta_{U|S}=[0.95, 0.9, \ldots, 0.55, 0.5]$, \ie, from tighter privacy regimes to lower ones). In this case, following Eq.~\eqref{eq:alpha} and~\cite[Proposition 9]{Murakami2021}, when $k_j$ (attribute's domain size) is small, one do not need to use an LDP protocol (\ie, $y=v$).
    
    \item \textbf{Privacy metric across users.} Uniform (cf. Section~\ref{sub:plausible_deniability_uniform}) and non-uniform (cf. Section~\ref{sub:plausible_deniability_non_uniform}) privacy metrics.

\end{itemize}

First, within the same setting of Fig.~\ref{fig:reident_smp} (\textbf{Adult dataset}), Fig.~\ref{fig:reident_smp_acs} (\textbf{ACSEmployement dataset}) illustrates the attacker's RID-ACC for top-k re-identification on using the SMP solution, the FK-RI model with uniform $\epsilon$-LDP privacy metric across users, by varying the LDP protocol and the number of surveys. 
One can notice that both Figs.~\ref{fig:reident_smp} and~\ref{fig:reident_smp_acs} follow similar patterns for all LDP protocols with only different upper bounds for the RID-ACC metric.
Therefore, results for the PK-RI model, $\alpha$-PIE and non-uniform privacy metrics were omitted for the ACSEmployement dataset as they follow similar results achieved with the Adult dataset where the difference is mainly in the attacker's RID-ACC upper bound. 

Secondly, to complement the results of Fig.~\ref{fig:reident_smp} (\textbf{FK-RI model}), Fig.~\ref{fig:reident_smp_pk_uni} (\textbf{PK-RI model}) illustrates the attacker's RID-ACC for top-k re-identification on using the SMP solution, the Adult dataset, with uniform $\epsilon$-LDP privacy metric across users, by varying the LDP protocol and the number of surveys. 
From Figs.~\ref{fig:reident_smp} and~\ref{fig:reident_smp_pk_uni}, one can notice that the PK-RI model minimized the RID-ACC by about $50\%$ for all LDP protocols, which represents a more realistic scenario for the background knowledge of adversaries in real-life.

Since Figs.~\ref{fig:reident_smp} and~\ref{fig:reident_smp_pk_uni} present results for the \textbf{uniform} $\epsilon$-LDP privacy metric setting, Fig.~\ref{fig:reident_smp_nonuni_eps} illustrates the attacker's RID-ACC on the Adult dataset for top-k re-identification on using the SMP solution, the full knowledge FK-RI model (left-side plots) and partial knowledge PK-RI model (right-side plots) with \textbf{non-uniform} $\epsilon$-LDP privacy metric across users, and by varying the LDP protocol and the number of surveys (i.e., data collections).
With a uniform privacy metric, all users report a different attribute per survey, which can lead to higher re-identification rates. 
On the other hand, if users sample attributes with replacement, there is a lower probability of selecting different attributes for each survey, which can bound the re-identification risks. 
For these reasons, the RID-ACC of Fig.~\ref{fig:reident_smp_nonuni_eps} decreased by about $40\%$ in comparison with the RID-ACC of Figs.~\ref{fig:reident_smp} and~\ref{fig:reident_smp_pk_uni}.

Focusing now on the $\mathbf{\alpha}$\textbf{-PIE metric}, Figs.~\ref{fig:reident_smp_uni_alpha} (\textbf{uniform setting}) and~\ref{fig:reident_smp_nonuni_alpha} (\textbf{non-uniform setting}) illustrate the attacker's RID-ACC on the Adult dataset for top-k re-identification on using the SMP solution, the full knowledge FK-RI model (left-side plots) and partial knowledge PK-RI model (right-side plots), and by varying the LDP protocol and the number of surveys (i.e., data collections).
As one can note, even using a high Bayes error probability as $\beta_{U|S}=0.95$ for each attribute and collection, the $(\mathcal{U}, \alpha)$-PIE privacy metric leads to higher attacker's RID-ACC in comparison with Figs.~\ref{fig:reident_smp},~\ref{fig:reident_smp_pk_uni} and~\ref{fig:reident_smp_nonuni_eps} in which $\epsilon =1$. 
In fact, these higher $RID\textrm{-}ACC$ when using $\beta_{U|S}$ can be explained due to not using a local randomizer when $k_j$ is small~\cite[Proposition 9]{Murakami2021} (the case for several attributes in the Adult dataset), with more difference for other protocols such as SUE, OUE and OLH. 
For this reason, the re-identification rates of all LDP protocols follow similar behaviors on Figs.~\ref{fig:reident_smp_uni_alpha} and~\ref{fig:reident_smp_nonuni_alpha}. 

Overall, one can note that the attackers' RID-ACC follow similar behavior as the analytical results in Fig.~\ref{fig:analytical_ACC_uniform_VS_non_uniform}. 
Indeed, the risks of re-identification deeply depend on how well the attacker can profile each user on multiple collections (\cf{} Sections~\ref{sub:plausible_deniability_uniform} and~\ref{sub:plausible_deniability_non_uniform}) such that the profile can be found in the background knowledge $\mathbb{D}_{BK}$ (or $\mathbb{D}_{PK}$). 
Besides, the results agree with intuitive expectations, as increasing $\epsilon$ (or decreasing $\beta_{U|S}$) also increases the privacy leakage, thus leading to higher re-identification rates.
Moreover, following the $\epsilon$-LDP metric, the lowest attacker's RID-ACC were achieved by both OLH and OUE protocols~\cite{tianhao2017}, as they have an expected upper bound for the attacker's $ACC_{FO}=1/2$~\cite{Gursoy2022} (cf. Section~\ref{sub:plausible_deniability}) per data collection. 
Besides, the $\omega$-SS protocol follows similar re-identification rates as the GRR protocol (also shown analytically in Fig.~\ref{fig:analytical_ACC_uniform_VS_non_uniform}), which is the highest among all LDP protocols. 
On the other hand, the re-identification rates of all LDP protocols follow similar behaviors when using the $\alpha$-PIE privacy model on Figs.~\ref{fig:reident_smp_uni_alpha} and~\ref{fig:reident_smp_nonuni_alpha}. 

\begin{figure*}[!ht]
\begin{subfigure}{.5\textwidth}
  \centering
  \includegraphics[width=1\linewidth]{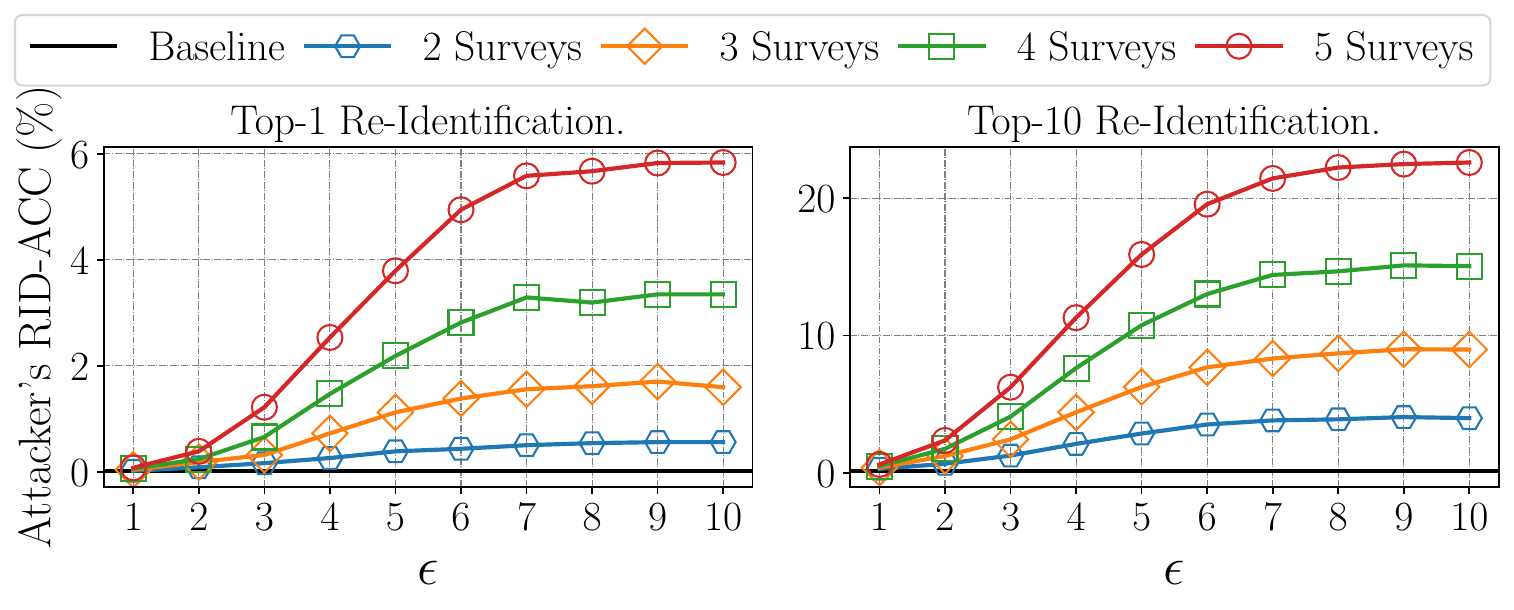}
  \caption{Re-identification risk of the GRR~\cite{kairouz2016discrete,kairouz2016extremal} protocol.}
\end{subfigure}%
\begin{subfigure}{.5\textwidth}
  \centering
  \includegraphics[width=1\linewidth]{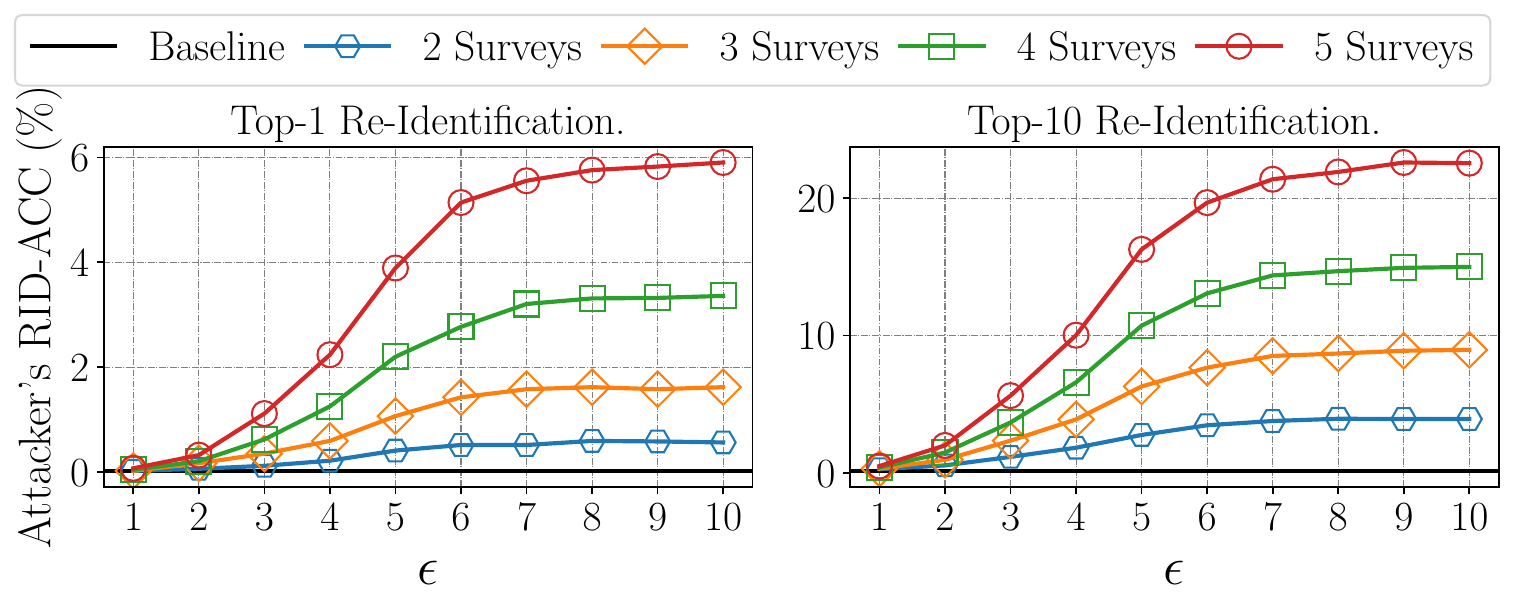}
  \caption{Re-identification risk of the $\omega$-SS~\cite{wang2016mutual,Min2018} protocol.}
\end{subfigure}
\\
\begin{subfigure}{.5\textwidth}
  \centering
  \includegraphics[width=1\linewidth]{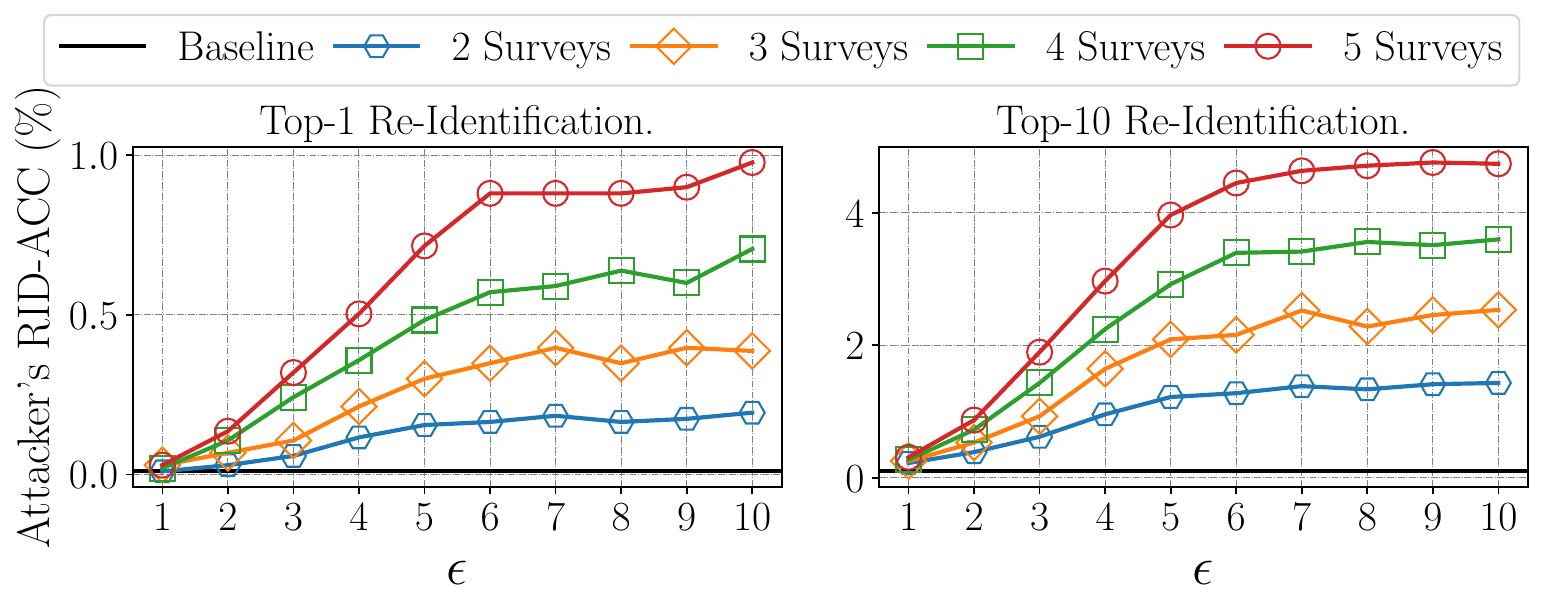}
  \caption{Re-identification risk of the OLH~\cite{tianhao2017} protocol.}
\end{subfigure}%
\begin{subfigure}{.5\textwidth}
  \centering
  \includegraphics[width=1\linewidth]{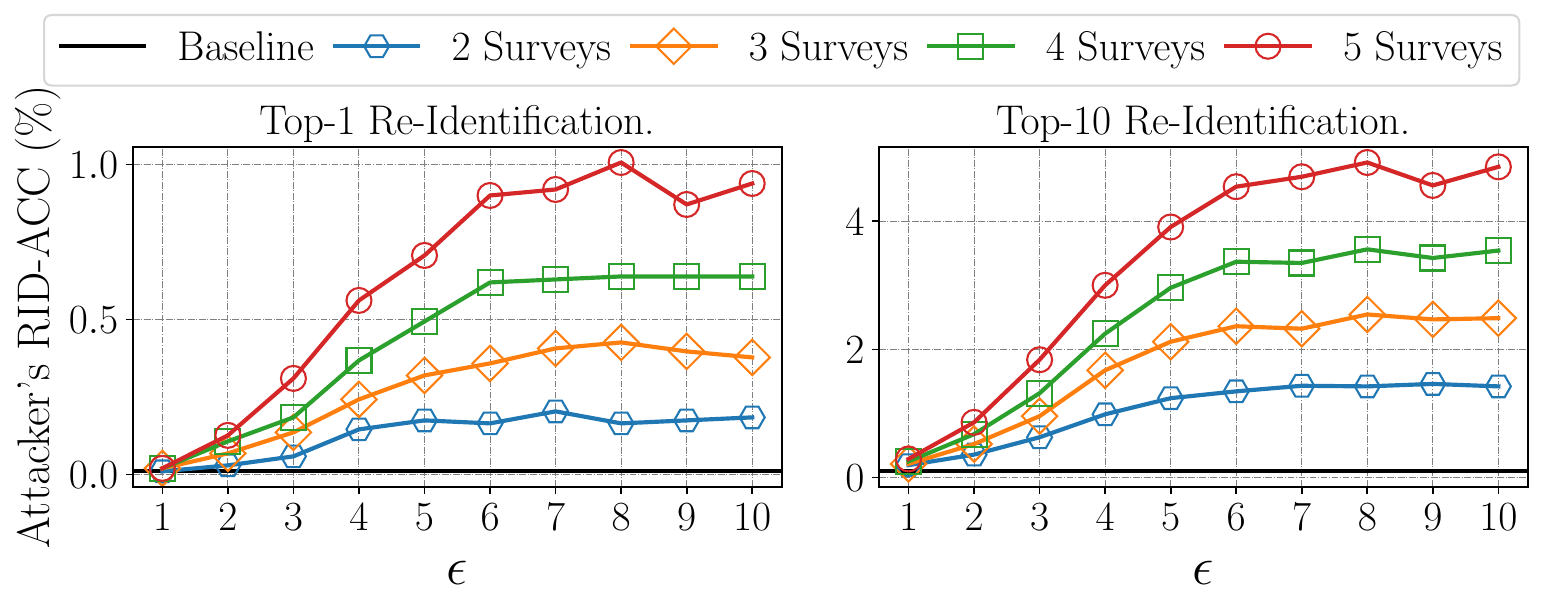}
  \caption{Re-identification risk of the OUE~\cite{tianhao2017} protocol.}
\end{subfigure}
\\
\begin{subfigure}{1\textwidth}
  \centering
  \includegraphics[width=.5\linewidth]{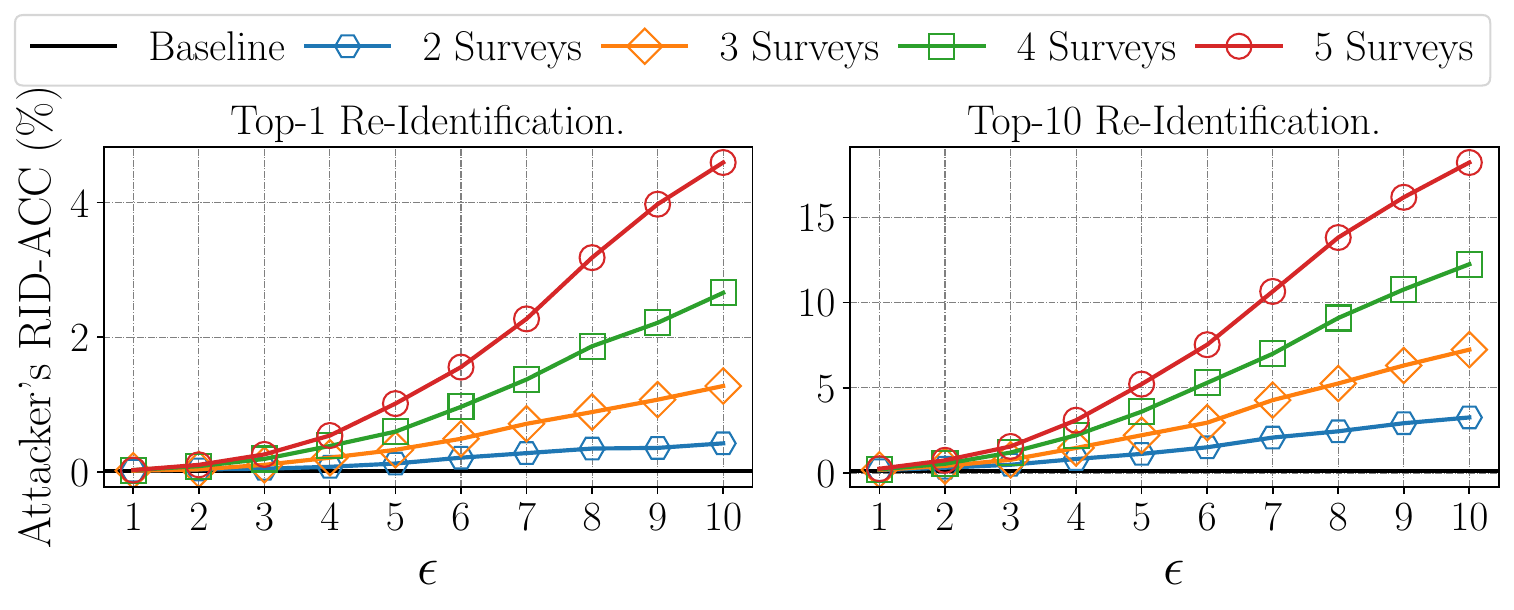}
  \caption{Re-identification risk of the SUE (\emph{a.k.a.} RAPPOR)~\cite{rappor} protocol.}
\end{subfigure}
\caption{Attacker's re-identification accuracy (RID-ACC) on the ACSEmployement dataset for top-k re-identification on using the SMP solution, the full knowledge FK-RI model with uniform $\epsilon$-LDP privacy metric across users, and by varying the LDP protocol and the number of surveys (i.e., data collections).}
\label{fig:reident_smp_acs}
\end{figure*}

\begin{figure*}[!ht]
\begin{subfigure}{.5\textwidth}
  \centering
  \includegraphics[width=1\linewidth]{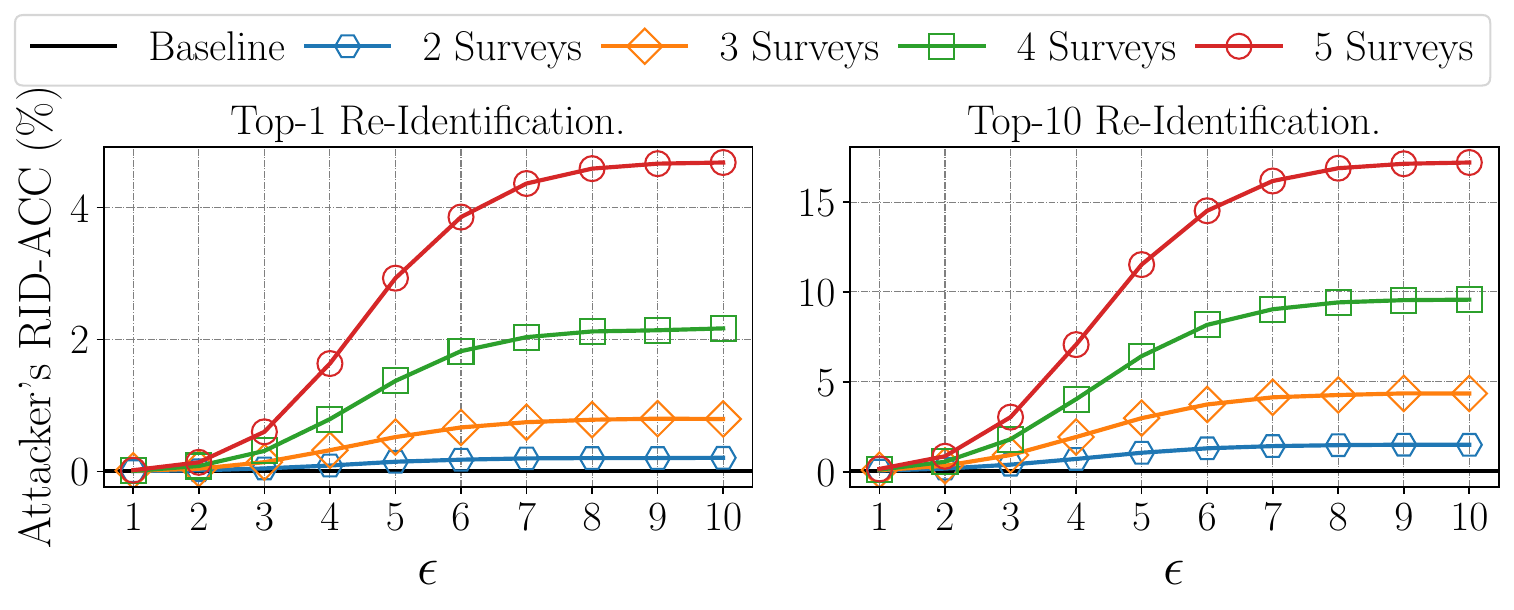}
  \caption{Re-identification risk of the GRR~\cite{kairouz2016discrete,kairouz2016extremal} protocol.}
\end{subfigure}%
\begin{subfigure}{.5\textwidth}
  \centering
  \includegraphics[width=1\linewidth]{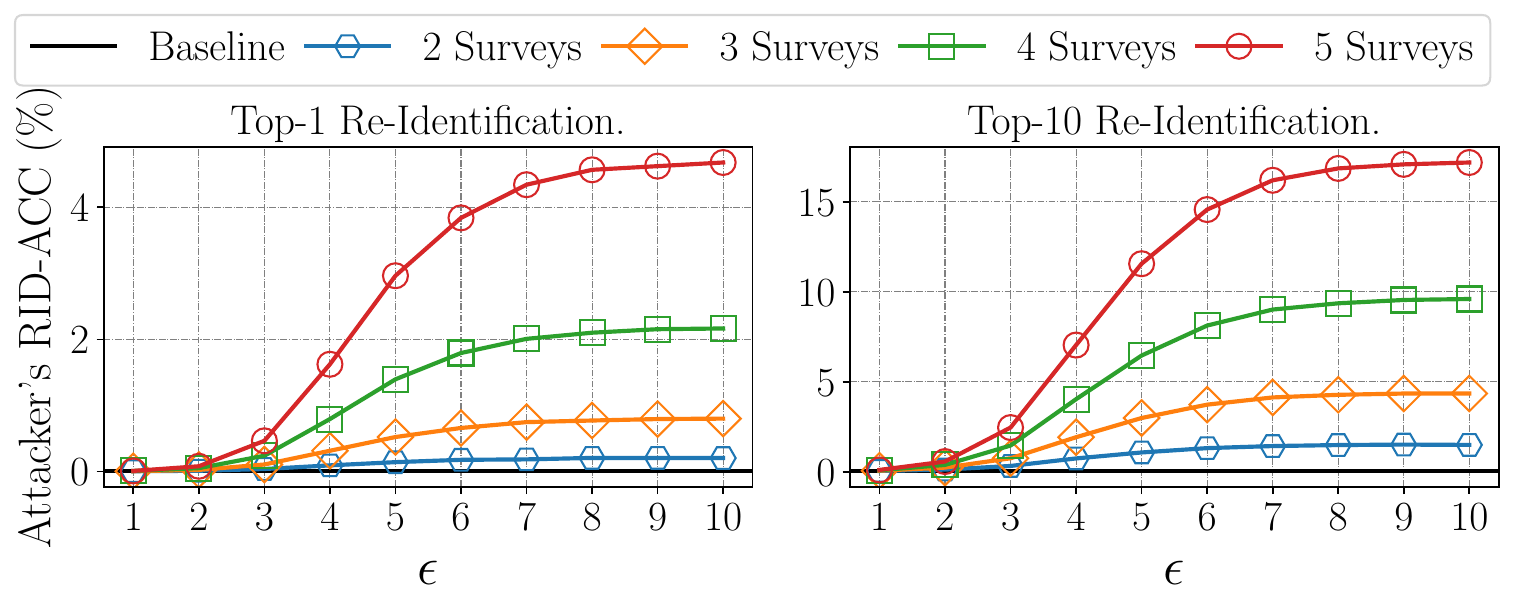}
  \caption{Re-identification risk of the $\omega$-SS~\cite{wang2016mutual,Min2018} protocol.}
\end{subfigure}
\\
\begin{subfigure}{.5\textwidth}
  \centering
  \includegraphics[width=1\linewidth]{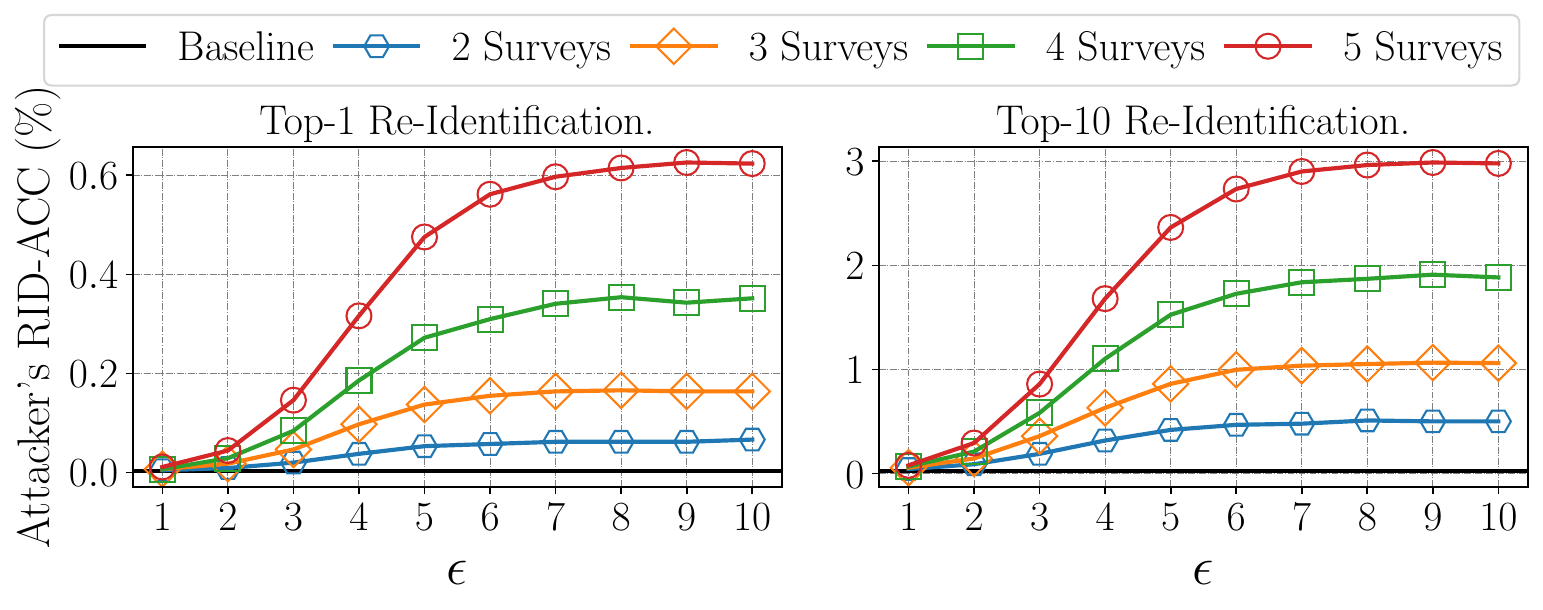}
  \caption{Re-identification risk of the OLH~\cite{tianhao2017} protocol.}
\end{subfigure}%
\begin{subfigure}{.5\textwidth}
  \centering
  \includegraphics[width=1\linewidth]{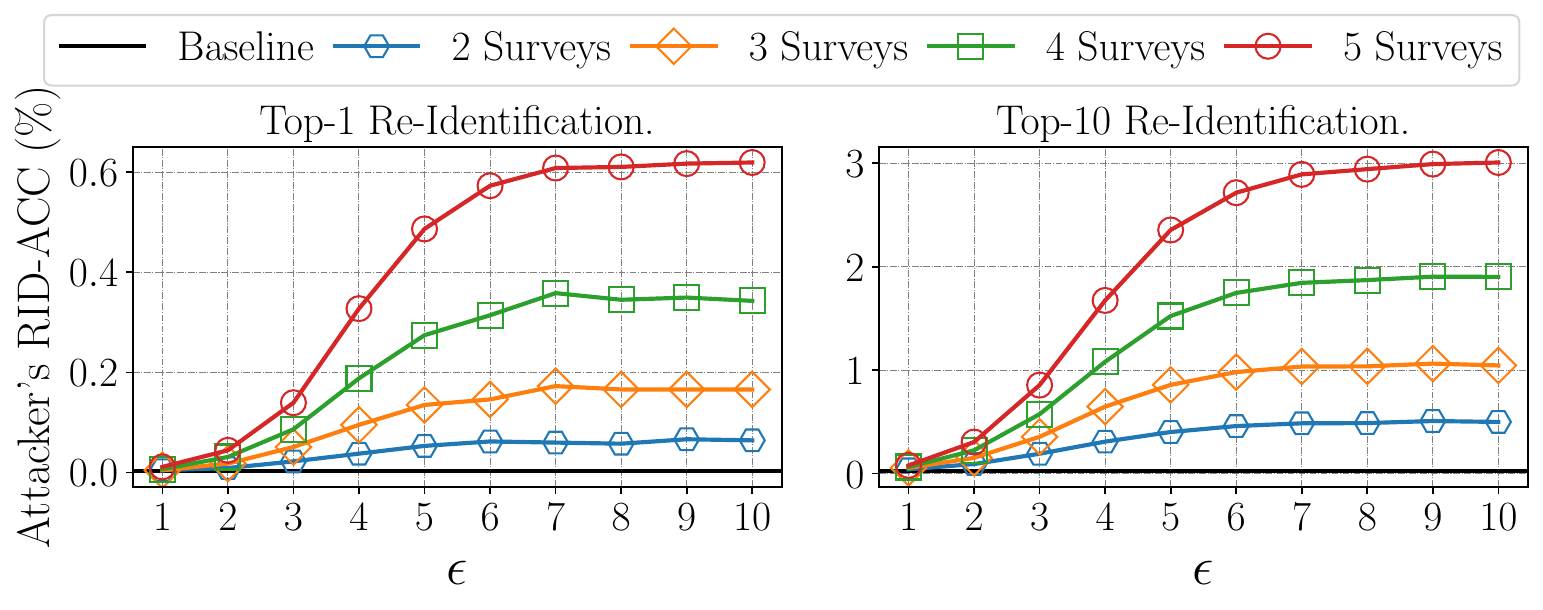}
  \caption{Re-identification risk of the OUE~\cite{tianhao2017} protocol.}
\end{subfigure}
\\
\begin{subfigure}{1\textwidth}
  \centering
  \includegraphics[width=.5\linewidth]{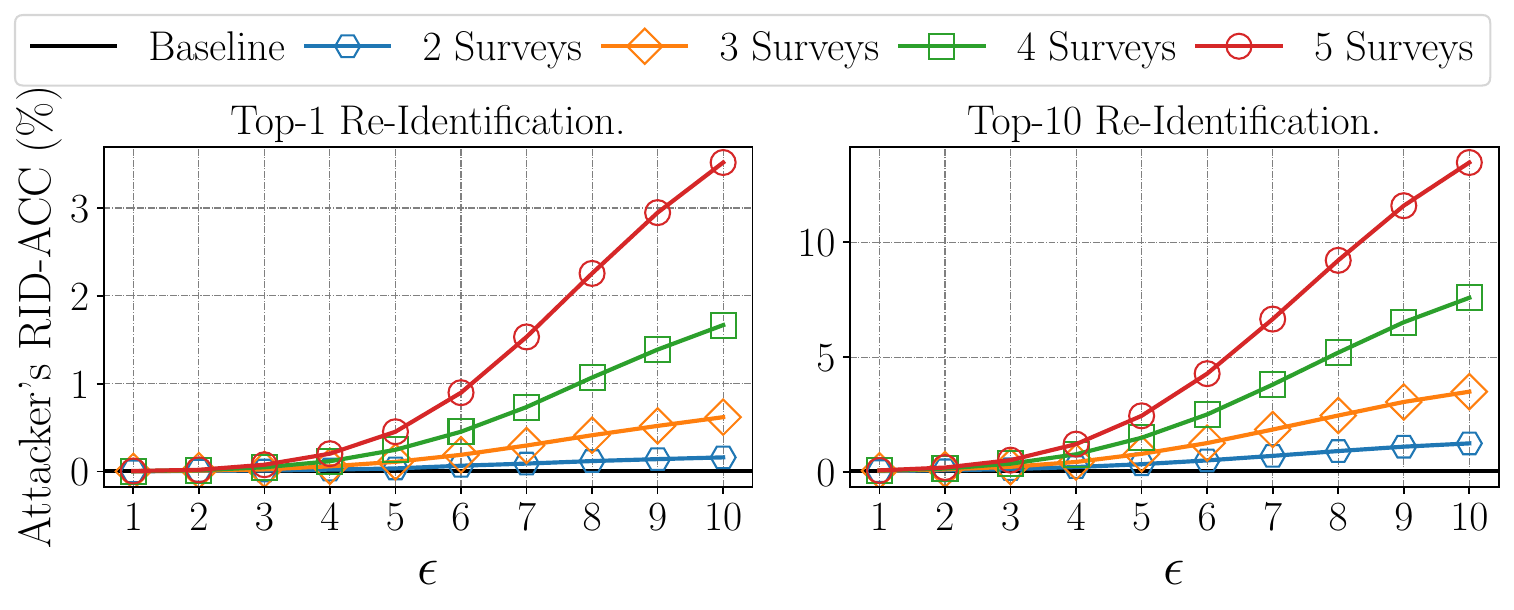}
  \caption{Re-identification risk of the SUE (\emph{a.k.a.} RAPPOR)~\cite{rappor} protocol.}
\end{subfigure}
\caption{Attacker's re-identification accuracy (RID-ACC) on the Adult dataset for top-k re-identification on using the SMP solution, the partial knowledge PK-RI model with uniform $\epsilon$-LDP privacy metric across users, and by varying the LDP protocol and the number of surveys (i.e., data collections).}
\label{fig:reident_smp_pk_uni}
\end{figure*}

\begin{figure*}[!ht]
\begin{subfigure}{.5\textwidth}
  \centering
  \includegraphics[width=0.95\linewidth]{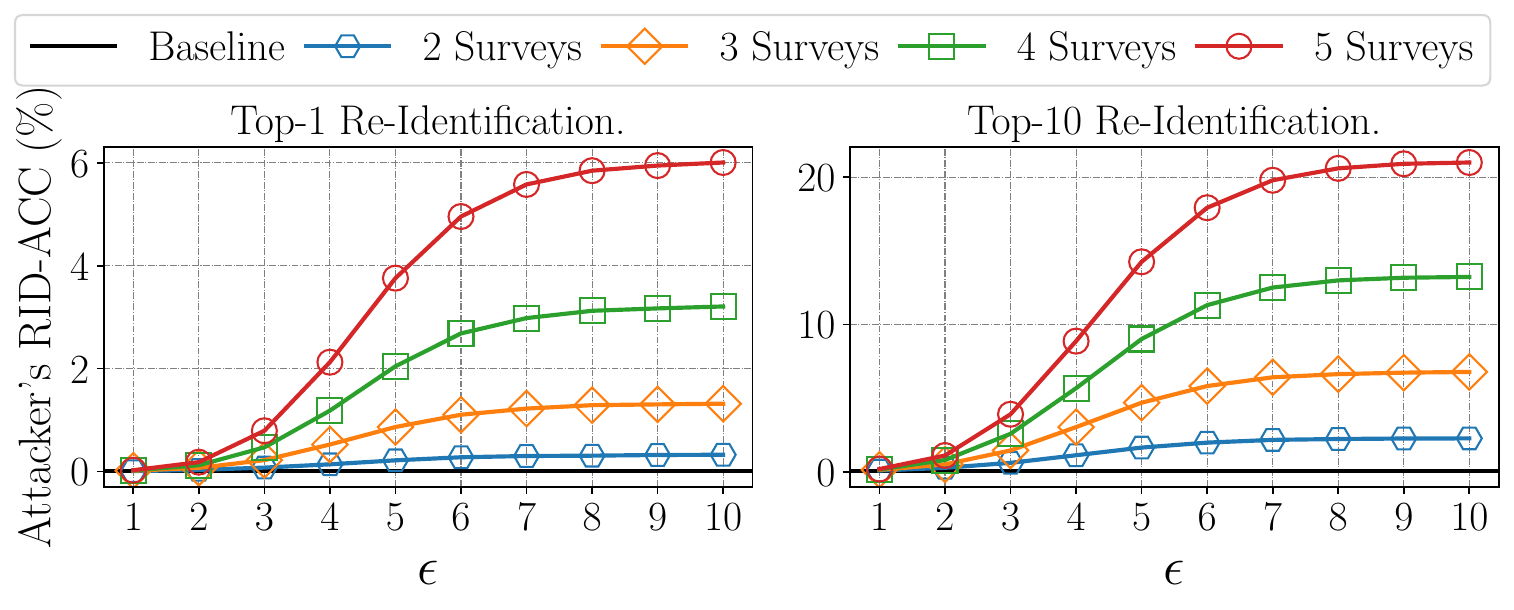}
  \caption{FK-RI risk of the GRR~\cite{kairouz2016discrete,kairouz2016extremal} protocol.}
\end{subfigure}%
\begin{subfigure}{.5\textwidth}
  \centering
  \includegraphics[width=0.95\linewidth]{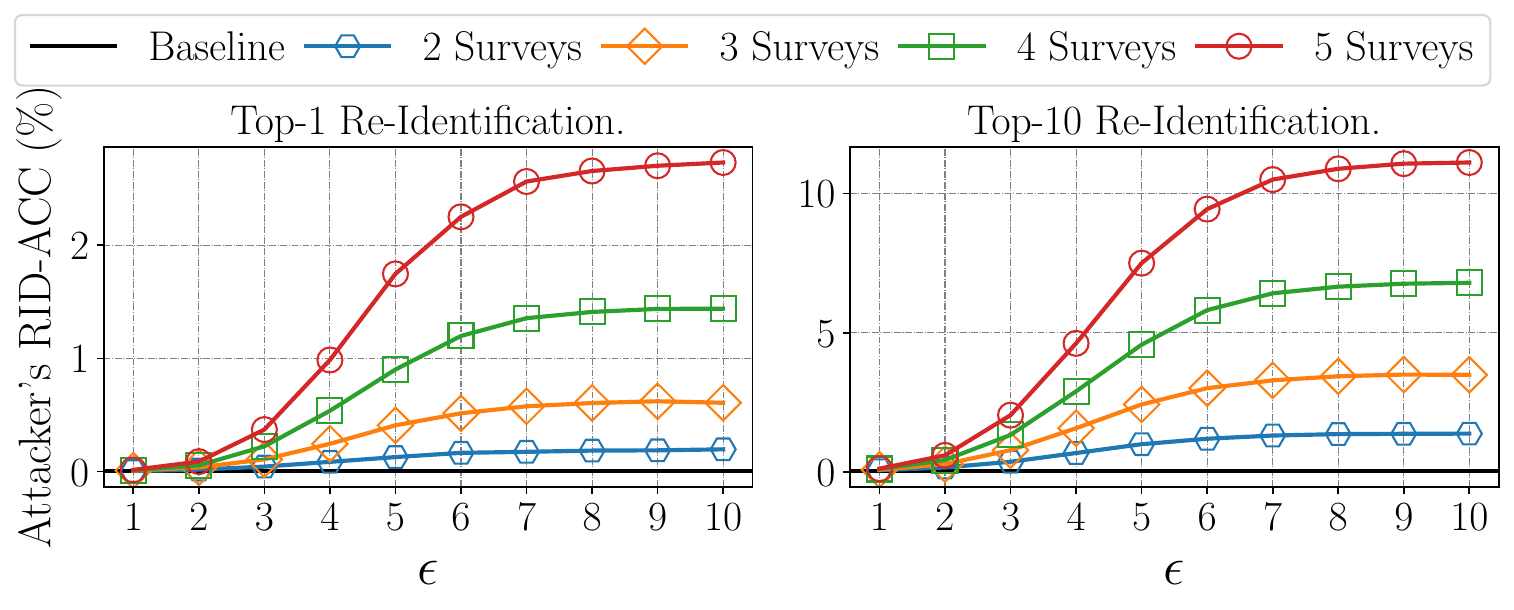}
  \caption{PK-RI risk of the GRR~\cite{kairouz2016discrete,kairouz2016extremal} protocol.}
\end{subfigure}
\\
\begin{subfigure}{.5\textwidth}
  \centering
  \includegraphics[width=0.95\linewidth]{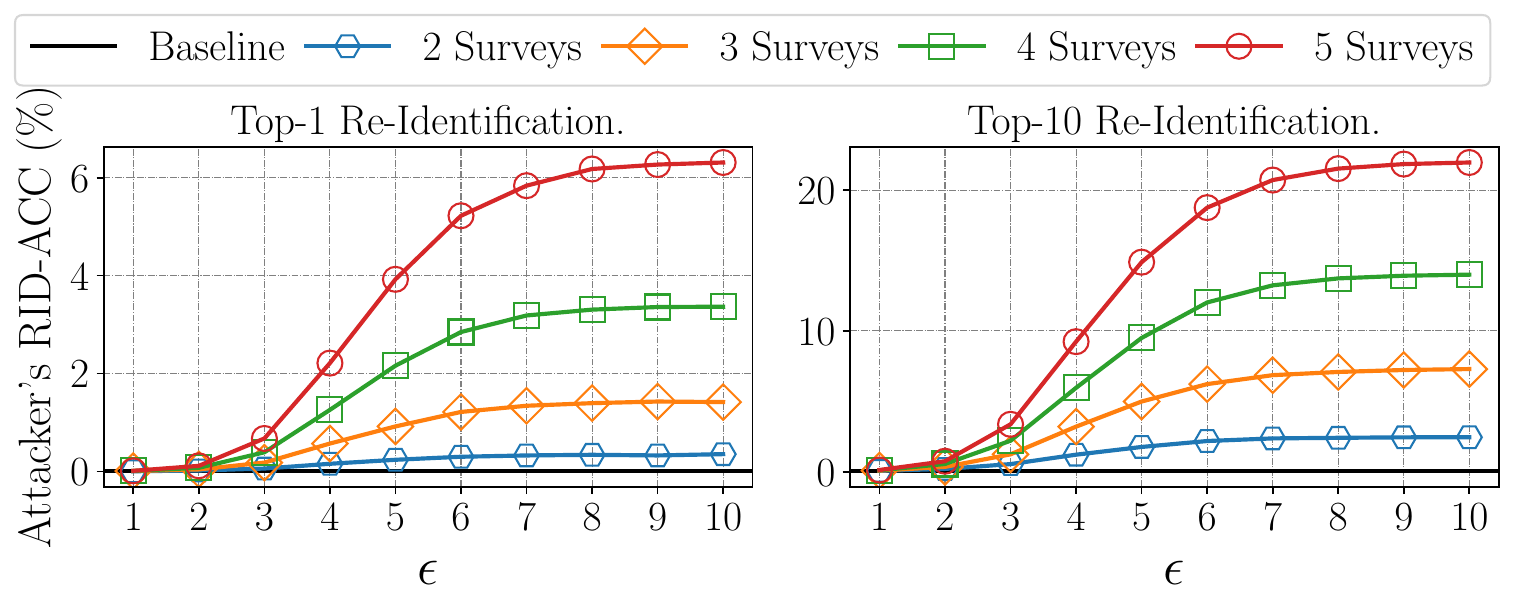}
  \caption{FK-RI risk of the $\omega$-SS~\cite{wang2016mutual,Min2018} protocol.}
\end{subfigure}%
\begin{subfigure}{.5\textwidth}
  \centering
  \includegraphics[width=0.95\linewidth]{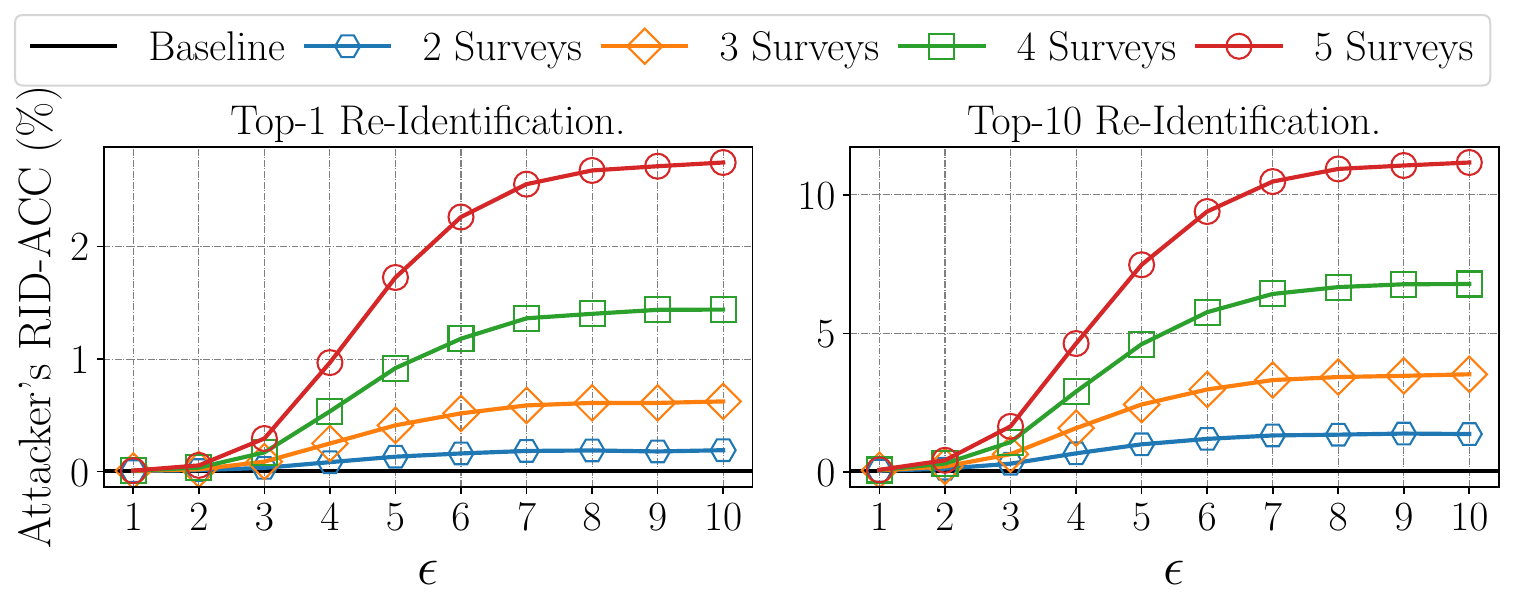}
  \caption{PK-RI risk of the $\omega$-SS~\cite{wang2016mutual,Min2018} protocol.}
\end{subfigure}
\\
\begin{subfigure}{.5\textwidth}
  \centering
  \includegraphics[width=0.95\linewidth]{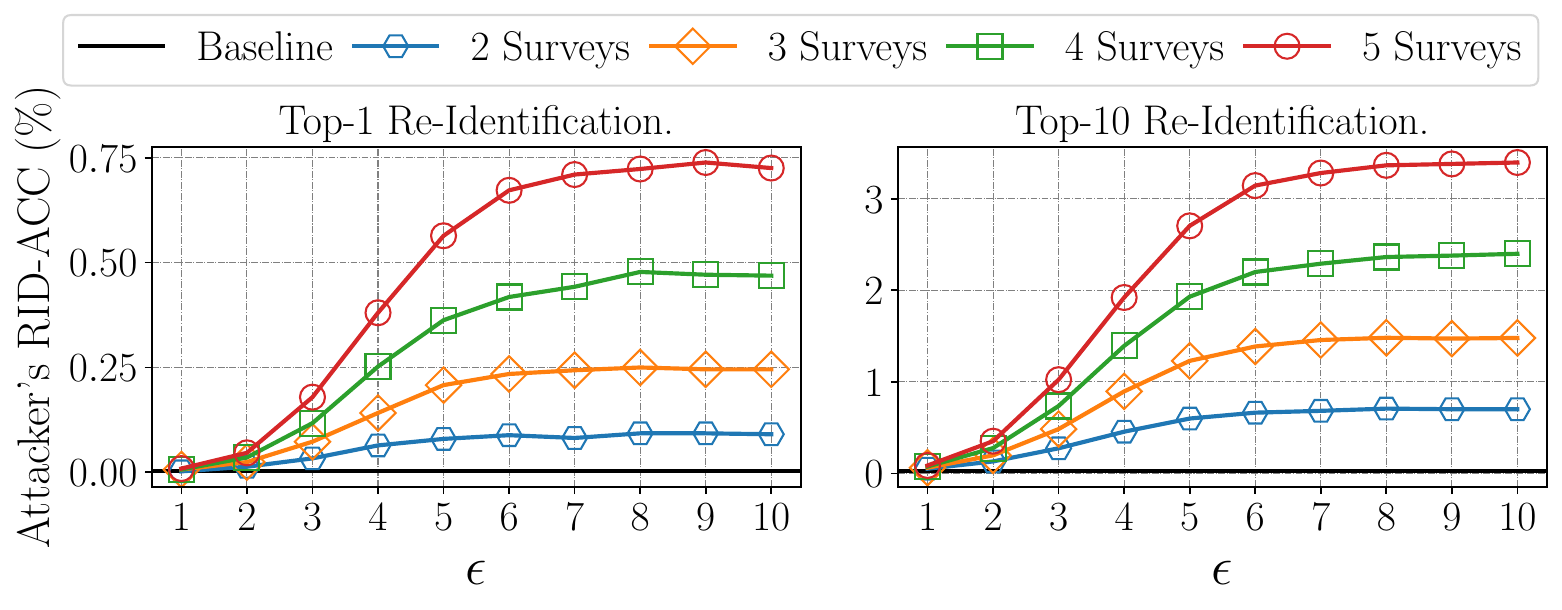}
  \caption{FK-RI risk of the OLH~\cite{tianhao2017} protocol.}
\end{subfigure}%
\begin{subfigure}{.5\textwidth}
  \centering
  \includegraphics[width=0.95\linewidth]{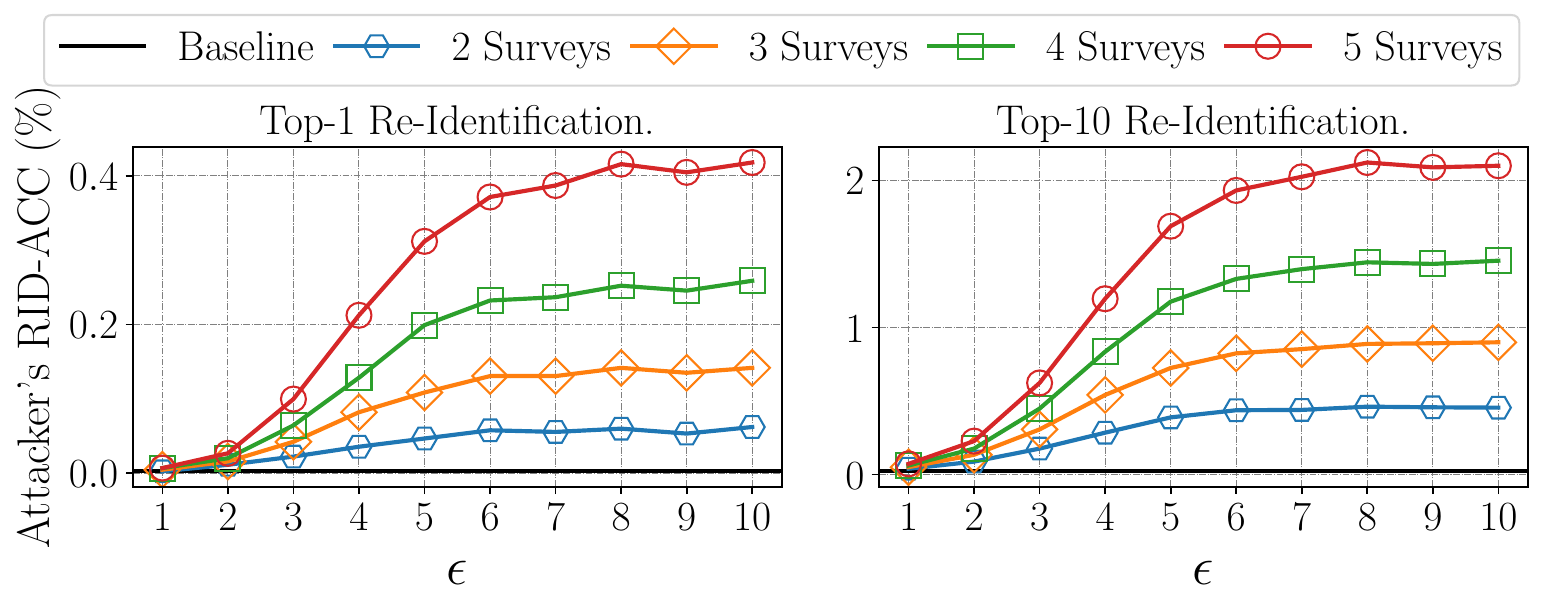}
  \caption{PK-RI risk of the OLH~\cite{tianhao2017} protocol.}
\end{subfigure}
\\
\begin{subfigure}{.5\textwidth}
  \centering
  \includegraphics[width=0.95\linewidth]{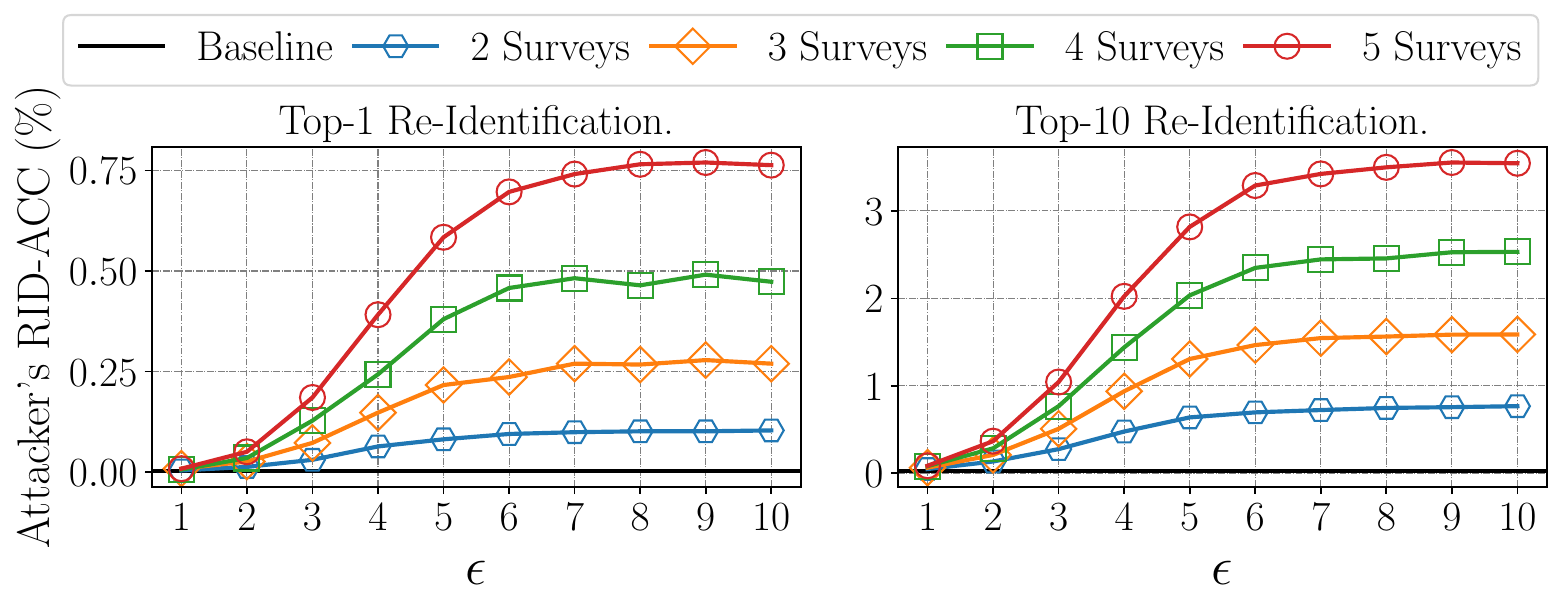}
  \caption{FK-RI risk of the OUE~\cite{tianhao2017} protocol.}
\end{subfigure}%
\begin{subfigure}{.5\textwidth}
  \centering
  \includegraphics[width=0.95\linewidth]{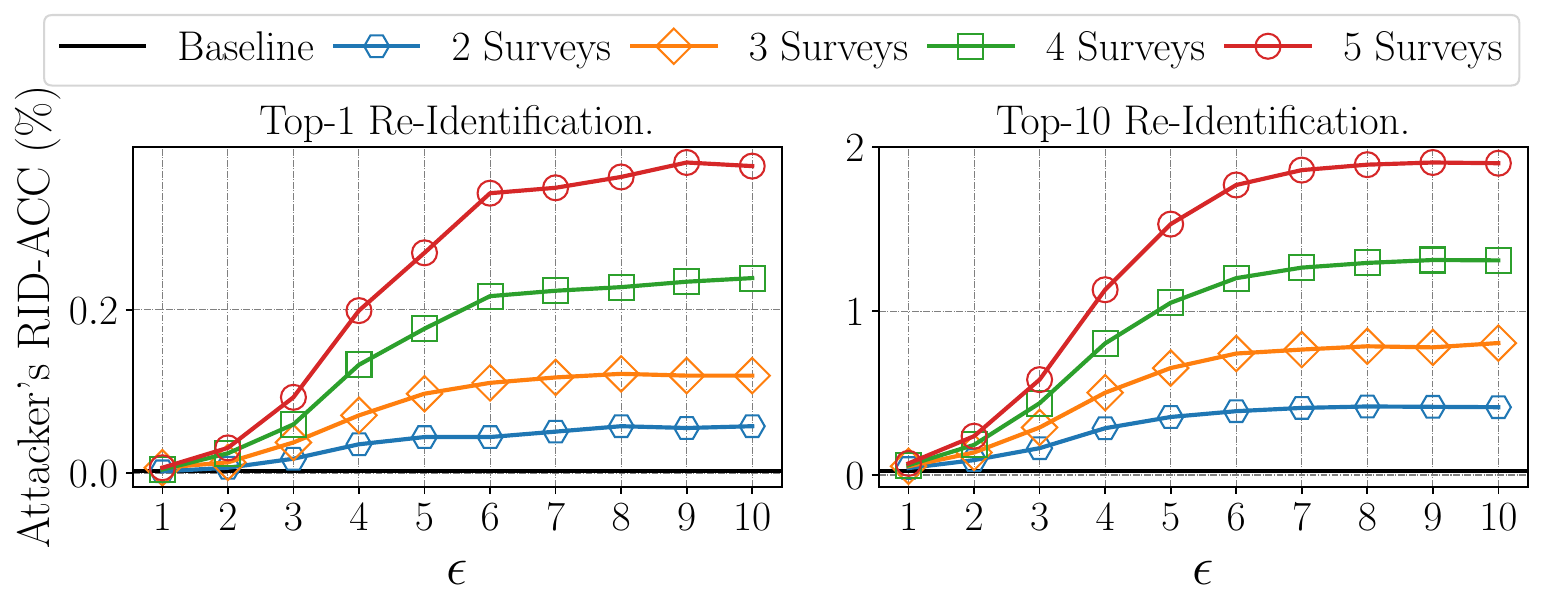}
  \caption{PK-RI risk of the OUE~\cite{tianhao2017} protocol.}
\end{subfigure}
\\
\begin{subfigure}{.5\textwidth}
  \centering
  \includegraphics[width=0.95\linewidth]{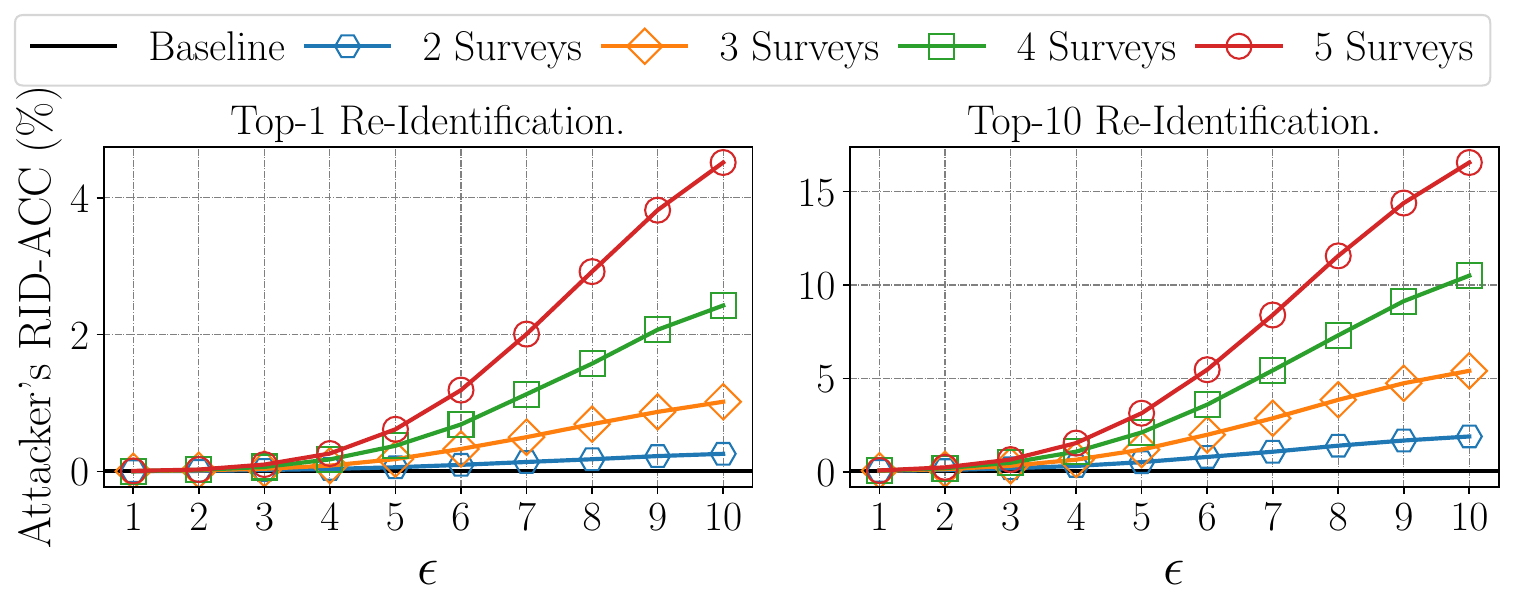}
  \caption{FK-RI risk of the SUE (\emph{a.k.a.} RAPPOR)~\cite{rappor} protocol.}
\end{subfigure}%
\begin{subfigure}{.5\textwidth}
  \centering
  \includegraphics[width=0.95\linewidth]{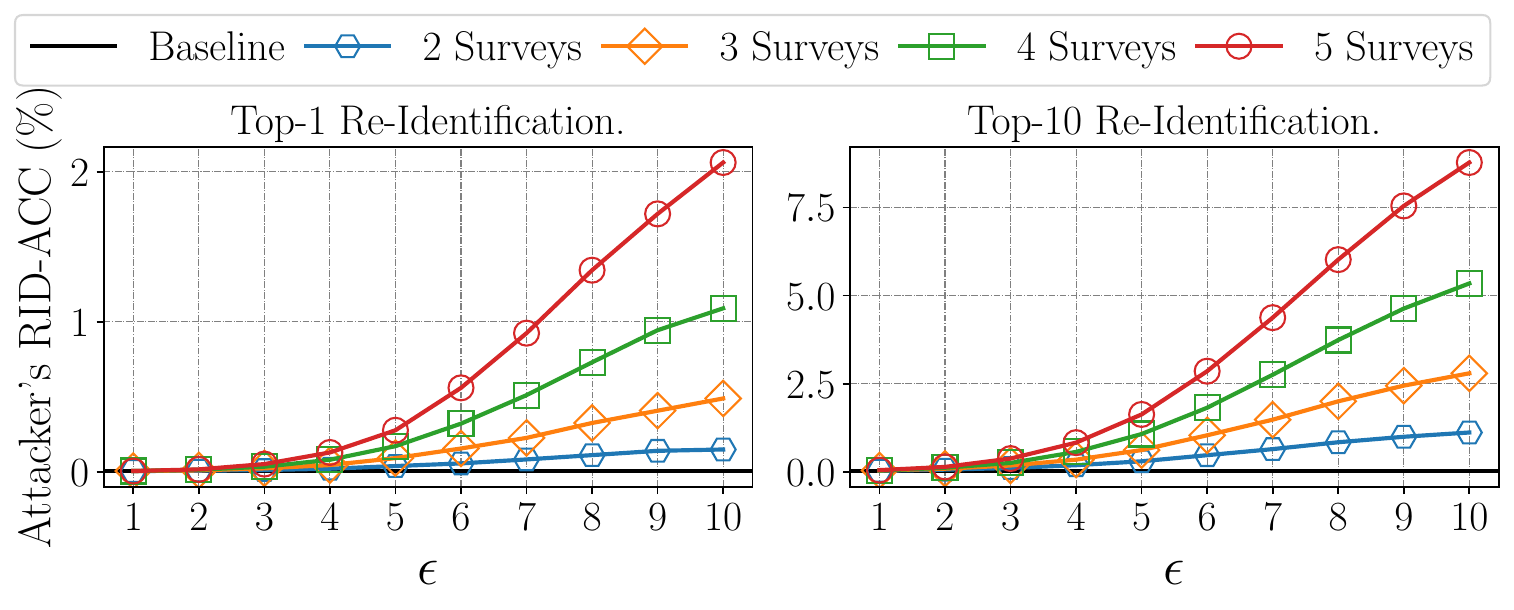}
  \caption{PK-RI risk of the SUE (\emph{a.k.a.} RAPPOR)~\cite{rappor} protocol.}
\end{subfigure}
\caption{Attacker's re-identification accuracy (RID-ACC) on the Adult dataset for top-k re-identification on using the SMP solution, the full knowledge FK-RI model (left-side plots) and partial knowledge PK-RI model (right-side plots) with non-uniform $\epsilon$-LDP privacy metric across users, and by varying the LDP protocol and the number of surveys (i.e., data collections).}
\label{fig:reident_smp_nonuni_eps}
\end{figure*}

\begin{figure*}[!ht]
\begin{subfigure}{.5\textwidth}
  \centering
  \includegraphics[width=0.95\linewidth]{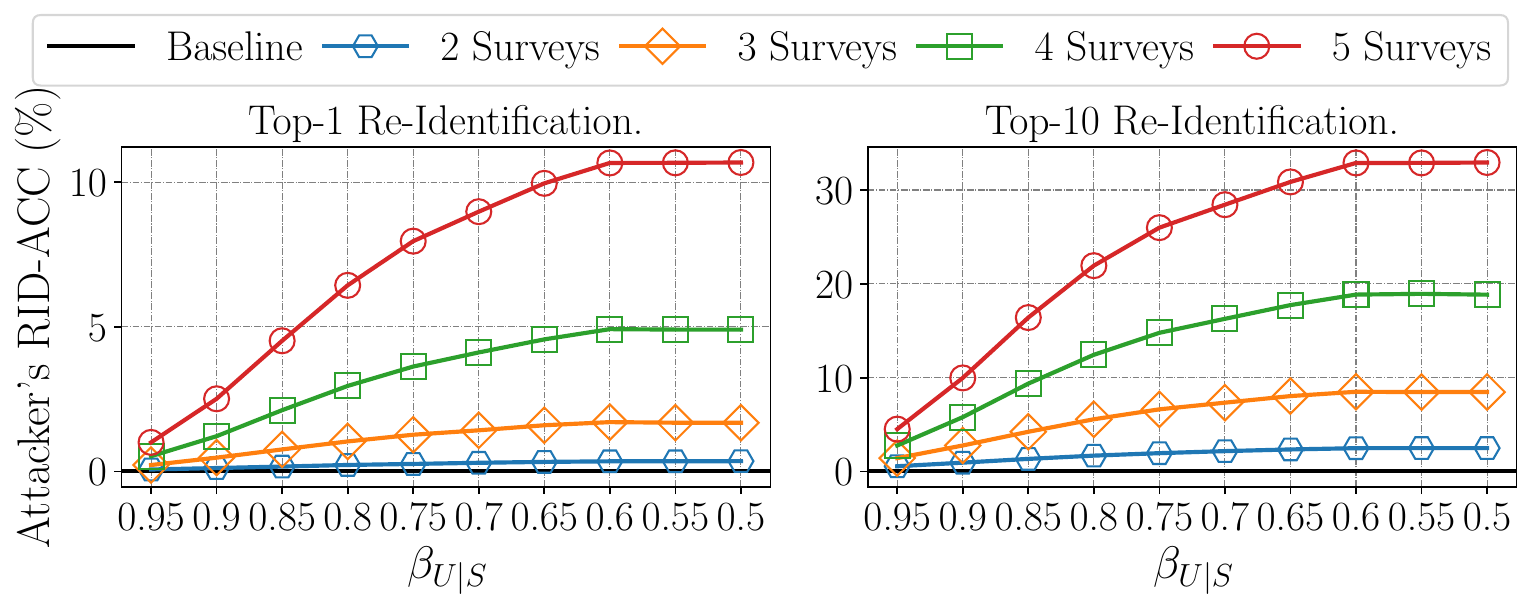}
  \caption{FK-RI risk of the GRR~\cite{kairouz2016discrete,kairouz2016extremal} protocol.}
\end{subfigure}%
\begin{subfigure}{.5\textwidth}
  \centering
  \includegraphics[width=0.95\linewidth]{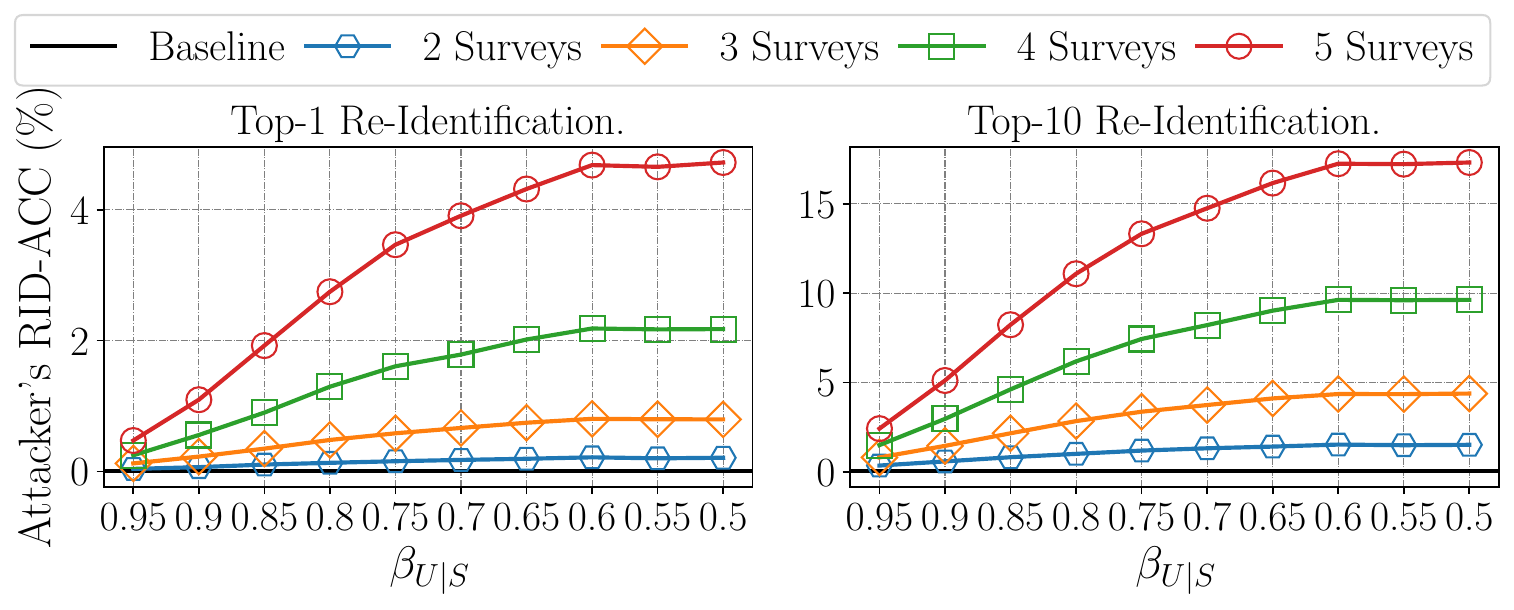}
  \caption{PK-RI risk of the GRR~\cite{kairouz2016discrete,kairouz2016extremal} protocol.}
\end{subfigure}
\\
\begin{subfigure}{.5\textwidth}
  \centering
  \includegraphics[width=0.95\linewidth]{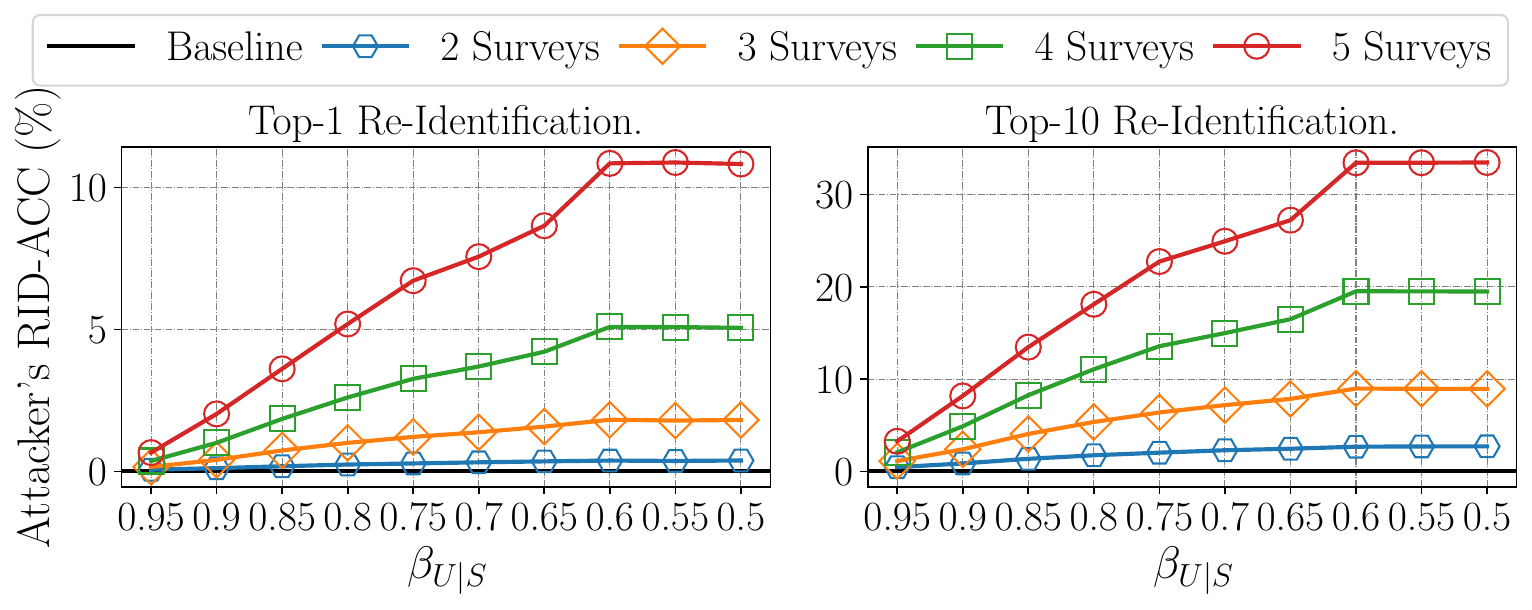}
  \caption{FK-RI risk of the $\omega$-SS~\cite{wang2016mutual,Min2018} protocol.}
\end{subfigure}%
\begin{subfigure}{.5\textwidth}
  \centering
  \includegraphics[width=0.95\linewidth]{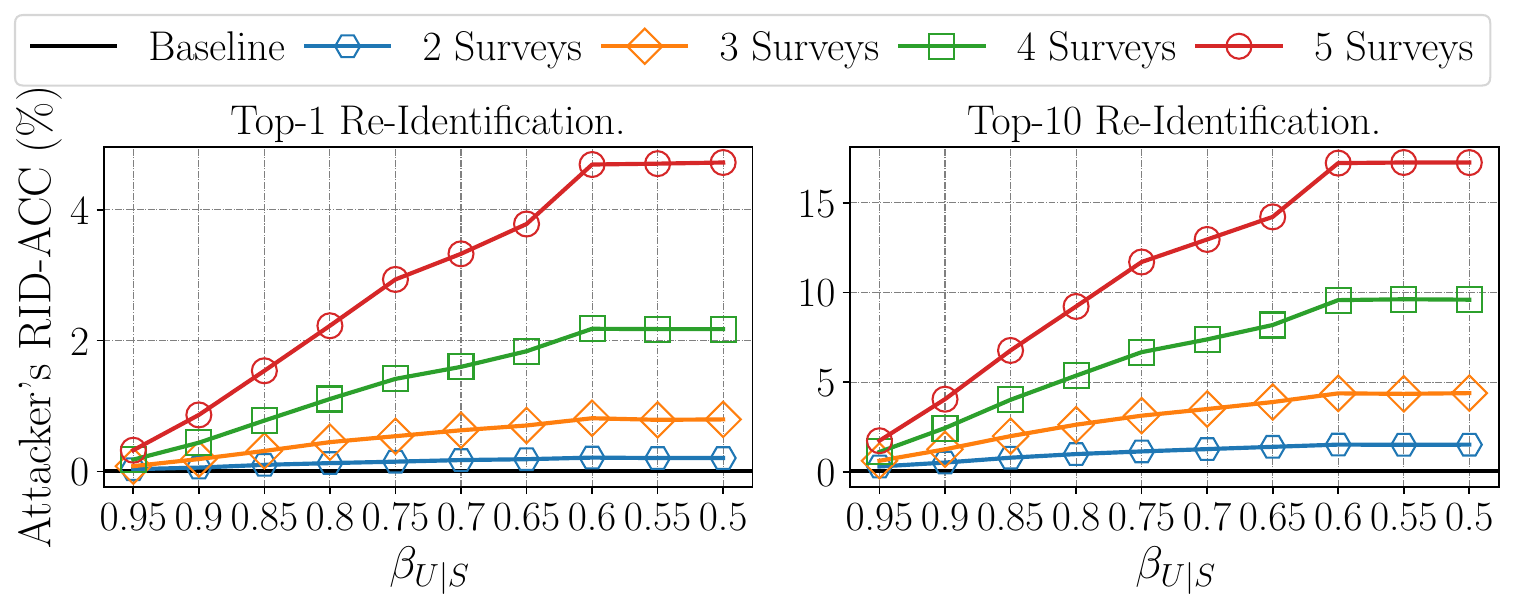}
  \caption{PK-RI risk of the $\omega$-SS~\cite{wang2016mutual,Min2018} protocol.}
\end{subfigure}
\\
\begin{subfigure}{.5\textwidth}
  \centering
  \includegraphics[width=0.95\linewidth]{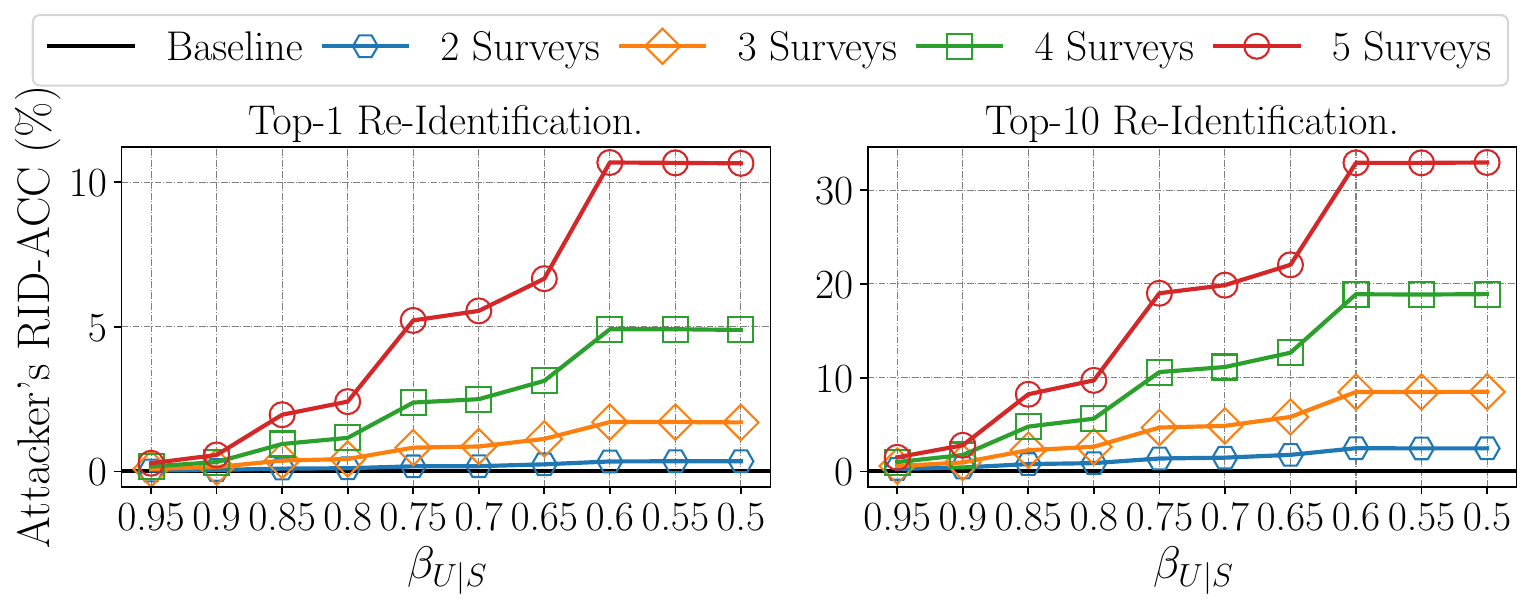}
  \caption{FK-RI risk of the OLH~\cite{tianhao2017} protocol.}
\end{subfigure}%
\begin{subfigure}{.5\textwidth}
  \centering
  \includegraphics[width=0.95\linewidth]{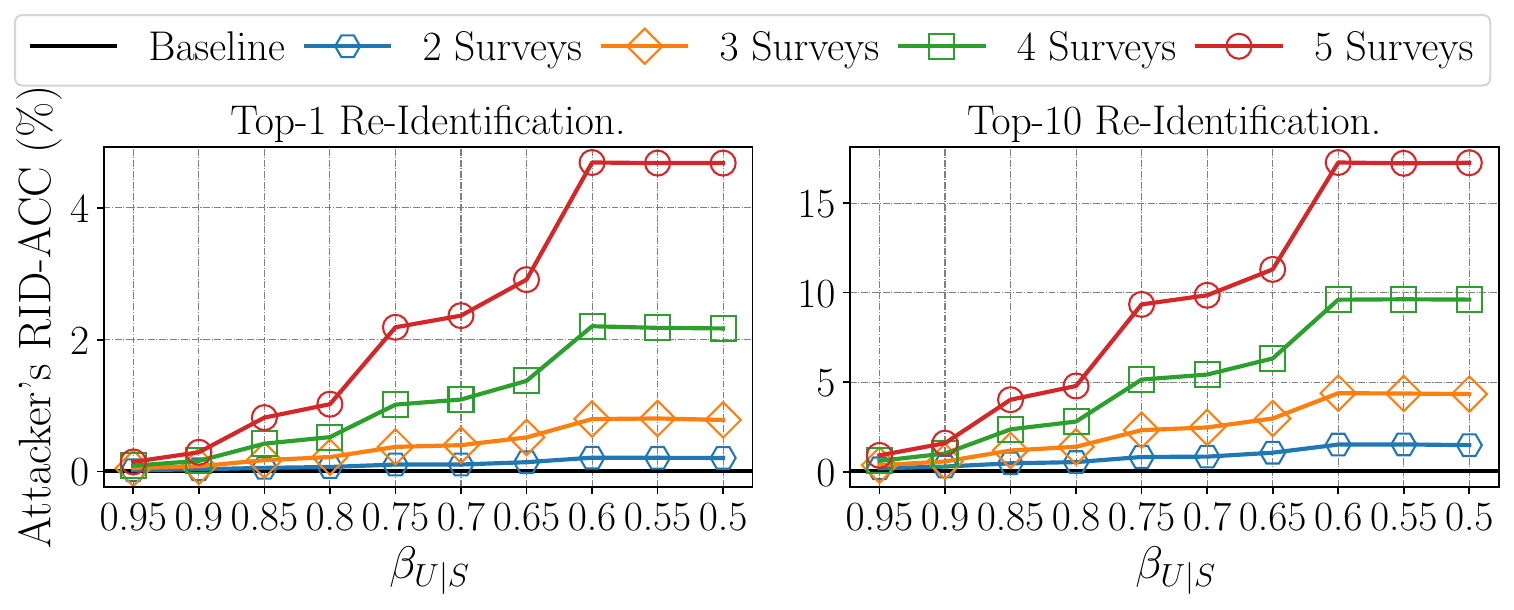}
  \caption{PK-RI risk of the OLH~\cite{tianhao2017} protocol.}
\end{subfigure}
\\
\begin{subfigure}{.5\textwidth}
  \centering
  \includegraphics[width=0.95\linewidth]{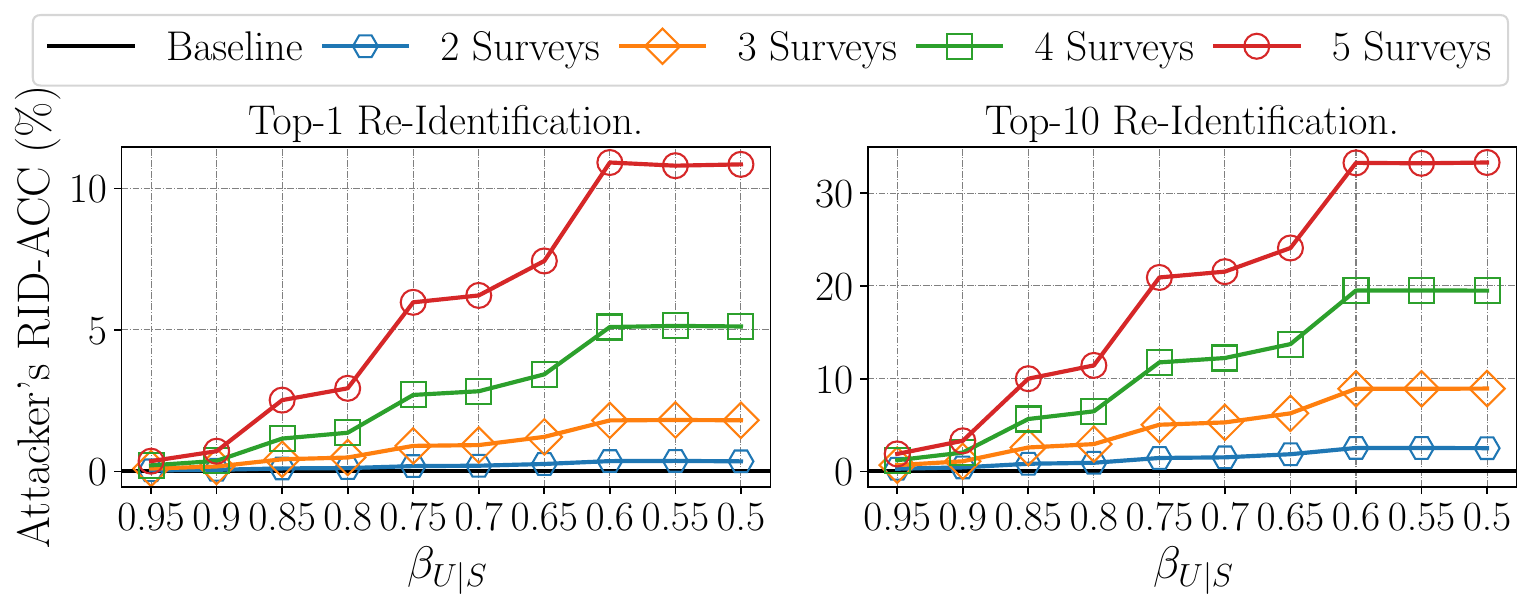}
  \caption{FK-RI risk of the OUE~\cite{tianhao2017} protocol.}
\end{subfigure}%
\begin{subfigure}{.5\textwidth}
  \centering
  \includegraphics[width=0.95\linewidth]{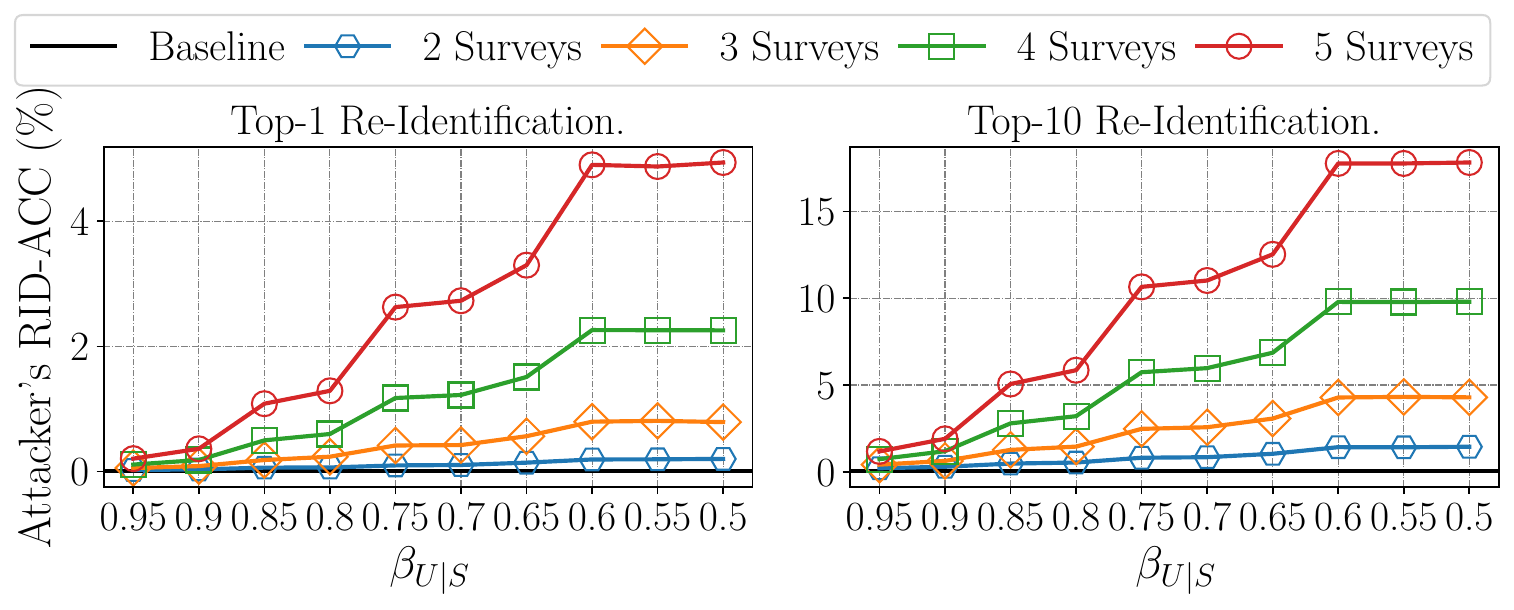}
  \caption{PK-RI risk of the OUE~\cite{tianhao2017} protocol.}
\end{subfigure}
\\
\begin{subfigure}{.5\textwidth}
  \centering
  \includegraphics[width=0.95\linewidth]{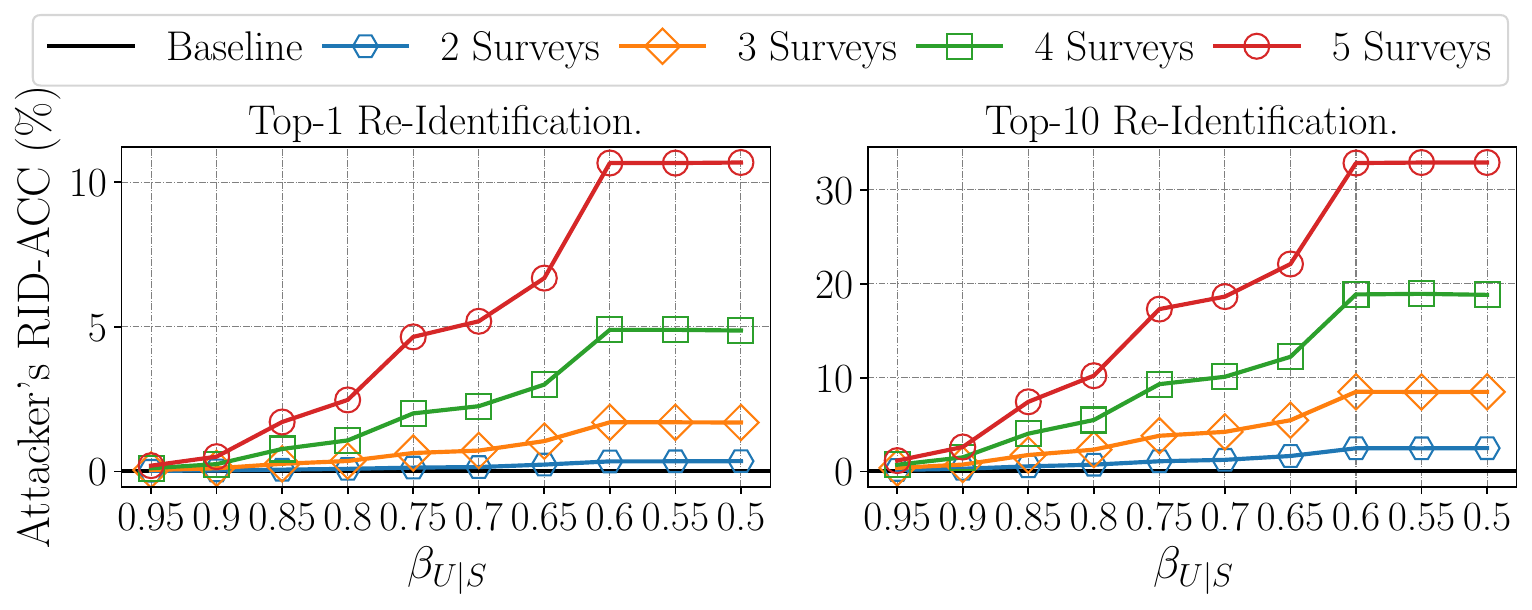}
  \caption{FK-RI risk of the SUE (\emph{a.k.a.} RAPPOR)~\cite{rappor} protocol.}
\end{subfigure}%
\begin{subfigure}{.5\textwidth}
  \centering
  \includegraphics[width=0.95\linewidth]{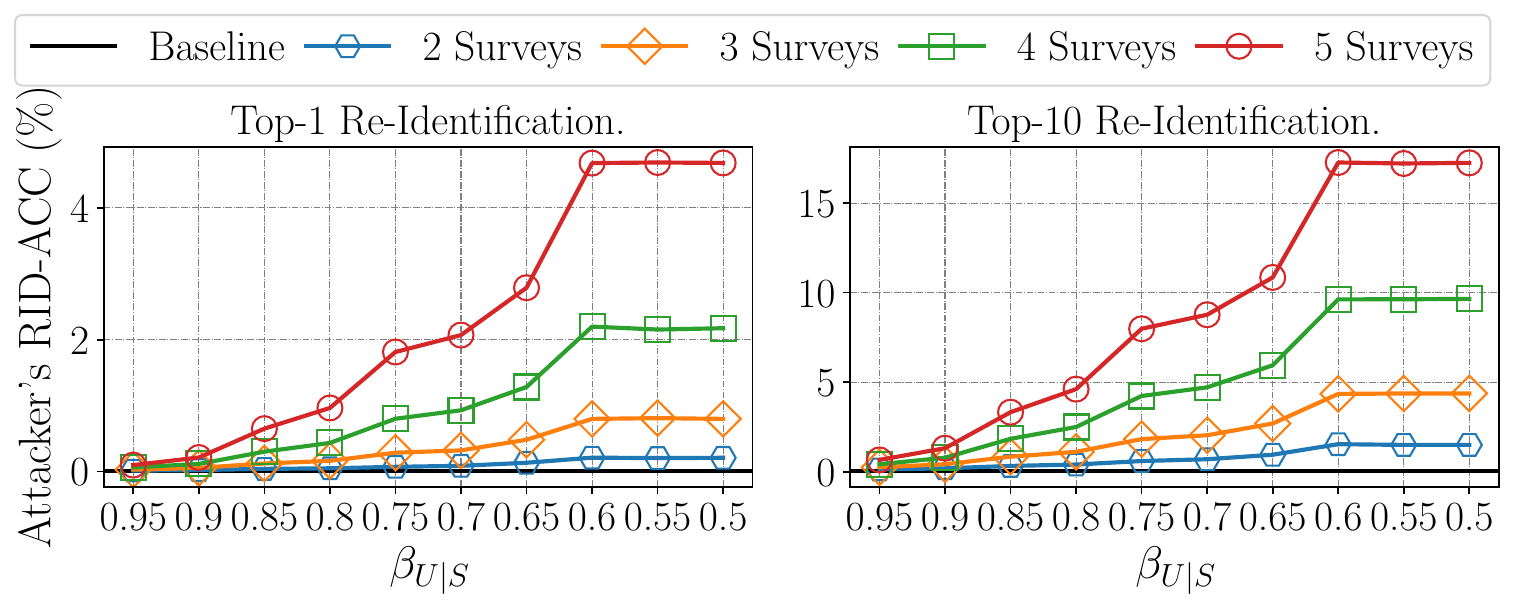}
  \caption{PK-RI risk of the SUE (\emph{a.k.a.} RAPPOR)~\cite{rappor} protocol.}
\end{subfigure}
\caption{Attacker's re-identification accuracy (RID-ACC) on the Adult dataset for top-k re-identification on using the SMP solution, the full knowledge FK-RI model (left-side plots) and partial knowledge PK-RI model (right-side plots) with uniform $\alpha$-PIE privacy metric across users, and by varying the LDP protocol and the number of surveys (i.e., data collections).}
\label{fig:reident_smp_uni_alpha}
\end{figure*}

\begin{figure*}[!ht]
\begin{subfigure}{.5\textwidth}
  \centering
  \includegraphics[width=0.95\linewidth]{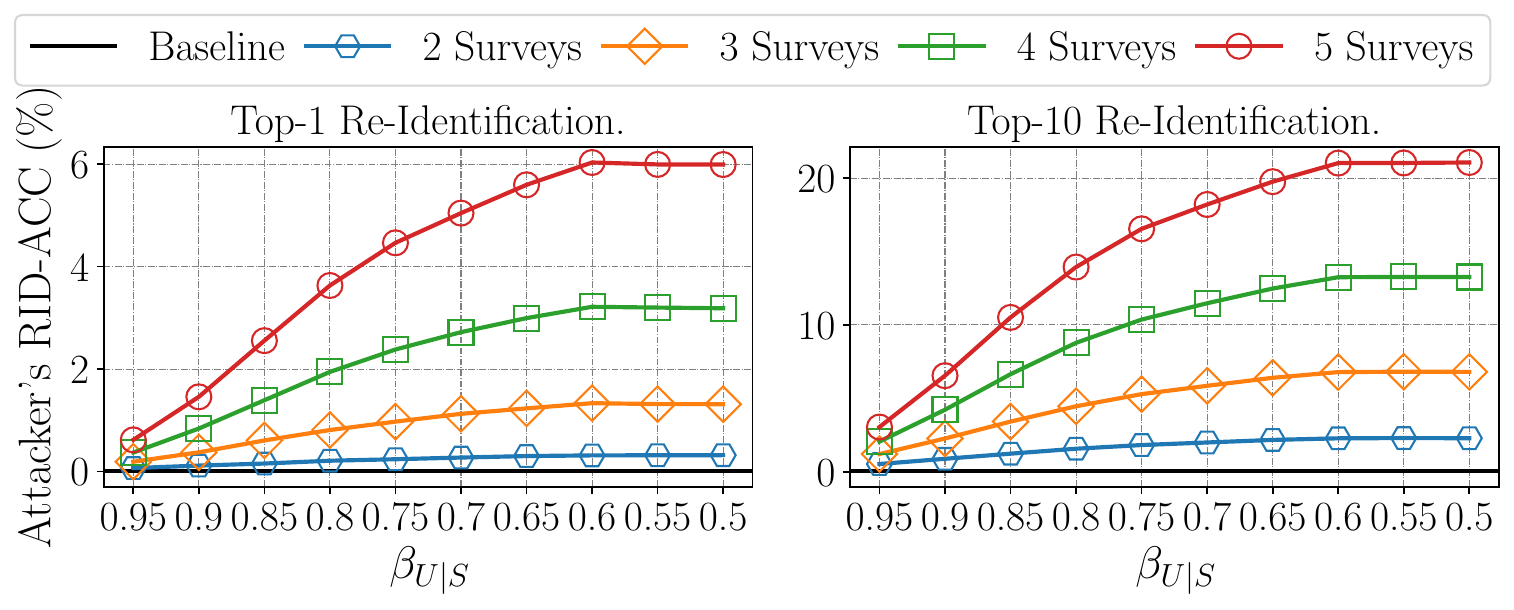}
  \caption{FK-RI risk of the GRR~\cite{kairouz2016discrete,kairouz2016extremal} protocol.}
\end{subfigure}%
\begin{subfigure}{.5\textwidth}
  \centering
  \includegraphics[width=0.95\linewidth]{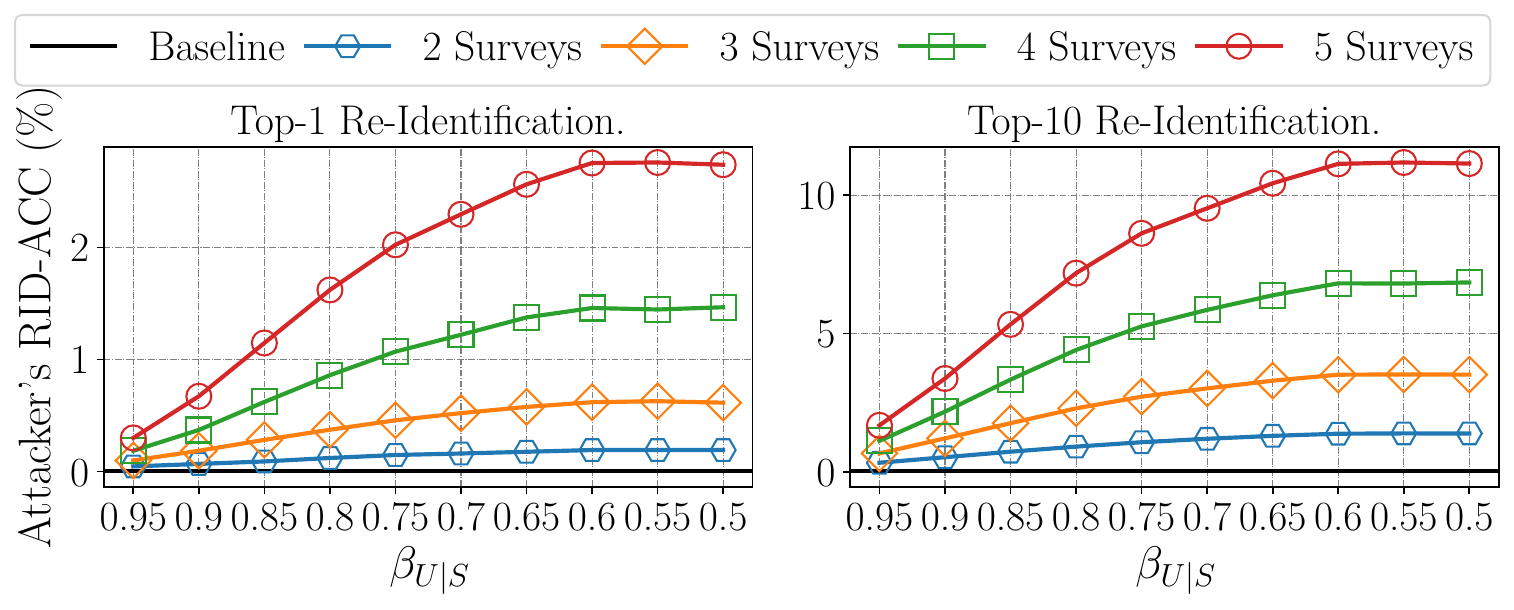}
  \caption{PK-RI risk of the GRR~\cite{kairouz2016discrete,kairouz2016extremal} protocol.}
\end{subfigure}
\\
\begin{subfigure}{.5\textwidth}
  \centering
  \includegraphics[width=0.95\linewidth]{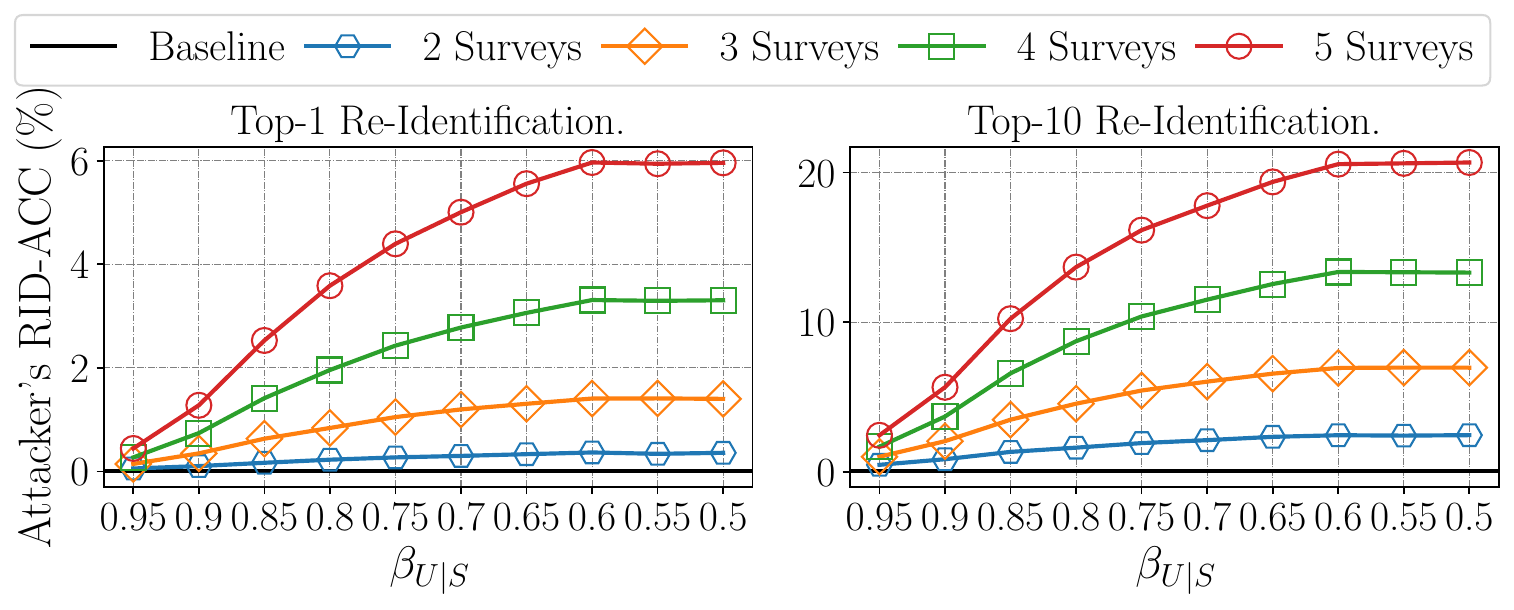}
  \caption{FK-RI risk of the $\omega$-SS~\cite{wang2016mutual,Min2018} protocol.}
\end{subfigure}%
\begin{subfigure}{.5\textwidth}
  \centering
  \includegraphics[width=0.95\linewidth]{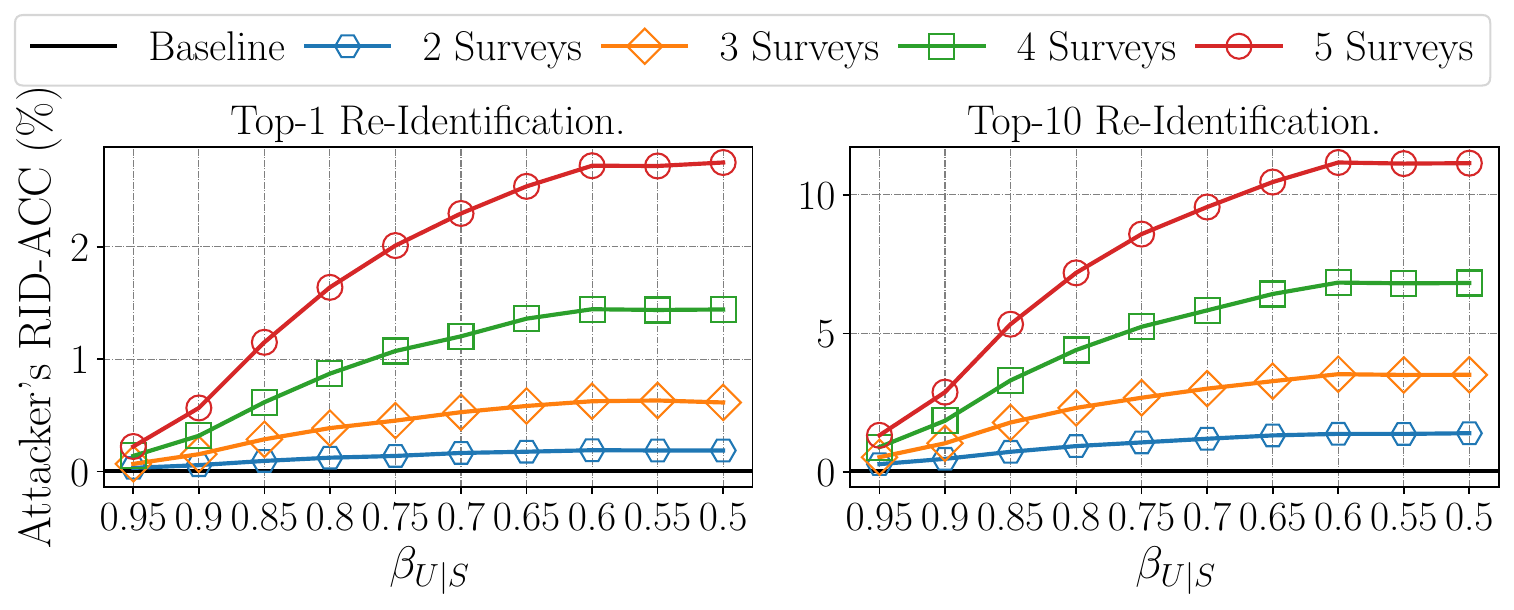}
  \caption{PK-RI risk of the $\omega$-SS~\cite{wang2016mutual,Min2018} protocol.}
\end{subfigure}
\\
\begin{subfigure}{.5\textwidth}
  \centering
  \includegraphics[width=0.95\linewidth]{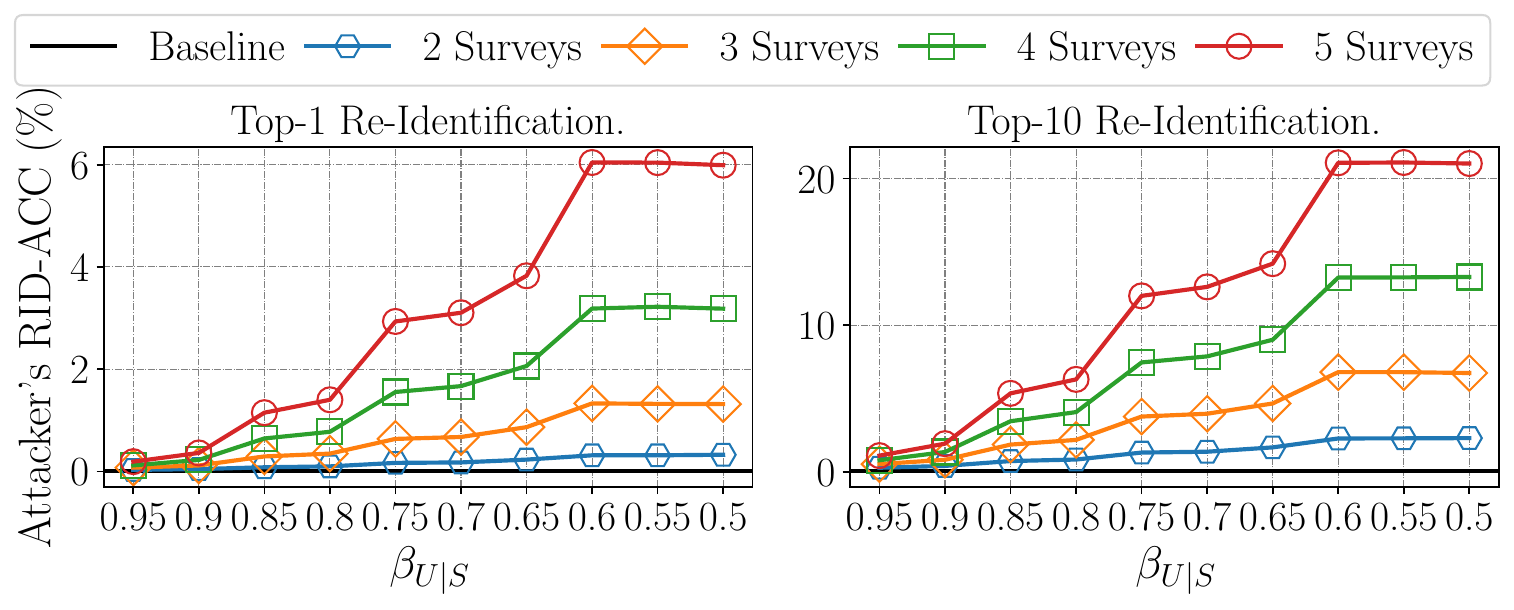}
  \caption{FK-RI risk of the OLH~\cite{tianhao2017} protocol.}
\end{subfigure}%
\begin{subfigure}{.5\textwidth}
  \centering
  \includegraphics[width=0.95\linewidth]{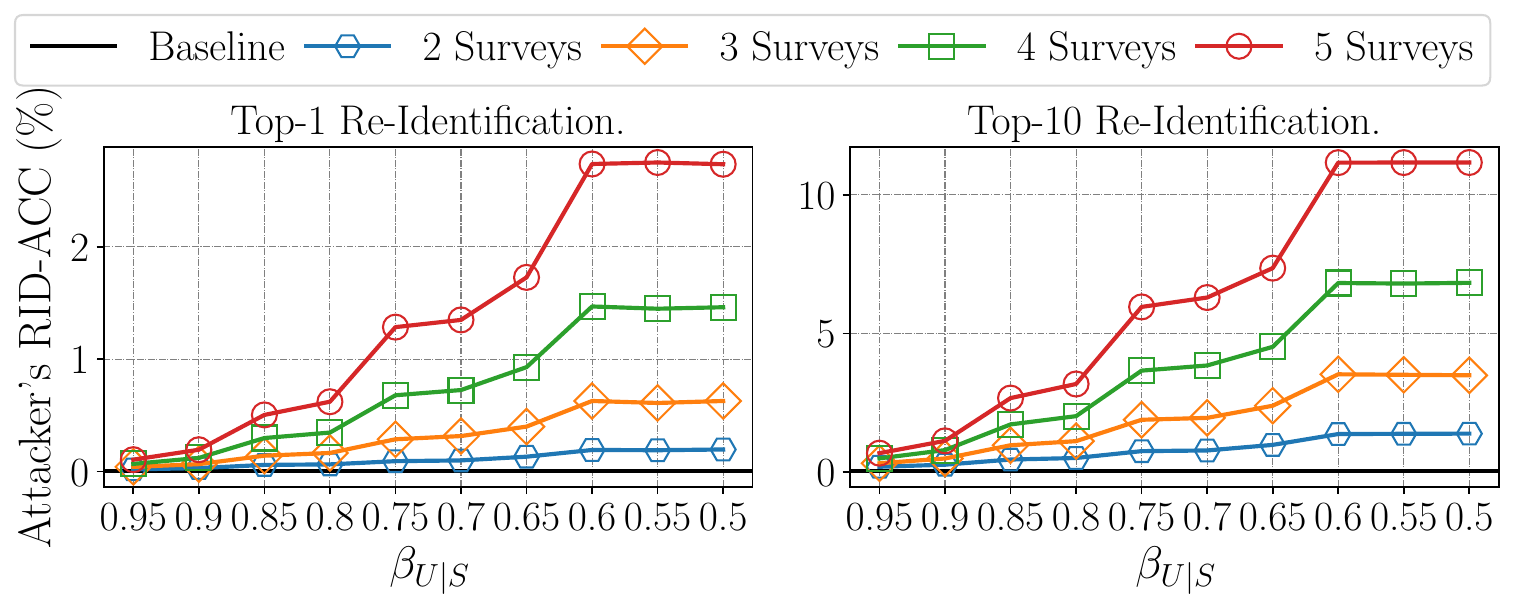}
  \caption{PK-RI risk of the OLH~\cite{tianhao2017} protocol.}
\end{subfigure}
\\
\begin{subfigure}{.5\textwidth}
  \centering
  \includegraphics[width=0.95\linewidth]{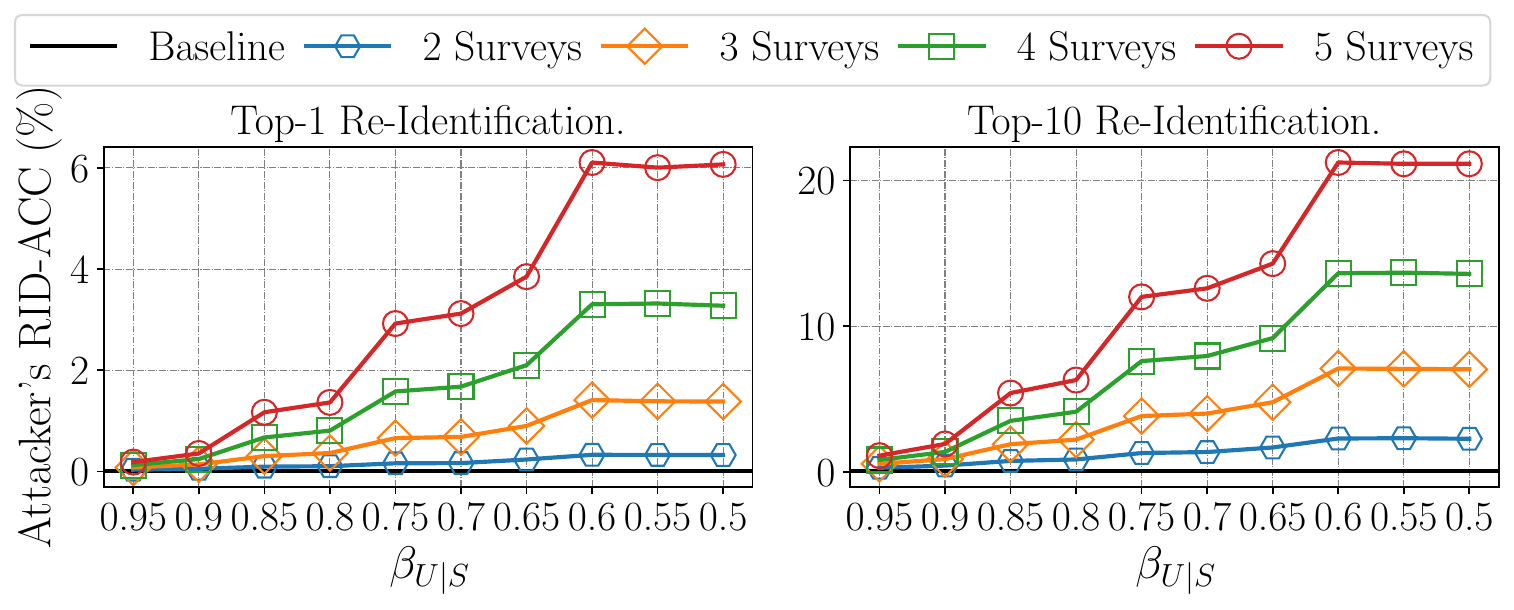}
  \caption{FK-RI risk of the OUE~\cite{tianhao2017} protocol.}
\end{subfigure}%
\begin{subfigure}{.5\textwidth}
  \centering
  \includegraphics[width=0.95\linewidth]{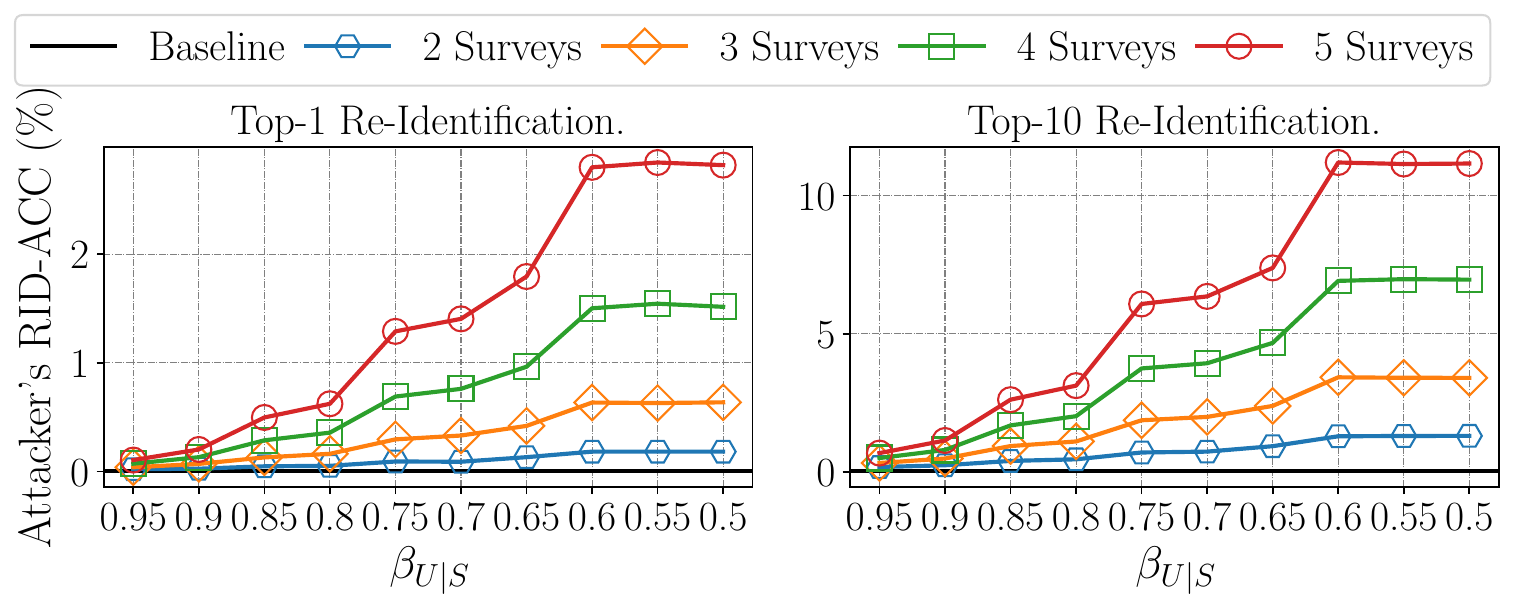}
  \caption{PK-RI risk of the OUE~\cite{tianhao2017} protocol.}
\end{subfigure}
\\
\begin{subfigure}{.5\textwidth}
  \centering
  \includegraphics[width=0.95\linewidth]{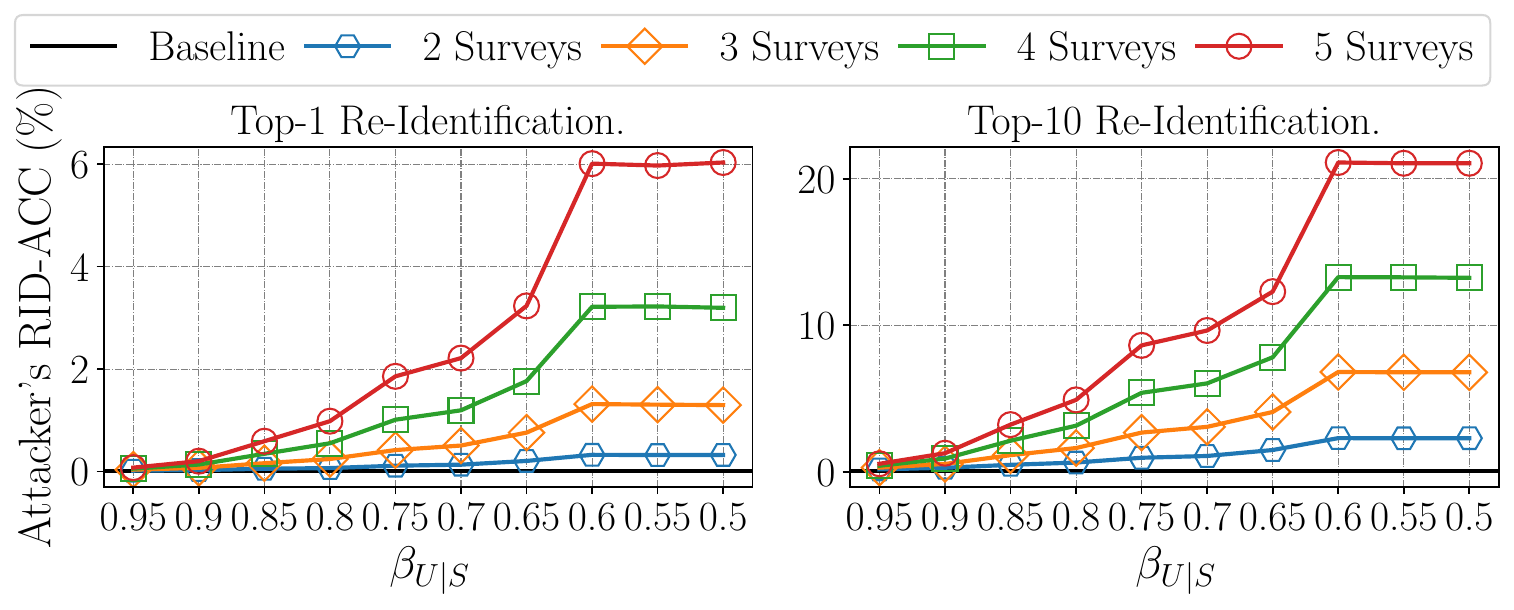}
  \caption{FK-RI risk of the SUE (\emph{a.k.a.} RAPPOR)~\cite{rappor} protocol.}
\end{subfigure}%
\begin{subfigure}{.5\textwidth}
  \centering
  \includegraphics[width=0.95\linewidth]{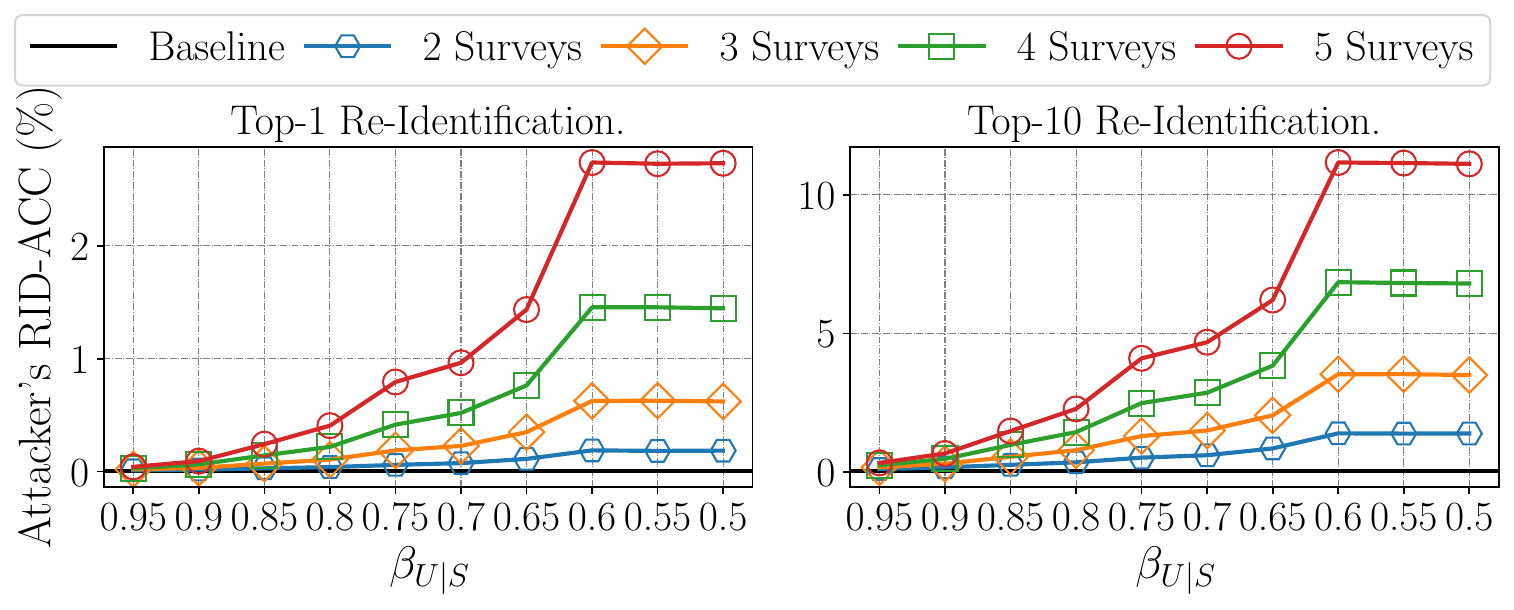}
  \caption{PK-RI risk of the SUE (\emph{a.k.a.} RAPPOR)~\cite{rappor} protocol.}
\end{subfigure}
\caption{Attacker's re-identification accuracy (RID-ACC) on the Adult dataset for top-k re-identification on using the SMP solution, the full knowledge FK-RI model (left-side plots) and partial knowledge PK-RI model (right-side plots) with non-uniform $\alpha$-PIE privacy metric across users, and by varying the LDP protocol and the number of surveys (i.e., data collections).}
\label{fig:reident_smp_nonuni_alpha}
\end{figure*}

\section{Additional Results for Section~\ref{sub:results_att_inf_rspfd}} \label{appD:add_RSpFD}

This section provides additional results for the inference of the sampled attribute on collecting multidimensional data with the RS+FD~\cite{Arcolezi2021_rs_fd} solution. 
We use the state-of-the-art XGBoost~\cite{XGBoost} algorithm to predict the sampled attribute of users in a multiclass classification framework (i.e., $d$ attributes) with default parameters. 
We follow the experimental evaluation described in Section~\ref{sub:results_att_inf_rspfd} and we vary the following:

\begin{itemize}

    \item \textbf{Dataset.} We use the Adult ($d=10$ attributes, $n=45,222$ and $\textbf{k}=[74, 7, 16, 7, 14, 6, 5, 2, 41, 2]$) and Nursery ($d=9$ attributes, $n=12,959$ and $\textbf{k}=[3, 5, 4, 4, 3, 2, 3, 3, 5]$) datasets from the UCI ML repository~\cite{uci}.
    
    \item \textbf{LDP protocol within RS+FD.} All protocols from Section~\ref{sub:rspfd_sol}, namely, RS+FD[GRR], RS+FD[SUE-z], RS+FD[SUE-r], RS+FD[OUE-z] and RS+FD[OUE-r]. 
    
    \item \textbf{Attribute inference model.} All five protocols are evaluated with the three attack models of Section~\ref{sub:atk_models_rspfd}, namely, No Knowledge (NK), Partial-Knowledge (PK) and Hybrid Model (HM).

\end{itemize}

Figs.~\ref{fig:attack_rspfd_adult} and~\ref{fig:attack_rspfd_nursery} illustrates the attacker's attribute inference accuracy (AIF-ACC) metric on the Adult and Nursery datasets, respectively, with the three attack models (i.e., NK, PK and HM) and all five protocols (i.e., RS+FD[GRR], RS+FD[SUE-z], RS+FD[OUE-z], RS+FD[SUE-r] and RS+FD[OUE-r]), varying $\epsilon$, the number of synthetic profiles $s$ and the number of compromised profiles $n_{pk}$. 

Similar to the results with the ACSEmployement~\cite{ding2021retiring} dataset in Fig.~\ref{fig:attack_rspfd}, one can notice in Figs.~\ref{fig:attack_rspfd_adult} and~\ref{fig:attack_rspfd_nursery} that the proposed attack models, namely, NK, PK and HM present significant increments in the attacker's AIF-ACC over the Baseline model. More precisely, with the Adult dataset, there is about a 1.3-10 fold increment over a random Baseline model with our NK, PK and HM models. On the one hand, there is about a 0.1-10 fold increment with the Nursery dataset. More precisely, the attack models with both RS+FD[GRR] and RS+FD[UE-r] protocols did not provide a meaningful increment over the Baseline model in the Nursery dataset. The reason behind this is that the attributes follow uniform-like distributions. Thus, since fake data are also generated uniformly at random with the RS+FD solution, the classifier is not able to distinguish between real and fake data when predicting the sampled attribute. Yet, the attacker's AIF-ACC also achieves about $100\%$ with RS+FD[SUE-z] with all three datasets (see Figs.~\ref{fig:attack_rspfd} and.~\ref{fig:attack_rspfd_adult}). Lastly, increasing the number of synthetic profiles the attacker generates $s$ and/or the number of compromised profiles the attacker has access to $n_{pk}$, had few influence with the Adult dataset in Fig.~\ref{fig:attack_rspfd_adult}. Conversely, both ACSEmployement (Fig.~\ref{fig:attack_rspfd}) and Nursery (Fig.~\ref{fig:attack_rspfd_nursery}) datasets showed sensitivity to a change in both parameters, especially with $n_{pk}$ in the PK model.

\begin{figure*}
\begin{subfigure}{.33\textwidth}
  \centering
  \includegraphics[width=1\linewidth]{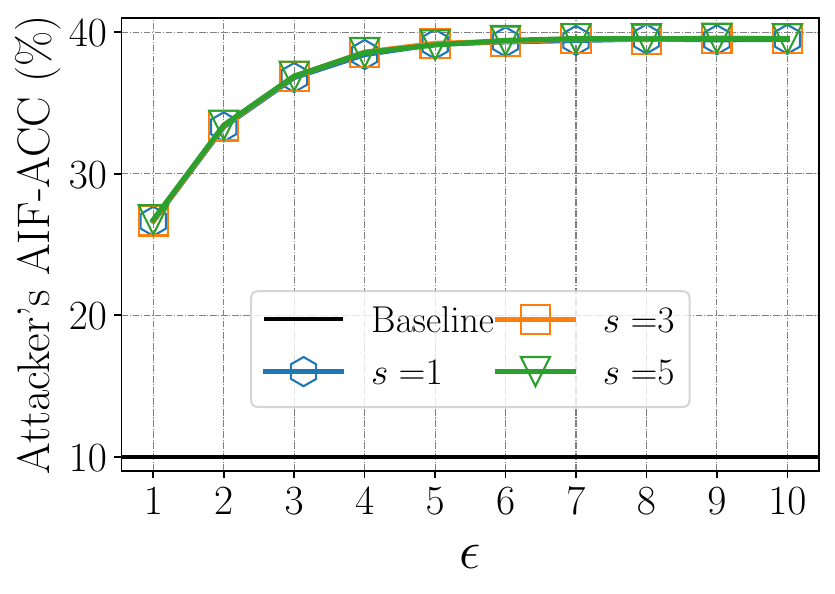}
  \caption{NK model with RS+FD[GRR] protocol.}
\end{subfigure}%
\begin{subfigure}{.33\textwidth}
  \centering
  \includegraphics[width=1\linewidth]{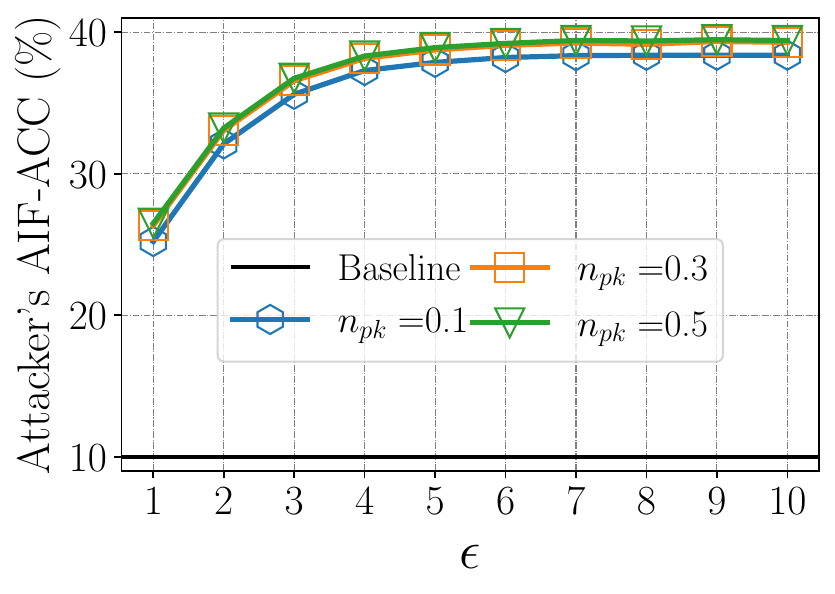}
  \caption{PK model with RS+FD[GRR] protocol.}
\end{subfigure}
\begin{subfigure}{.33\textwidth}
  \centering
  \includegraphics[width=1\linewidth]{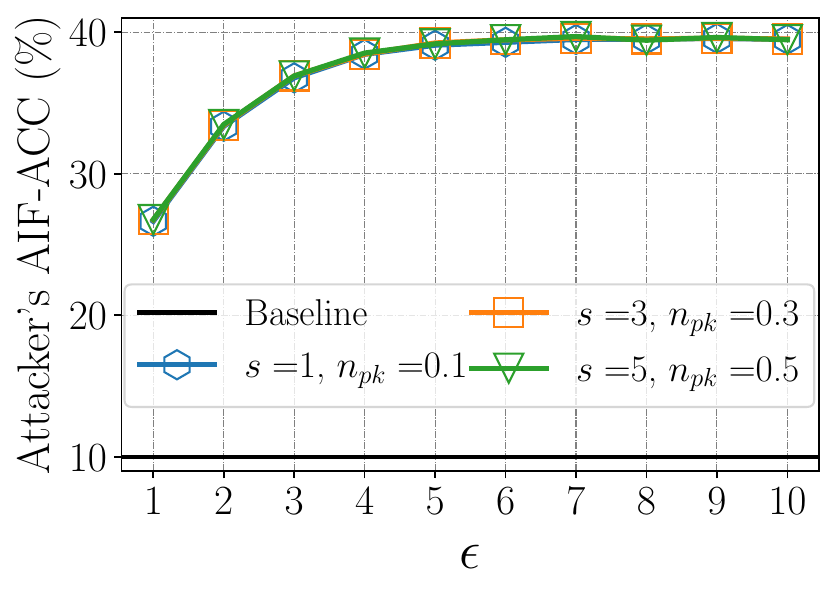}
  \caption{Hybrid model with RS+FD[GRR] protocol.}
\end{subfigure}
\\
\begin{subfigure}{.33\textwidth}
  \centering
  \includegraphics[width=1\linewidth]{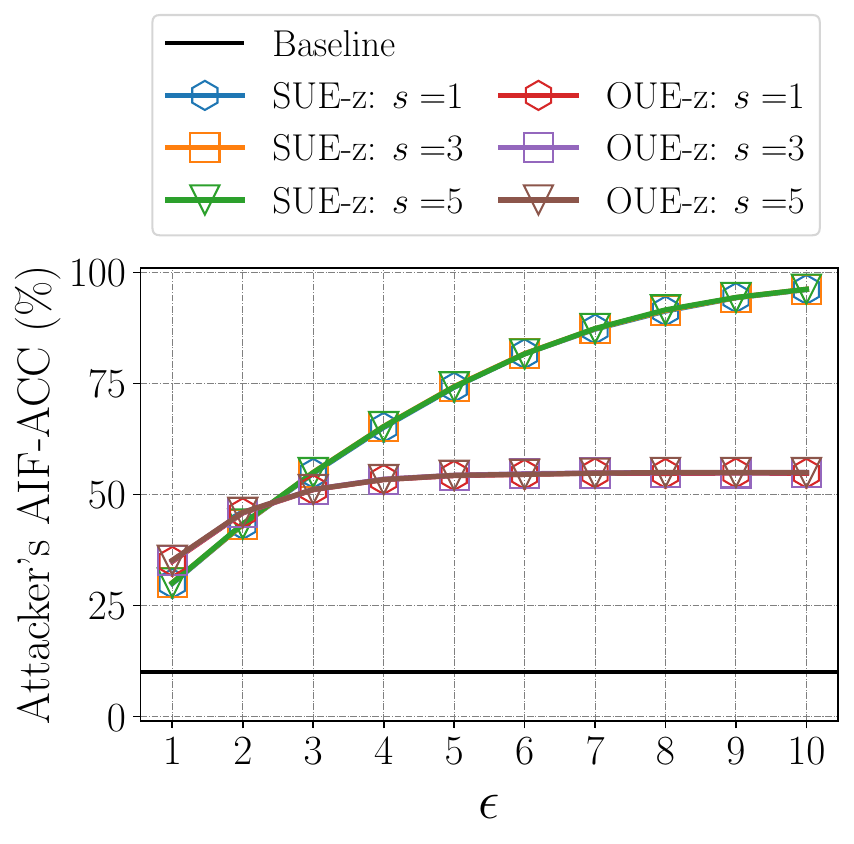}
  \caption{NK model with RS+FD[UE-z] protocols.}
\end{subfigure}%
\begin{subfigure}{.33\textwidth}
  \centering
  \includegraphics[width=1\linewidth]{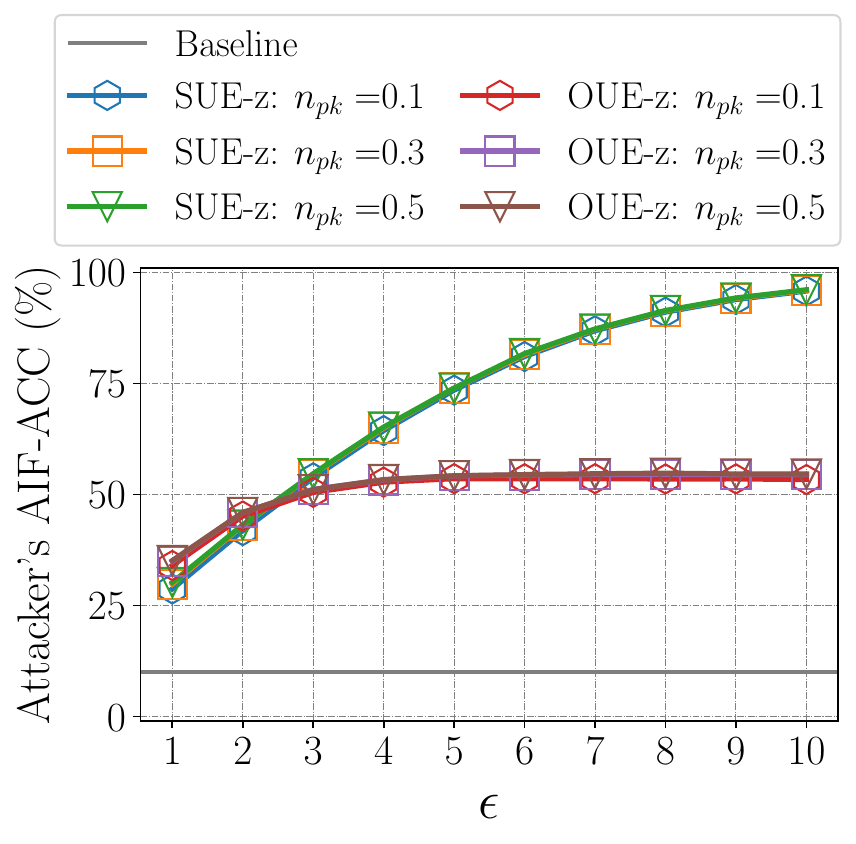}
  \caption{PK model with RS+FD[UE-z] protocols.}
\end{subfigure}
\begin{subfigure}{.33\textwidth}
  \centering
  \includegraphics[width=1\linewidth]{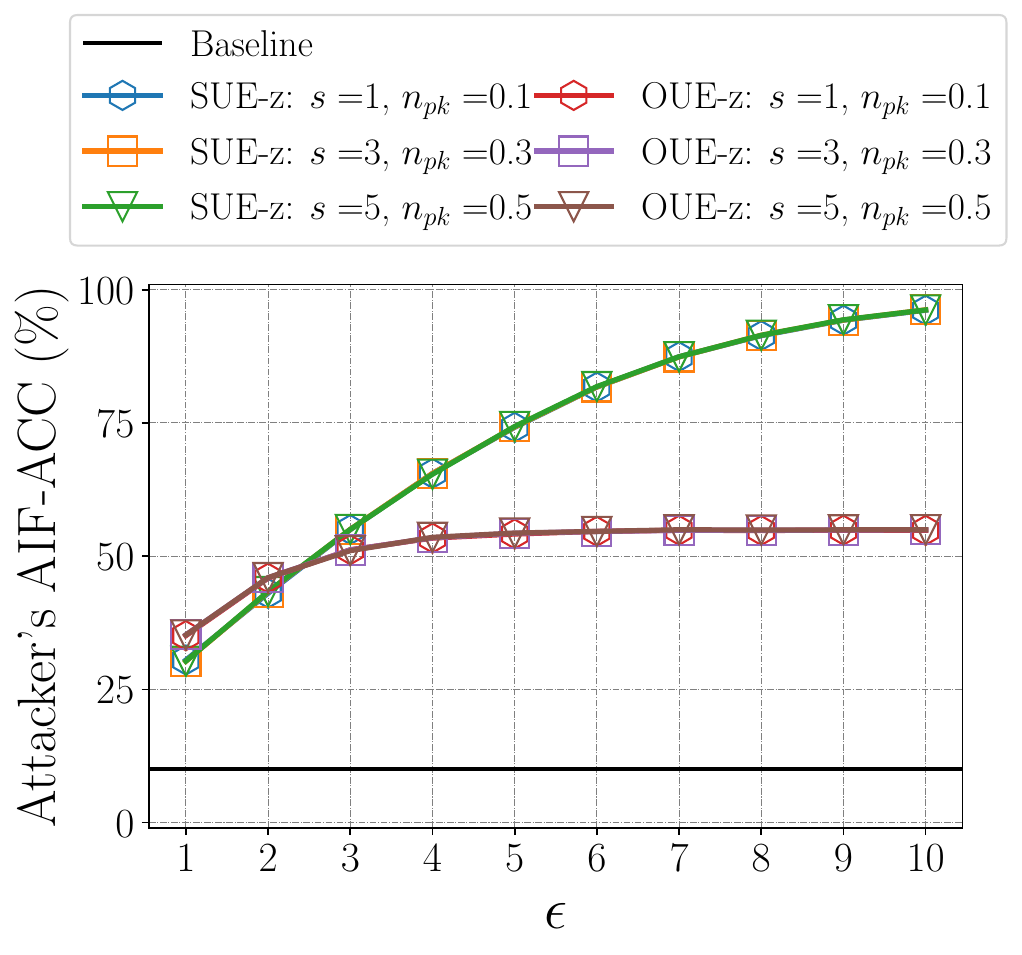}
  \caption{Hybrid model with RS+FD[UE-z] protocols.}
\end{subfigure}
\\
\begin{subfigure}{.33\textwidth}
  \centering
  \includegraphics[width=1\linewidth]{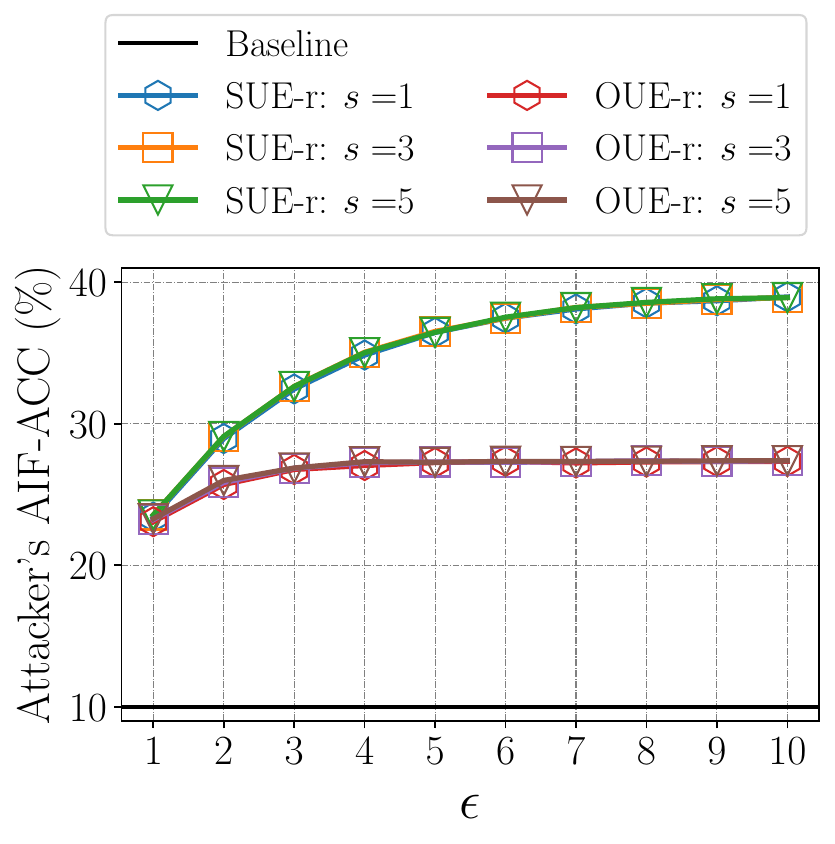}
  \caption{NK model with RS+FD[UE-r] protocols.}
\end{subfigure}%
\begin{subfigure}{.33\textwidth}
  \centering
  \includegraphics[width=1\linewidth]{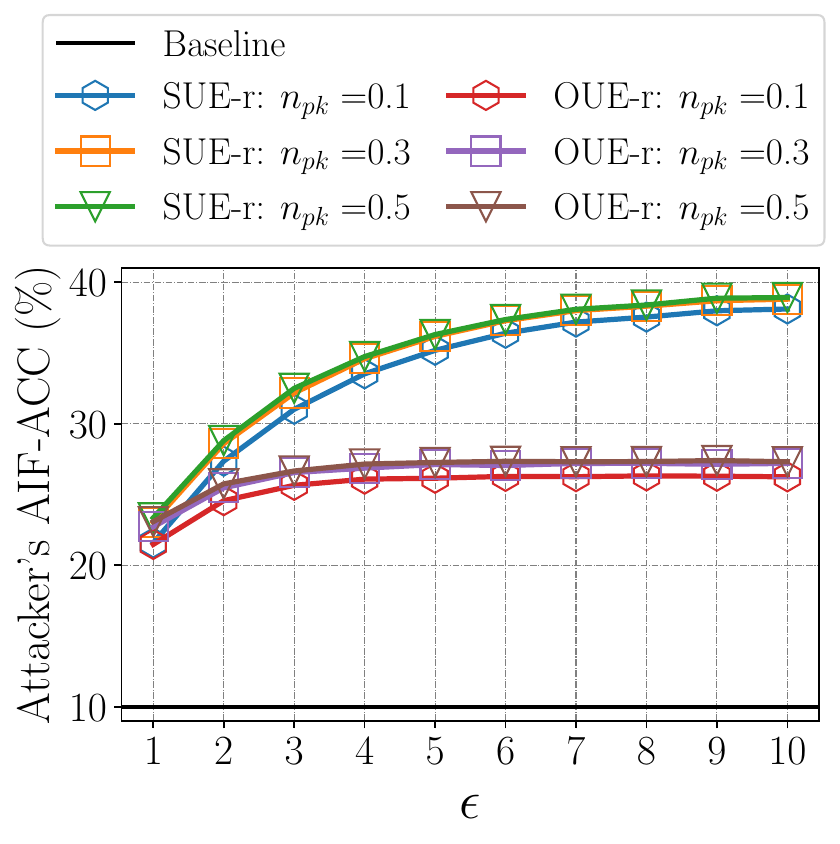}
  \caption{PK model with RS+FD[UE-r] protocols.}
\end{subfigure}
\begin{subfigure}{.33\textwidth}
  \centering
  \includegraphics[width=1\linewidth]{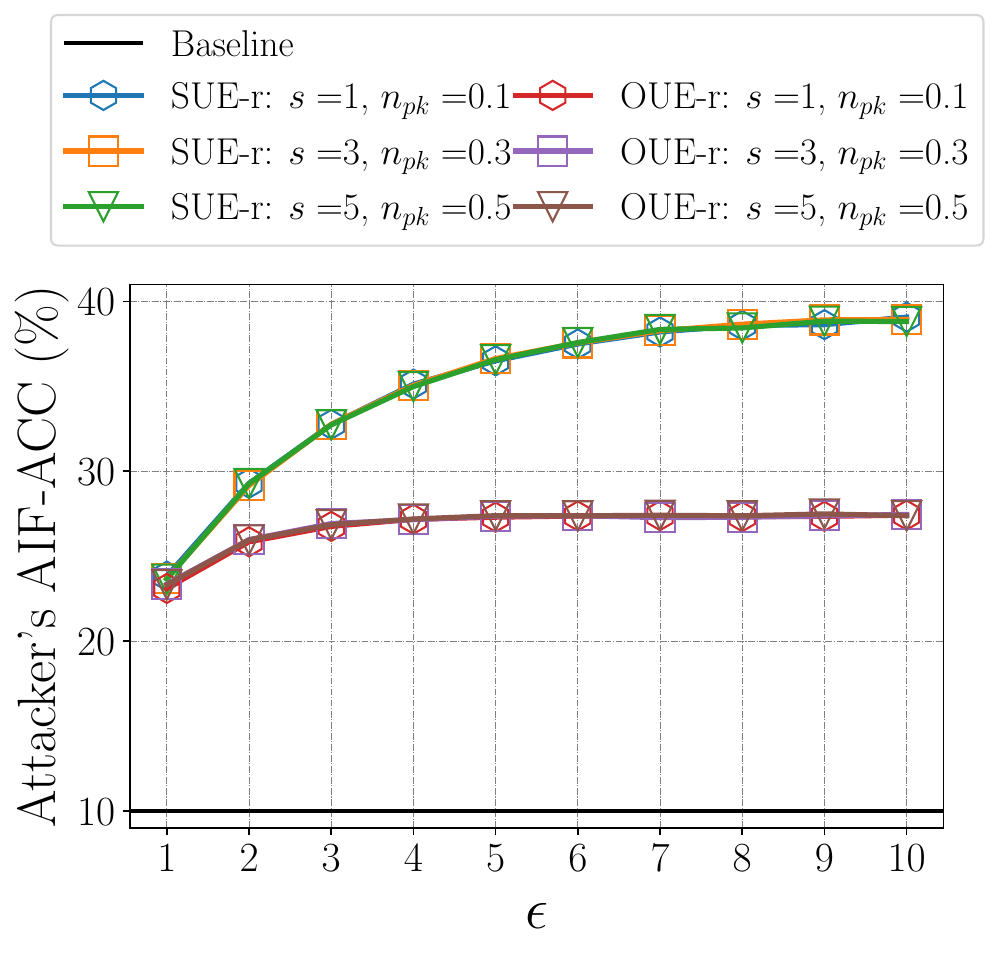}
  \caption{Hybrid model with RS+FD[UE-r] protocols.}
\end{subfigure}
\caption{Attacker's attribute inference accuracy (AIF-ACC) on the Adult dataset with three attack models (i.e., NK, PK and hybrid) and five protocols (i.e., RS+FD[GRR], RS+FD[SUE-z], RS+FD[OUE-z], RS+FD[SUE-r] and RS+FD[OUE-r]), varying $\epsilon$, the number of synthetic profiles $s$ the attacker generates and the number of compromised profiles $n_{pk}$ the attacker has access to.}
\label{fig:attack_rspfd_adult}
\end{figure*}

\begin{figure*}
\begin{subfigure}{.33\textwidth}
  \centering
  \includegraphics[width=1\linewidth]{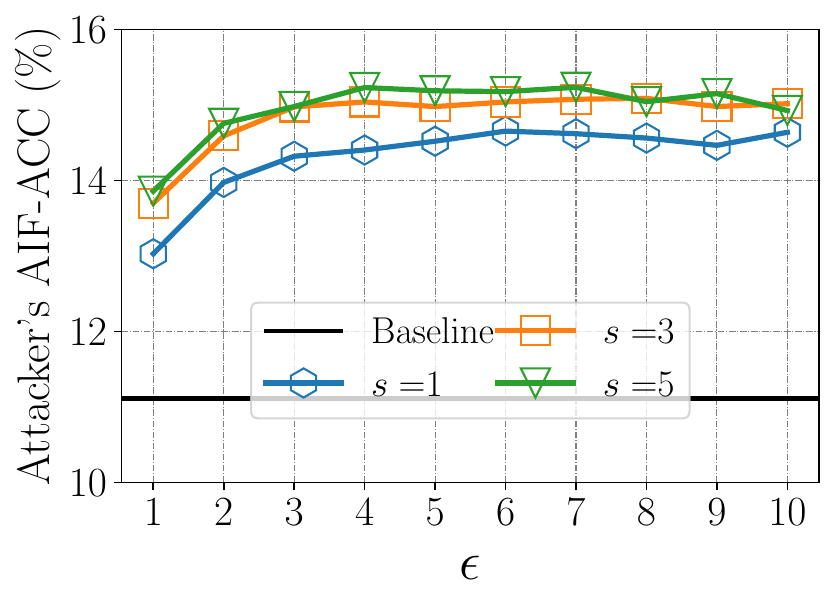}
  \caption{NK model with RS+FD[GRR] protocol.}
\end{subfigure}%
\begin{subfigure}{.33\textwidth}
  \centering
  \includegraphics[width=1\linewidth]{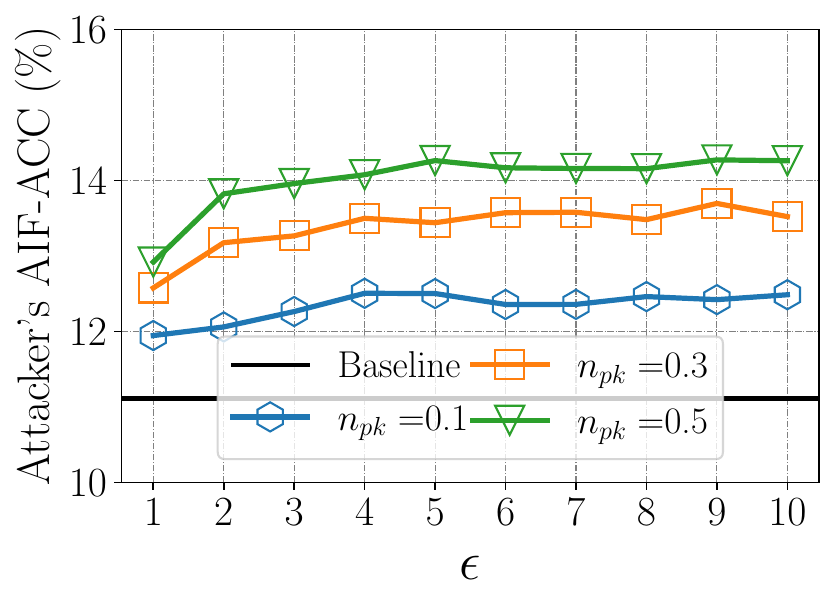}
  \caption{PK model with RS+FD[GRR] protocol.}
\end{subfigure}
\begin{subfigure}{.33\textwidth}
  \centering
  \includegraphics[width=1\linewidth]{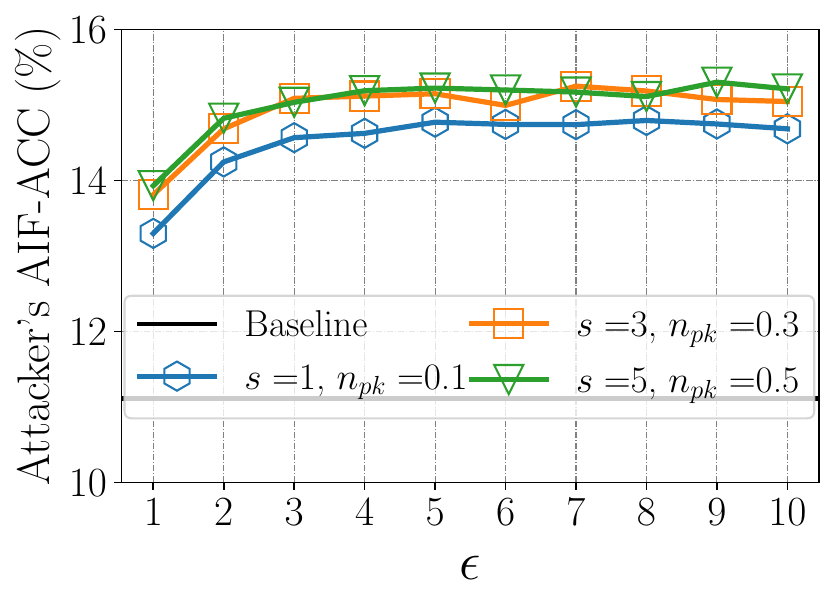}
  \caption{Hybrid model with RS+FD[GRR] protocol.}
\end{subfigure}
\\
\begin{subfigure}{.33\textwidth}
  \centering
  \includegraphics[width=1\linewidth]{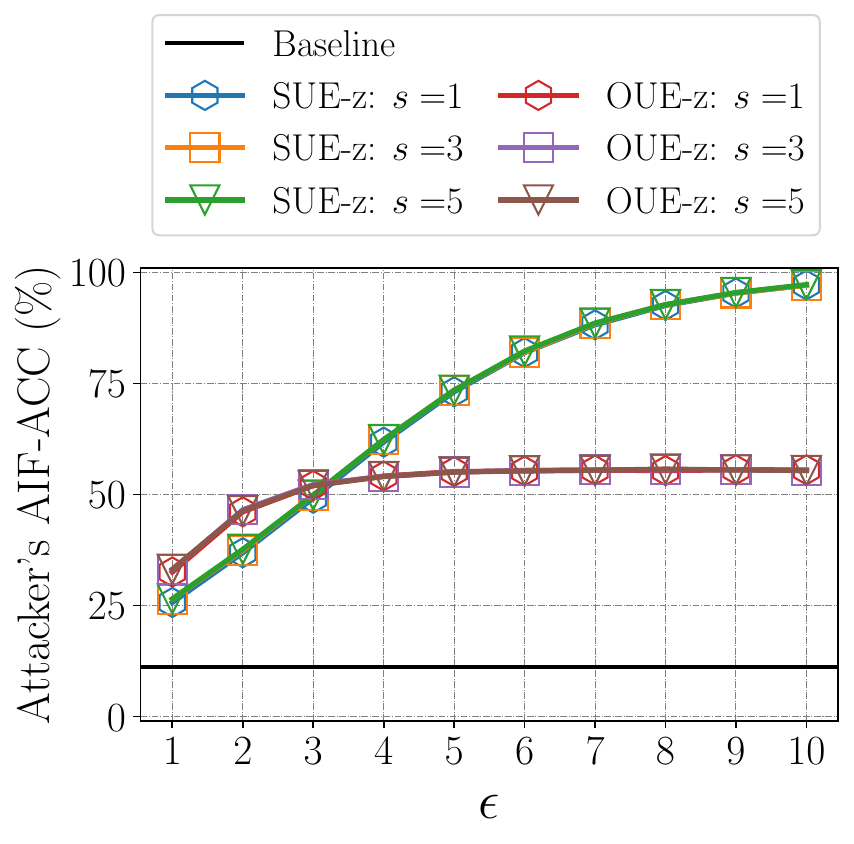}
  \caption{NK model with RS+FD[UE-z] protocols.}
\end{subfigure}%
\begin{subfigure}{.33\textwidth}
  \centering
  \includegraphics[width=1\linewidth]{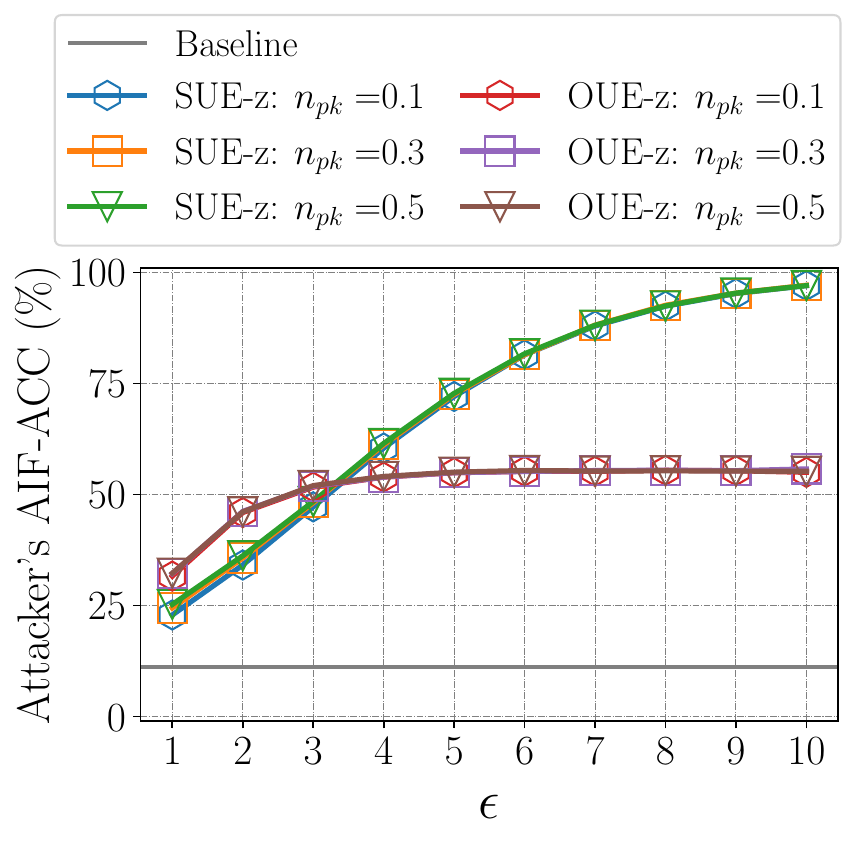}
  \caption{PK model with RS+FD[UE-z] protocols.}
\end{subfigure}
\begin{subfigure}{.33\textwidth}
  \centering
  \includegraphics[width=1\linewidth]{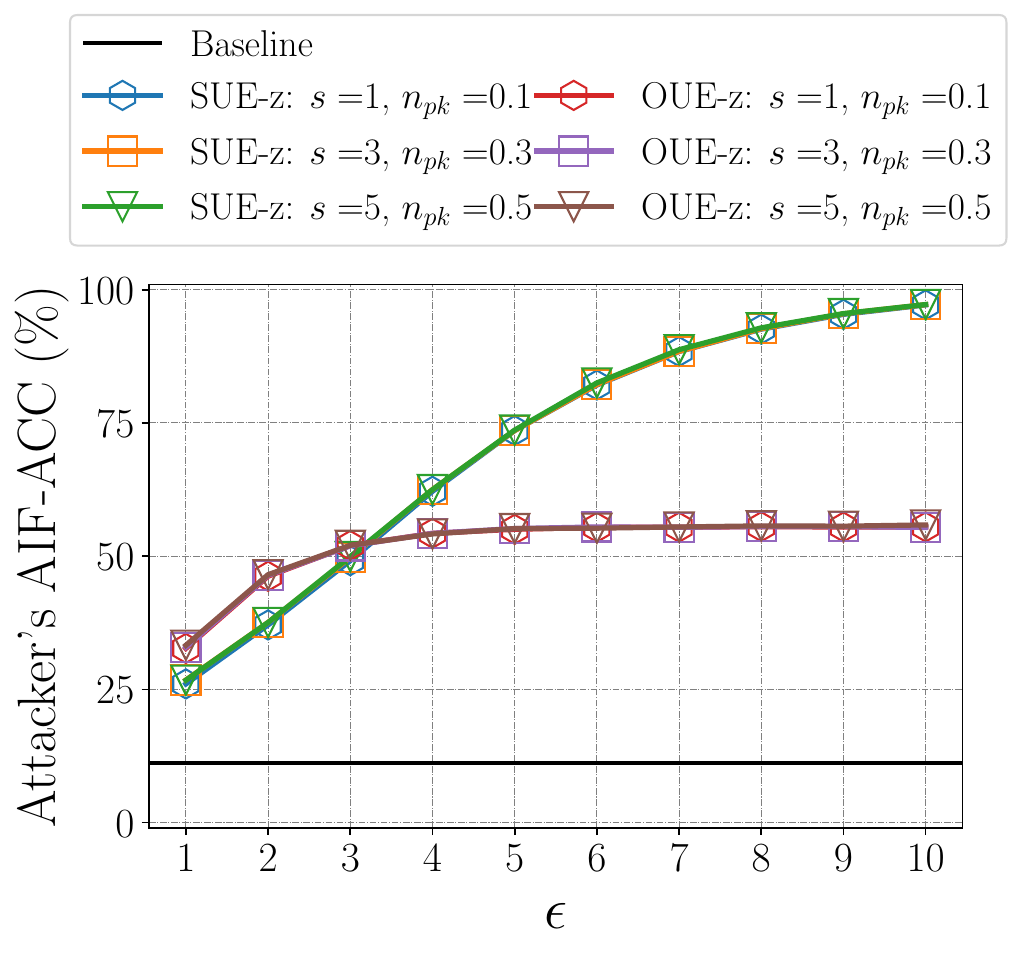}
  \caption{Hybrid model with RS+FD[UE-z] protocols.}
\end{subfigure}
\\
\begin{subfigure}{.33\textwidth}
  \centering
  \includegraphics[width=1\linewidth]{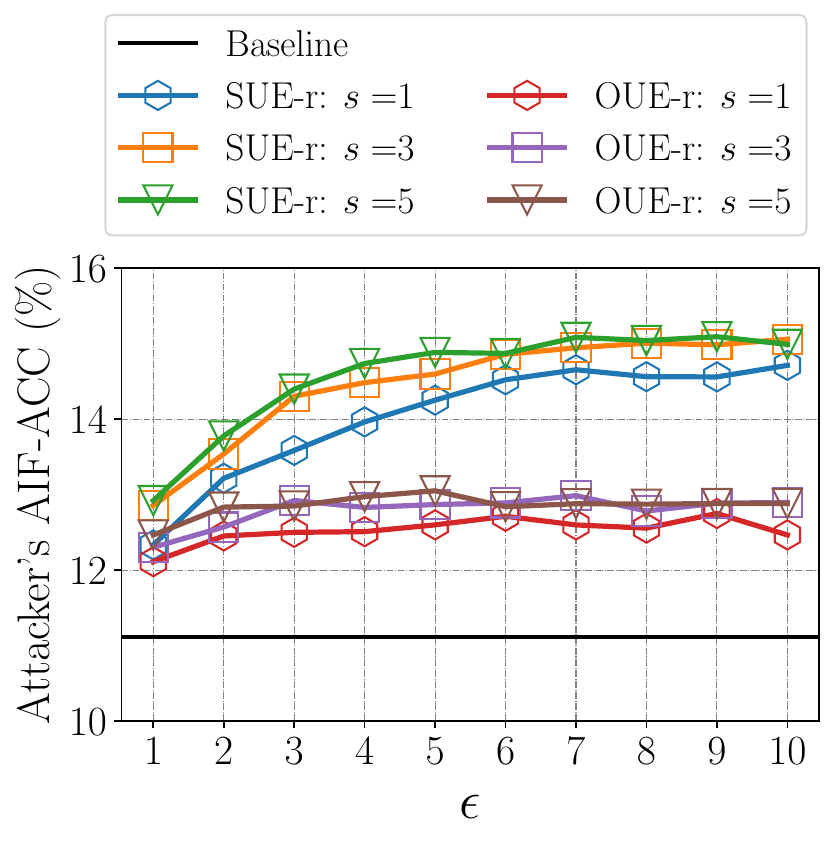}
  \caption{NK model with RS+FD[UE-r] protocols.}
\end{subfigure}%
\begin{subfigure}{.33\textwidth}
  \centering
  \includegraphics[width=1\linewidth]{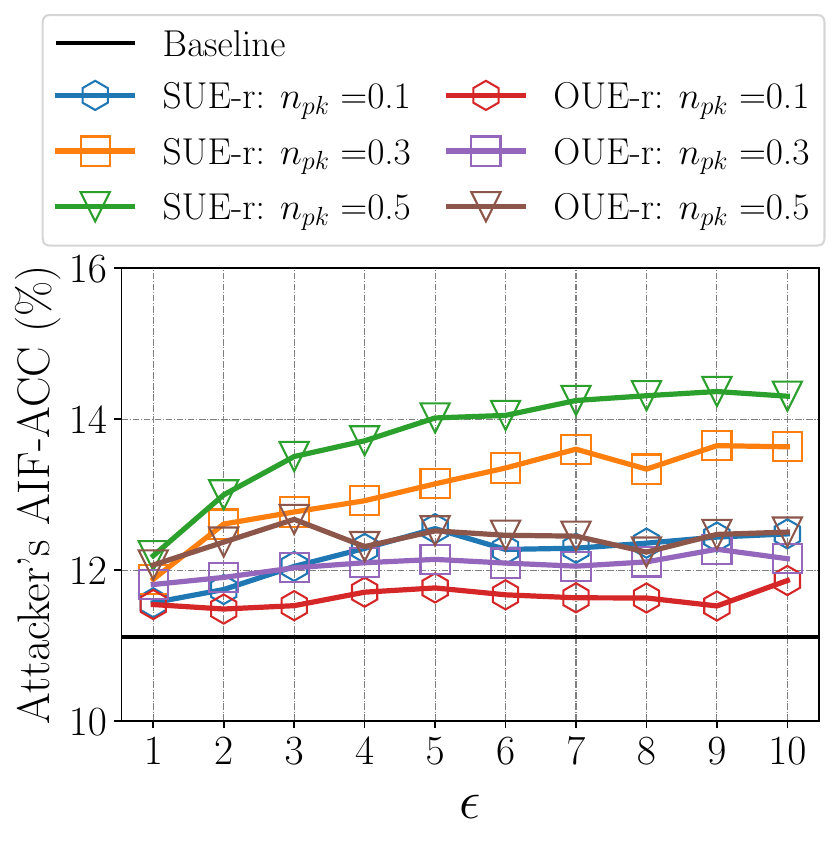}
  \caption{PK model with RS+FD[UE-r] protocols.}
\end{subfigure}
\begin{subfigure}{.33\textwidth}
  \centering
  \includegraphics[width=1\linewidth]{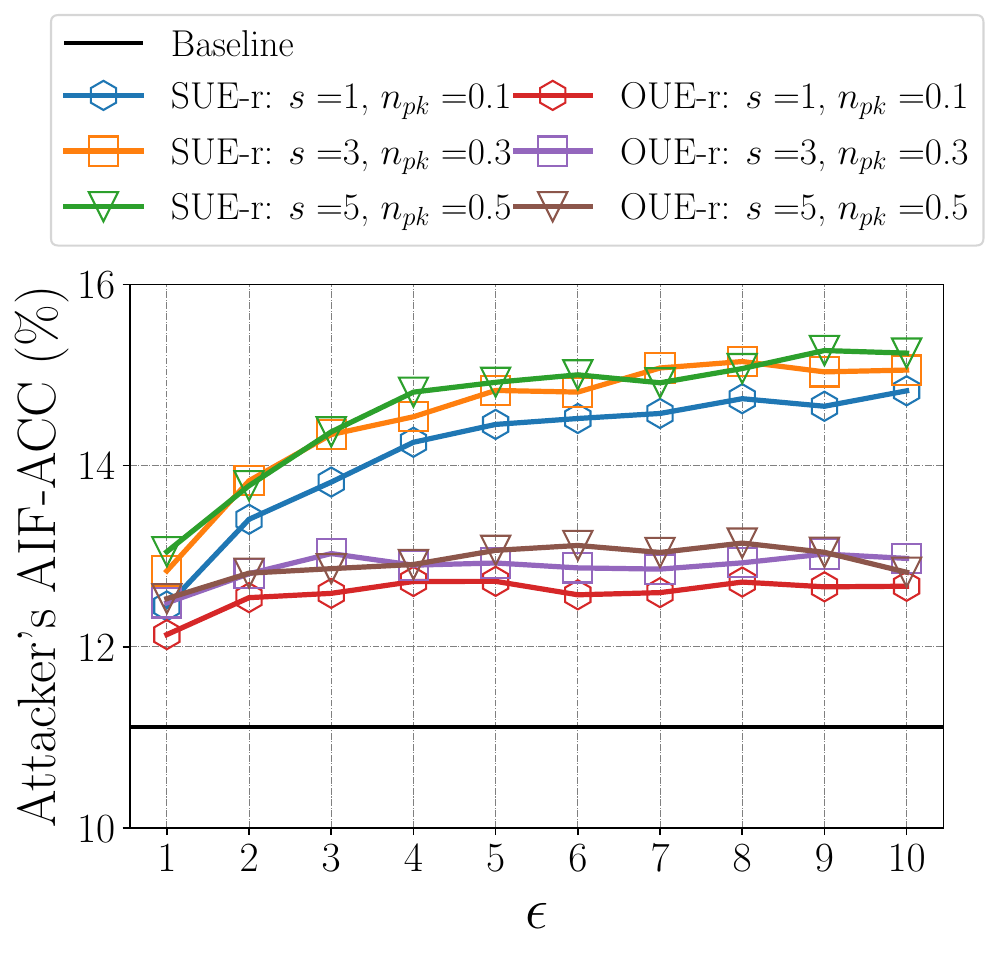}
  \caption{Hybrid model with RS+FD[UE-r] protocols.}
\end{subfigure}
\caption{Attacker's attribute inference accuracy (AIF-ACC) on the Nursery dataset with three attack models (i.e., NK, PK and hybrid) and five protocols (i.e., RS+FD[GRR], RS+FD[SUE-z], RS+FD[OUE-z], RS+FD[SUE-r] and RS+FD[OUE-r]), varying $\epsilon$, the number of synthetic profiles $s$ the attacker generates and the number of compromised profiles $n_{pk}$ the attacker has access to.}
\label{fig:attack_rspfd_nursery}
\end{figure*}

\section{Additional Results for Section~\ref{sub:results_rsprfd}} \label{appE:add_rsprfd}

This section provides additional results for the utility and privacy analysis of our proposed RS+RFD solution. 

\subsection{Multidimensional Frequency Estimation}
We compare the utility of our RS+RFD and RS+FD~\cite{Arcolezi2021_rs_fd} solutions for frequency estimation of multiple attributes following the experimental evaluation described in Section~\ref{sub:mult_freq_rsrfd} by using:

\begin{itemize}

    \item \textbf{Dataset.} We use the Adult~\cite{uci} ($d=10$ attributes, $n=45,222$ and $\textbf{k}=[74, 7, 16, 7, 14, 6, 5, 2, 41, 2]$).
    
    \item \textbf{LDP protocol within RS+FD.} RS+FD[GRR], RS+FD[SUE-r] and RS+FD[OUE-r], described in Section~\ref{sub:rspfd_sol}. 
    
    \item \textbf{LDP protocol within RS+RFD.} RS+RFD[GRR], RS+RFD[SUE-r] and RS+RFD[OUE-r], described in Section~\ref{sub:RSpRFD}. 
    
    \item \textbf{``Correct'' priors.} To simulate ``Correct'' prior distributions $\tilde{\textbf{f}}=[\tilde{f}_1, \tilde{f}_2, \ldots, \tilde{f}_d]$ to be used to generate realistic fake data with RS+RFD, we perturb the real frequency of each attribute $j\in[d]$ with the standard Laplace mechanism~\cite{Dwork2006,Dwork2006DP,dwork2014algorithmic} in centralized DP satisfying $\epsilon=0.1/d$ (\ie, split $\epsilon=0.1$ by $d$ attributes).
    
    \item \textbf{``Incorrect'' priors.} To simulate ``Incorrect'' scenarios where prior distributions are wrongly specified, we use the following distributions: Dirichlet distributions (DIR) with parameter $\mathbf{1}$, Zipf distributions (ZIPF) with parameter $s=1.01$ and Exponential distributions (EXP) with $\lambda=1$.
    For both ZIPF and EXP distributions, one hundred thousand samples are generated, and for each attribute $j\in[d]$, we reconstruct the histogram with $k_j$ buckets.
    
    \item \textbf{Analytical analysis.} We use the \textit{approximate} variance to measure the utility loss of our protocols by setting $f(v_i)=0$ in Eq.~\eqref{var:rs+rfd_grr} for RS+RFD[GRR] and in Eq.~\ref{var:rs+rfd_ue_r} for RS+RFD[UE-r].
    
    \item \textbf{Evaluation metrics.} To compare with~\cite{Arcolezi2021_rs_fd}, we vary $\epsilon$ in the interval $\epsilon=[\ln{(2)},\ln{(3)},\ldots,\ln{(7)}]$ and we measure the quality of the estimated frequencies with the averaged mean squared error metric: $MSE_{avg} = \frac{1}{d} \sum_{j \in [d]} \frac{1}{|A_j|} \sum_{v \in A_j}(f(v) - \hat{f}(v) )^2$.

\end{itemize}

Fig.~\ref{fig:mult_freq_results_adults} illustrates for all methods analytical (averaged approximate variance values) and experimental (averaged MSE metric) results varying $\epsilon$ for ``Correct'' (a and b plots) and ``Incorrect'' (c to h) priors for multidimensional frequency estimation with the RS+RFD and RS+FD solutions with the Adult dataset.

For all plots of Fig.~\ref{fig:mult_freq_results_adults}, one can notice that the experimental evaluations are consistent with the numerical evaluation of the variance for each protocol.
Moreover, for ``Correct'' priors, similar to Fig.~\ref{fig:mult_freq_results}, one can observe that the $MSE_{avg}$ metric of our proposed RS+RFD protocols consistently and considerably outperforms the utility of their respective version within the RS+FD solution. 
The intuition is that since random noise is drawn from realistic prior distributions, the fake data also contributes to the estimation of the attribute. 
With ``Incorrect'' priors (\eg, DIR), our RS+RFD protocols still outperforms the RS+FD protocols, with the exception of RS+RFD[OUE-r] with similar utility to RS+FD[OUE-r] in low privacy regimes (similar to Fig.~\ref{fig:mult_freq_results}).
Last, with both ``Incorrect'' priors following ZIPF and EXP distributions, the $MSE_{avg}$ (approximate variance, respectively) improved considerably in comparison with RS+FD protocols following uniform distribution, especially for our RS+RFD[GRR] protocol.

\begin{figure*}[!htb]
\begin{subfigure}{1\columnwidth}
  \centering
  \includegraphics[width=0.73\linewidth]{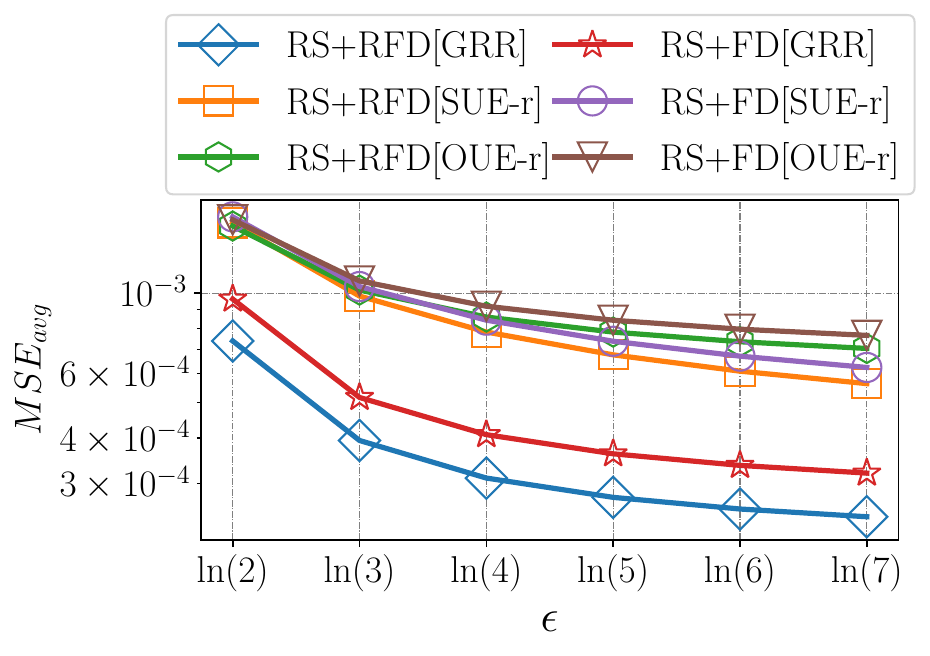}
  \caption{Analytical: ``Correct'' priors.}
\end{subfigure}%
\begin{subfigure}{1\columnwidth}
  \centering
  \includegraphics[width=0.73\linewidth]{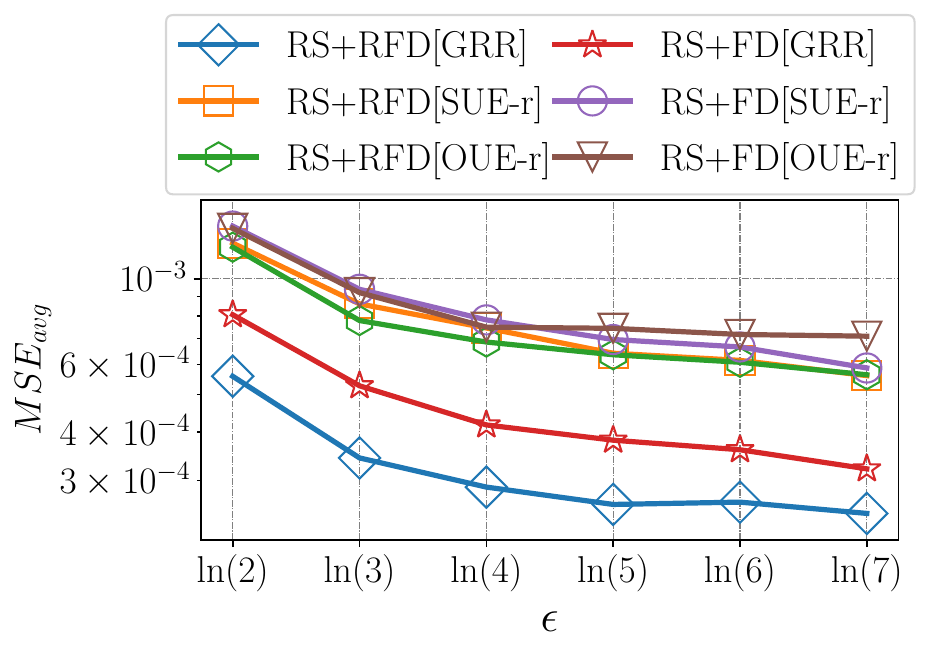}
  \caption{Experimental: ``Correct'' priors.}
\end{subfigure}\\
\begin{subfigure}{1\columnwidth}
  \centering
  \includegraphics[width=0.73\linewidth]{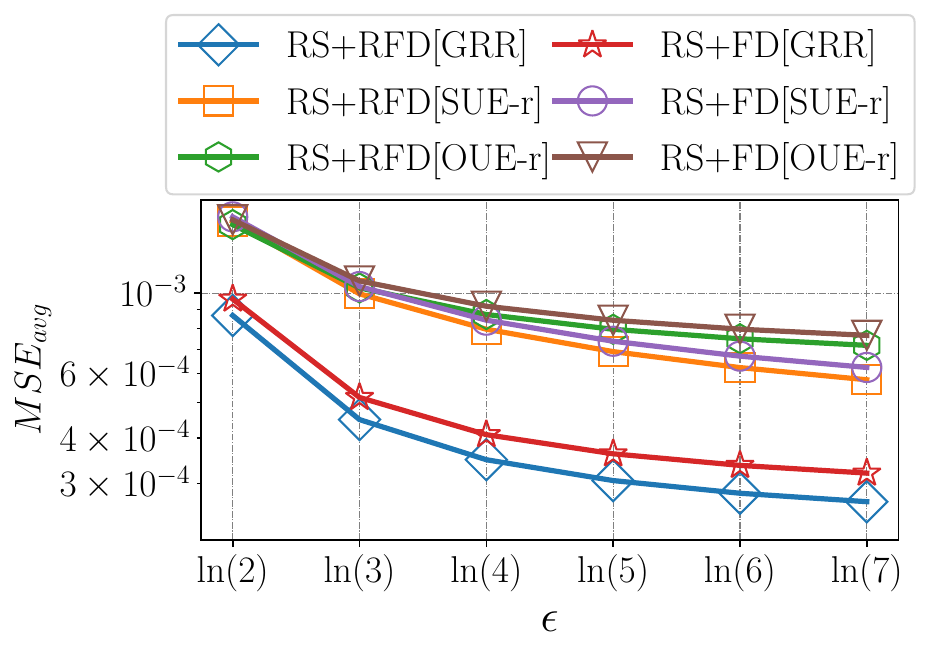}
  \caption{Analytical: ``Incorrect'' DIR priors.}
\end{subfigure}%
\begin{subfigure}{1\columnwidth}
  \centering
  \includegraphics[width=0.73\linewidth]{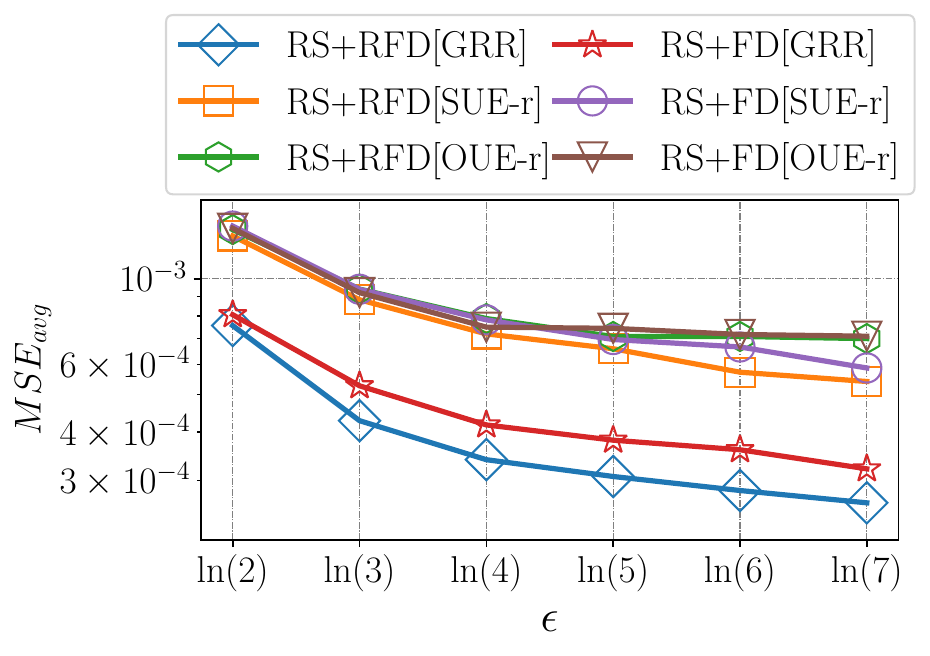}
  \caption{Experimental: ``Incorrect'' DIR priors.}
\end{subfigure}\\
\begin{subfigure}{1\columnwidth}
  \centering
  \includegraphics[width=0.73\linewidth]{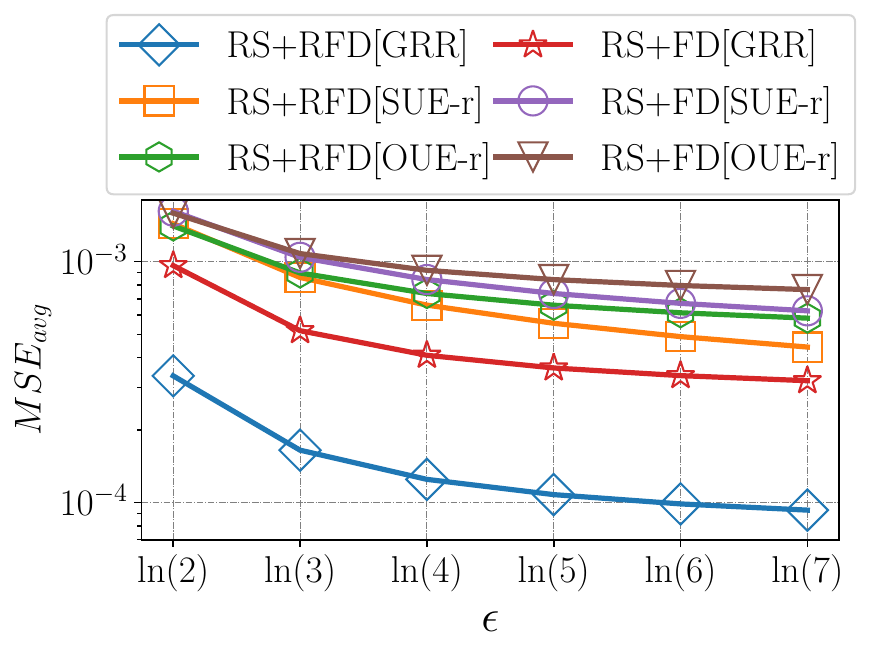}
  \caption{Analytical: ``Incorrect'' ZIPF priors.}
\end{subfigure}%
\begin{subfigure}{1\columnwidth}
  \centering
  \includegraphics[width=0.73\linewidth]{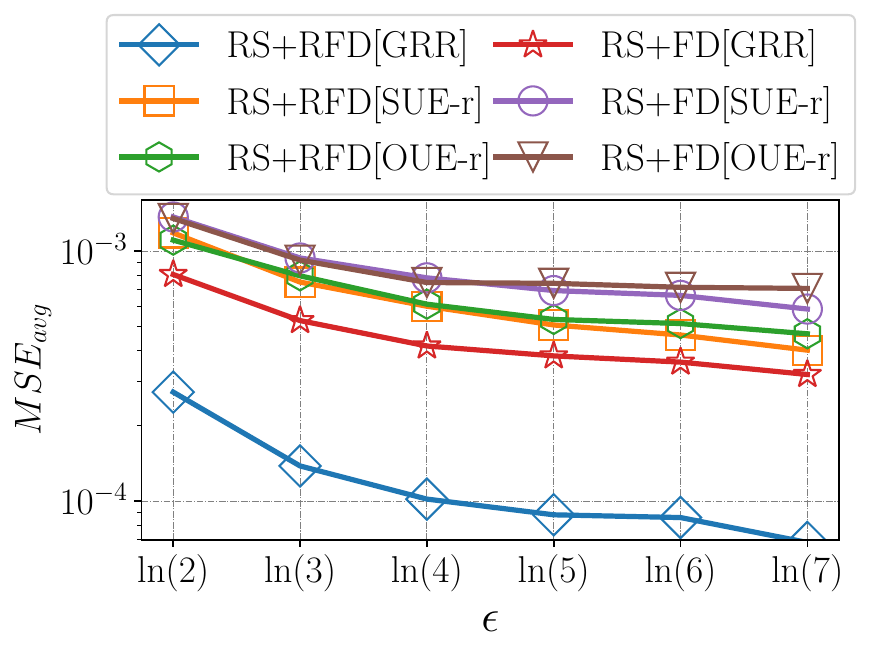}
  \caption{Experimental: ``Incorrect'' ZIPF priors.}
\end{subfigure}\\
\begin{subfigure}{1\columnwidth}
  \centering
  \includegraphics[width=0.73\linewidth]{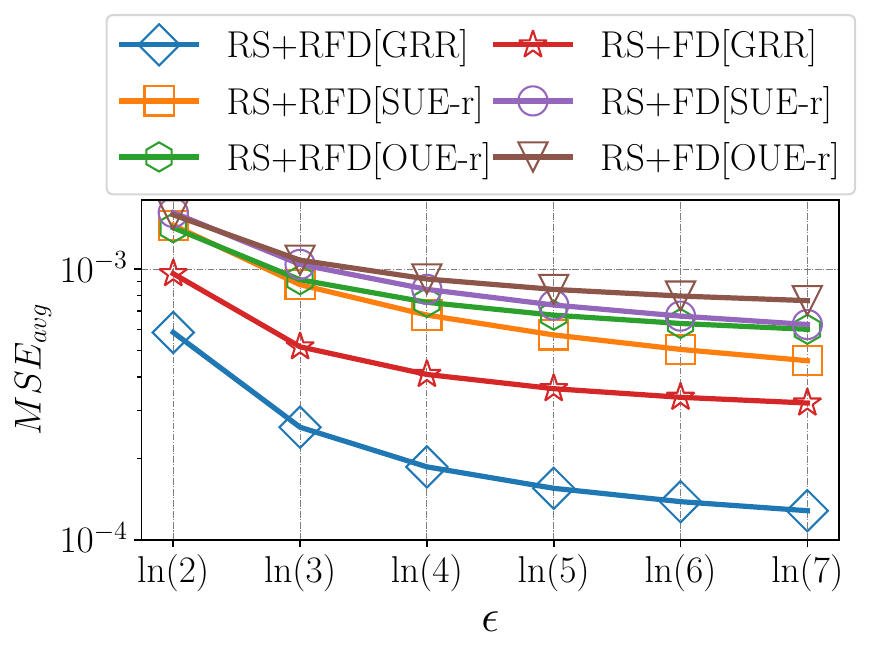}
  \caption{Analytical: ``Incorrect'' EXP priors.}
\end{subfigure}%
\begin{subfigure}{1\columnwidth}
  \centering
  \includegraphics[width=0.73\linewidth]{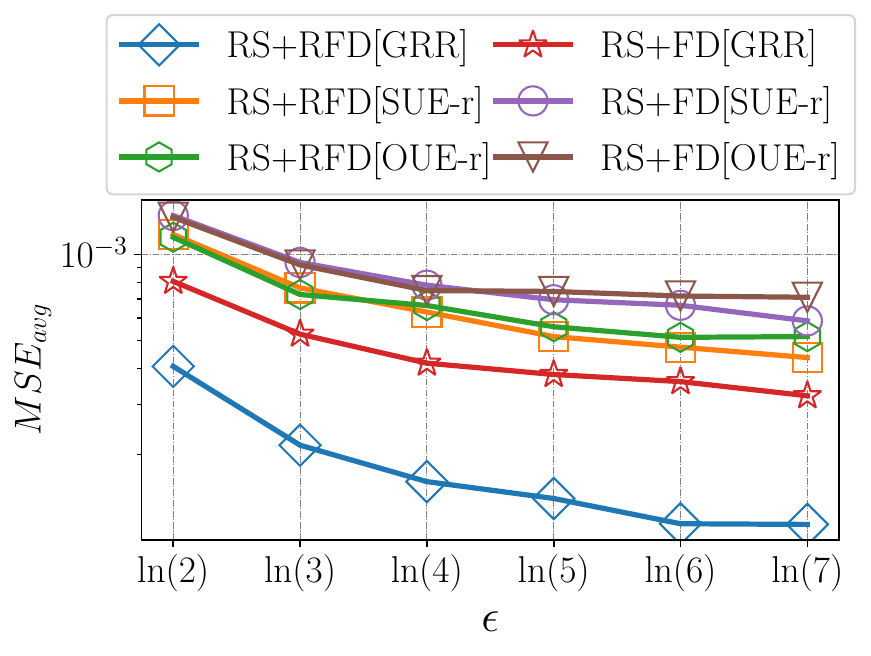}
  \caption{Experimental: ``Incorrect'' EXP priors.}
\end{subfigure}\\
\caption{Analytical (averaged approximate variance values) and experimental (averaged MSE metric) results varying $\epsilon$ for ``Correct'' (a and b plots) and ``Incorrect'' (c to h) priors for multidimensional frequency estimation with the RS+RFD and RS+FD solutions with the Adult dataset.}
\label{fig:mult_freq_results_adults}
\end{figure*}

\subsection{Inference of Sampled Attribute Attack}
Next, we compare the robustness of our RS+RFD protocols against the inference of the sampled attribute attack following the experimental evaluation described in Section~\ref{sub:att_inf_rsprfd} by using:

\begin{itemize}
    \item \textbf{Classifier.} We use the state-of-the-art XGBoost~\cite{XGBoost} algorithm to predict the sampled attribute of users in a multiclass classification framework (\ie, $d$ attributes) with default parameters.
    
    \item \textbf{Dataset.} We use the ACSEmployement~\cite{ding2021retiring} ($d=18$ attributes, $n=10,336$ and $\textbf{k}=[92, 25, 5, 2, 2, 9, 4, 5, 5, 4, 2, 18, 2, 2, 3, 9, 3, 6]$) dataset.
    
    \item \textbf{LDP protocol within RS+RFD.} RS+RFD[GRR], RS+RFD[SUE-r] and RS+RFD[OUE-r], described in Section~\ref{sub:RSpRFD}. 
    
    \item \textbf{``Correct'' priors.} To simulate ``Correct'' prior distributions $\tilde{\textbf{f}}=[\tilde{f}_1, \tilde{f}_2, \ldots, \tilde{f}_d]$ to be used to generate realistic fake data with RS+RFD, we perturb the real frequency of each attribute $j\in[d]$ with the standard Laplace mechanism~\cite{Dwork2006,Dwork2006DP,dwork2014algorithmic} in centralized DP satisfying $\epsilon=0.1/d$ (\ie, split $\epsilon=0.1$ by $d$ attributes).
    
    \item \textbf{``Incorrect'' priors.} To simulate ``Incorrect'' scenarios where prior distributions are wrongly specified, we use the following distributions: Dirichlet distributions (DIR) with parameter $\mathbf{1}$, Zipf distributions (ZIPF) with parameter $s=1.01$ and Exponential distributions (EXP) with $\lambda=1$.
    For both ZIPF and EXP distributions, one hundred thousand samples are generated, and for each attribute $j\in[d]$, we reconstruct the histogram with $k_j$ buckets.
    
    \item \textbf{Attribute inference model.} All three protocols are evaluated with the No Knowledge (NK) attack model of Section~\ref{sub:atk_models_rspfd}.

\end{itemize}

\begin{figure*}[!htb]
\begin{subfigure}{.33\textwidth}
  \centering
  \includegraphics[width=1\linewidth]{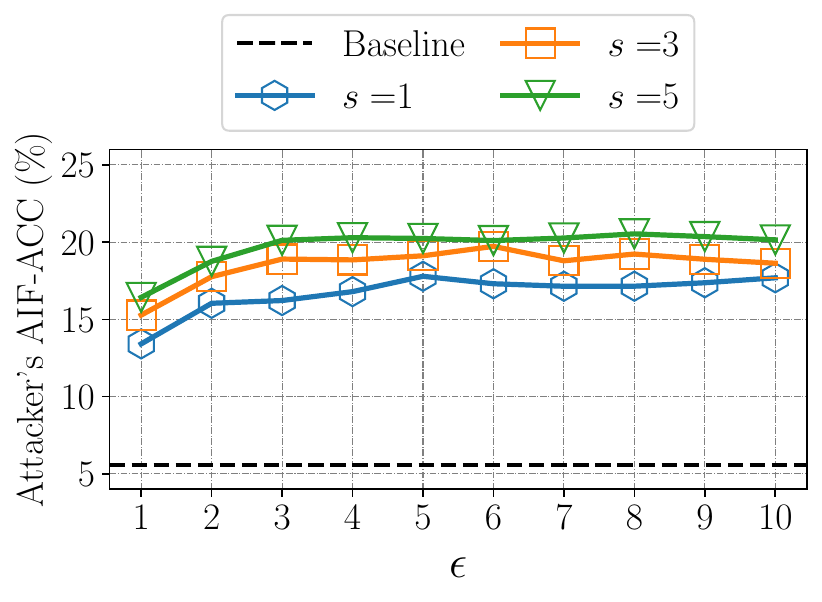}
  \caption{RS+RFD[GRR] with ``Incorrect'' DIR priors.}
\end{subfigure}%
\begin{subfigure}{.33\textwidth}
  \centering
  \includegraphics[width=1\linewidth]{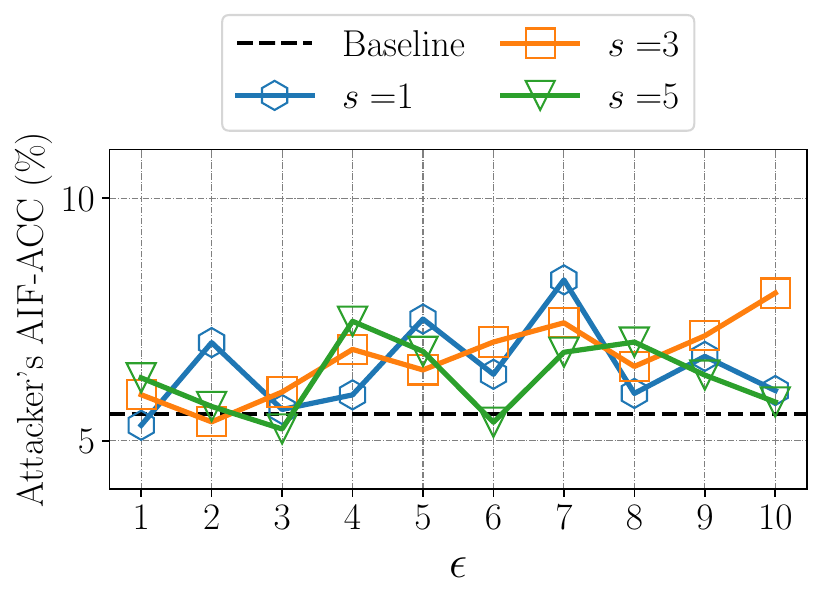}
  \caption{RS+RFD[GRR] with ``Incorrect'' ZIPF priors.}
\end{subfigure}
\begin{subfigure}{.33\textwidth}
  \centering
  \includegraphics[width=1\linewidth]{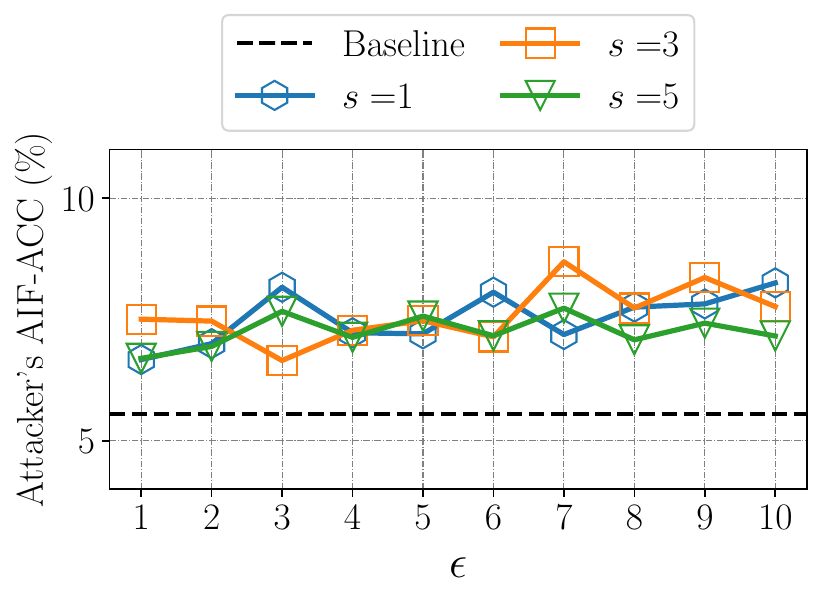}
  \caption{RS+RFD[GRR] with ``Incorrect'' EXP priors.}
\end{subfigure}
\\
\begin{subfigure}{.33\textwidth}
  \centering
  \includegraphics[width=1\linewidth]{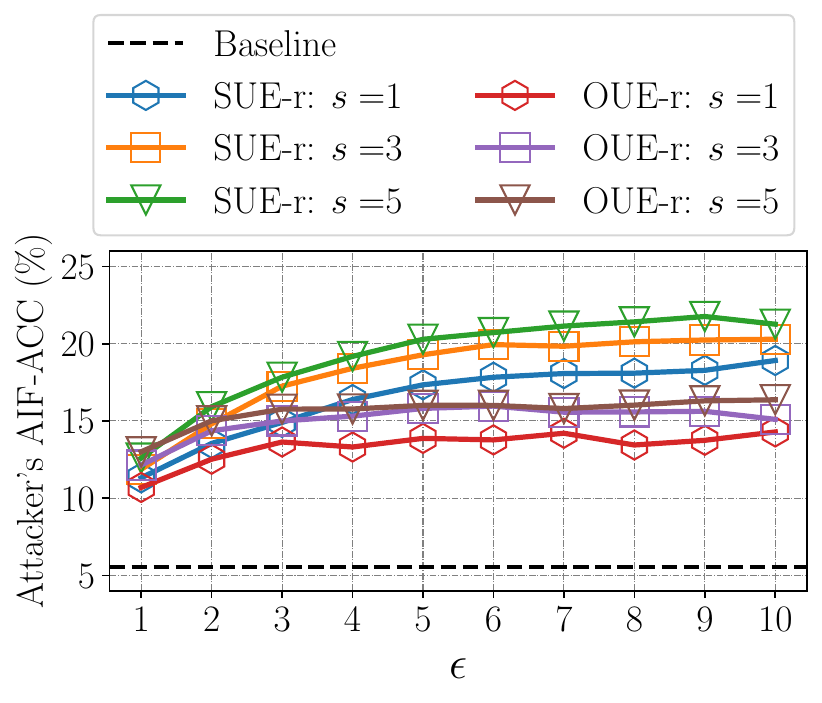}
  \caption{RS+RFD[UE-r] with ``Incorrect'' DIR priors.}
\end{subfigure}%
\begin{subfigure}{.33\textwidth}
  \centering
  \includegraphics[width=1\linewidth]{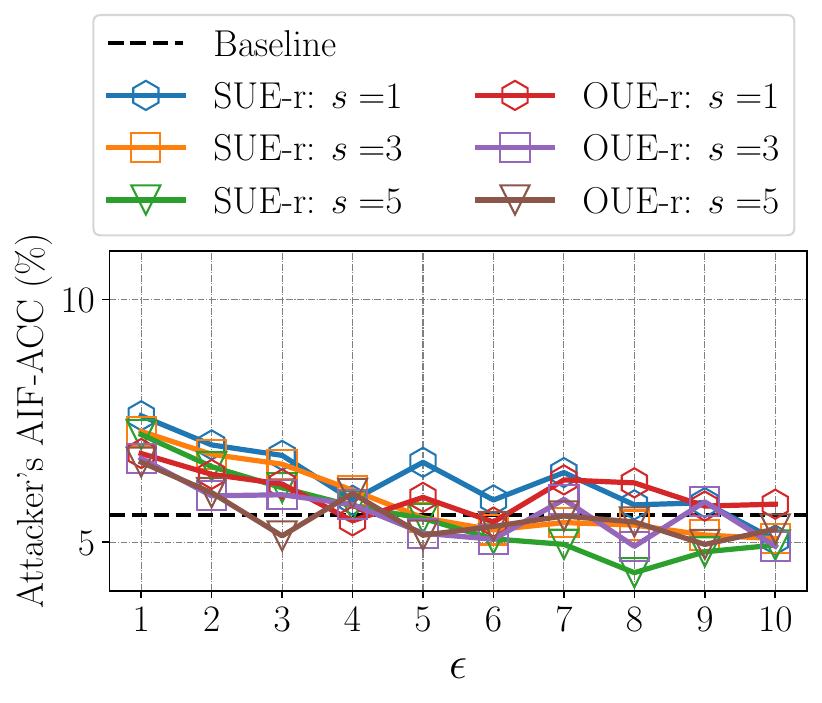}
  \caption{RS+RFD[UE-r] with ``Incorrect'' ZIPF priors.}
\end{subfigure}
\begin{subfigure}{.33\textwidth}
  \centering
  \includegraphics[width=1\linewidth]{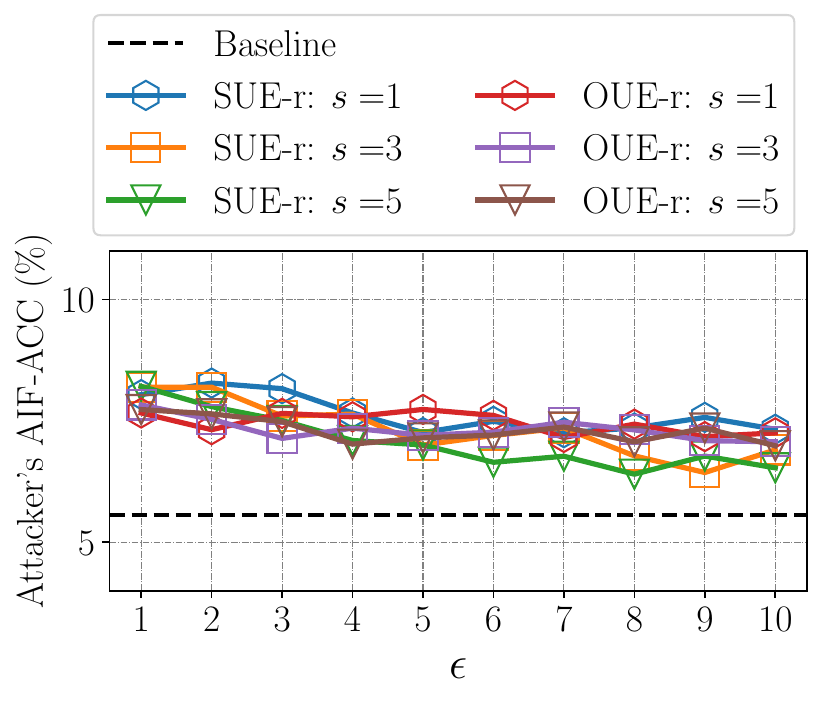}
  \caption{RS+RFD[UE-r] with ``Incorrect'' EXP priors.}
\end{subfigure}
\caption{Attacker's AIF-ACC on the ACSEmployement dataset with the No Knowledge (NK) attack model and our three protocols (\ie, RS+RFD[GRR], RS+RFD[SUE-r] and RS+RFD[OUE-r] with ``Incorrect'' priors), varying $\epsilon$, the number of synthetic profiles $s$ the attacker generates and the prior distribution (\ie, DIR, ZIPF and EXP).}
\label{fig:attack_rsprfd_inc_prior}
\end{figure*}

Fig.~\ref{fig:attack_rsprfd_inc_prior} illustrates the attacker's AIF-ACC metric on the ACSEmployement dataset with the NK attack model and our three protocols (\ie, RS+RFD[GRR], RS+RFD[SUE-r] and RS+RFD[OUE-r] with ``Incorrect'' priors), varying $\epsilon$, the number of synthetic profiles $s$ the attacker generates and the ``Incorrect'' prior distribution (\ie, DIR, ZIPF and EXP). 
We remark that the non-stability in the plots of Fig.~\ref{fig:attack_rsprfd_inc_prior} is due to different sources of randomness: different random ``Incorrect'' distributions $\mathbf{\tilde{f}}$, $\epsilon$-LDP randomization, fake data generation and the XGBoost algorithm. 

In Fig.~\ref{fig:attack_rsprfd_inc_prior}, for both ZIPF and EXP distributions, one can notice that our RS+RFD protocols considerably decrease the attacker's AIF-ACC when comparing with their respective RS+FD version in Fig.~\ref{fig:attack_rspfd}. 
More specifically, the RS+FD[GRR] and RS+FD[UE-r] protocols presented at most $\sim 25\%$ of AIF-ACC within the NK attack model (see Fig.~\ref{fig:attack_rspfd}).
On the other hand, our RS+RFD[GRR] and RS+RFD[UE-r] protocols with both ZIPF and EXP ``Incorrect'' priors presented at most $\sim 8\%$ of AIF-ACC (\cf{} Fig.~\ref{fig:attack_rsprfd_inc_prior}) within the same NK attack model, \ie, a 3x decrease factor.
Overall, the lowest performance was achieved when using DIR ``Incorrect'' priors, which showed small reductions compared to the NK curves in Fig.~\ref{fig:attack_rspfd} (\ie, RS+FD with prior uniforms).

\end{document}